\RequirePackage{fix-cm}
\documentclass[envcountsect,numbook,natbib]{svjour3_arxiv4}   
\smartqed  

\usepackage{graphicx}
\graphicspath{{./Figures/}}
\usepackage[utf8]{inputenc}
\usepackage{dsfont}
\usepackage{float}

\usepackage{color}
\usepackage{url}
\usepackage{enumerate} 
\usepackage{booktabs}
\usepackage{tabularx}
\usepackage{ltablex}
\usepackage{mathptmx}
\usepackage{bm}
\usepackage{amssymb,amsmath,amsfonts,latexsym,amsbsy} 
\usepackage{longtable} 
\usepackage{textcomp}
\usepackage{lscape}

\usepackage[english]{babel}

\usepackage[T1]{fontenc}
\usepackage{helvet}
\usepackage{etoolbox}

\usepackage{booktabs}
\usepackage{xcolor} 
\usepackage[colorlinks=black, citecolor=cyan]{hyperref}
\usepackage[font=footnotesize,labelfont=bf]{caption}

\usepackage{scrextend}
\usepackage{graphicx}
\usepackage{mathtools}
\mathtoolsset{showonlyrefs}

\usepackage{subfig}
\usepackage{calrsfs}
\DeclareMathAlphabet{\pazocal}{OMS}{zplm}{m}{n}
\let\mathcal\pazocal
\usepackage{enumitem}
\setlist[enumerate]{leftmargin=*}

\usepackage[12pt]{extsizes}
\usepackage{geometry}
\usepackage{ifthen}
\allowdisplaybreaks

\spnewtheorem{theorem}{Theorem}{\bfseries}{\itshape}
\spnewtheorem{corollary}[theorem]{Corollary}{\bfseries}{\itshape}
\spnewtheorem{lemma}[theorem]{Lemma}{\bfseries}{\itshape}
\spnewtheorem{proposition}[theorem]{Proposition}{\bfseries}{\itshape}
\spnewtheorem{definition}[theorem]{Definition}{\bfseries}{\itshape}
\spnewtheorem{remark}[theorem]{Remark}{\bfseries}{\upshape}
\spnewtheorem{assumption}[theorem]{Assumption}{\bfseries}{\itshape}
\counterwithin{theorem}{section}
\smartqed

\usepackage{tikz}

\def \N{\mathbb{N}}               
\def \P{\mathbb{P}}             
\def \1{{\bf 1}}                
\def \0{{\bf 0}}
\def\qed{\hfill$\Box$}

\definecolor{myred}{rgb}{0.9,0,0}  
\definecolor{mygreen}{rgb}{0,0.7,0}  
\definecolor{myblue}{rgb}{0.2,0,0.8}  
\definecolor{orange}{rgb}{1,0.6,0}  
\definecolor{olive}{rgb}{0.5,0.5,0}  
\definecolor{mylila}{rgb}{0.8,0.5,0.2}  
\definecolor{mygrey}{rgb}{0.6,0.6,0.6}  
\definecolor{mybrown}{rgb}{0.65,0.16,0.16}  
\definecolor{mymaroon}{rgb}{0.11,0.0,0.0}

\def \one{\mathcal{I}}

\newcommand{\condvar}{{Q}}  
\newcommand{\condmean}{M}  
\newcommand{\condvarEKF}{{\widetilde{Q}}}  
\newcommand{\condmeanEKF}{{\widetilde{M}}}  

\renewcommand{\P}{\mathbb{P}}

\newcommand{\TimCh}{\tau}  
\def\var{\operatorname{Var}}
\def\cov{\operatorname{Cov}}
\newcommand{\Var}{\var}
\newcommand{\E}{\mathbb{E}}

\newcommand{\Zpath}{{\mathcal{Z}}}
\newcommand{\Zpathn}{{\Zpath_n}} 
\def\cf{{ \color{black} \widetilde{f}_n}}
\def\ch{{ \color{black} \widetilde{h}_n}}
\def\csigma{{\widetilde{\sigma}_n}}
\def\cg{{\widetilde{g}_n}}
\def\cell{{\widetilde{\ell}_n}}
\def\cb{{\widetilde{b}_n}}

\def\dfcd{{U_{\dagger}}}
\def\dfc{{U}}
\def\ecc{{v}} 
\def\dfcmean{{\condmean}_{\dfc}}
\def\dfcdmean{{\condmean}_{\dfcd}}

\def\dfcdvar{{\condvar}_{\dfcd}}
\def\Fprior{{\mathcal{F}^I_0}}
\def\eccmean{{\condmean}_{\ecc}}
\def\eccvar{{\condvar}_{\ecc}}
\def\L{L}

\def\LR{\L^R}
\def\K{d}  

\def\KR{\K^R}
\def\PP{P}

\def\PR{\PP^R}

\def\Rin{R^\text{in}}
\def\Rout{R^\text{out}}
\def\Count{\Theta}
\def\halfsat{\phi}
\def\Isum{\mathcal{I}}
\def\propexposed{\psi}

\newcommand{\mymarginpar}[1]{ \marginpar{{\tiny #1}}}

\renewcommand{\mymarginpar}[1]{}

\date{}    
\usepackage{algorithm}
\usepackage{algpseudocode}
\usepackage[algo2e]{algorithm2e} 
\begin{document}

\title{Stochastic Models and Estimation of Undetected Infections in the Transmission of Zika Virus}
\titlerunning{Stochastic Models and Estimation of Undetected Infections}        
\author{Lillian Achola  Oluoch  \and Florent Ouabo Kamkumo \and Ralf Wunderlich}
\authorrunning{L.~Achola  Oluoch, F.~Ouabo Kamkumo, R.~Wunderlich} 
\institute{
Lillian Achola  Oluoch  \at The Technical University of Kenya; { Haile Selassie Avenue, P.O. Box 52428 - 00200, Nairobi, Kenya.}
\email{lilianoluoch@tukenya.ac.ke}
\and
Florent.OuaboKamkumo / Ralf Wunderlich \at
Brandenburg University of Technology Cottbus-Senftenberg, Institute of Mathematics, P.O. Box 101344, 03013 Cottbus, Germany;  
\email{\texttt{Florent.OuaboKamkumo / ralf.wunderlich@b-tu.de}}  
}

	\date{Version of  \today}
	
	\maketitle

\begin{abstract}
Zika fever, a mosquito-borne viral disease with potential severe neurological complications and birth defects, remains a significant public health concern. The epidemiological models often oversimplify the dynamics of Zika transmission by assuming immediate detection of all infected cases. This study provides an enhanced SEIR (Susceptible-Exposed-Infectious-Recovered) model to incorporate partial information by distinguishing between detected and undetected Zika infections (also known as “dark figures”). By distinguishing the compartments, the model captures the complexities of disease spread by accounting for uncertainties about transmission and the number of undetected infections.
This model implements the Kalman filter technique to estimate the hidden states from the observed states. Numerical simulations were performed to understand the dynamics of Zika transmission and real-world data was utilized for parameterization and validation of the model. The study aims to provide information on the impact of undetected Zika infections on disease spread within the population, which will contribute to evidence-based decision making in public health policy and practice.
\end{abstract}


\newpage
	
	\setcounter{tocdepth}{2}
	\tableofcontents
	
	\newpage

\section { Introduction}

Understanding the dynamics of transmission of infectious diseases is fundamental to designing effective public health strategies. Stochastic models, which incorporate random variables to account for uncertainties, have proven essential for predicting the course of disease outbreaks. The Zika virus  disease (ZIKV) is  a vector-borne disease transmitted by Aedes mosquitoes that gained worldwide attention during the major outbreak in the Americas between 2015 and 2016. Although often asymptomatic or mildly symptomatic in many individuals, ZIKV has been linked to severe consequences, including congenital microcephaly and neurological disorders, making its control a priority in global health efforts \cite{Baca,Blohm}.

The complexity of modeling ZIKV dynamics increases due to the presence of undetected cases: individuals who remain unreported because they are asymptomatic or because their symptoms are mild and do not lead them to seek medical attention \cite{Bra, BrittonPardoux2019}. These undetected cases are critical, as they introduce substantial uncertainty into the understanding of disease prevalence and transmission. This unreported component of an epidemic, often referred to as ``dark figures'', can cause a significant underestimation of the true scope of an outbreak and lead to misinformed public health interventions.

Furthermore, the challenge of tracking and managing ZIKV transmission is amplified by the partial information available to researchers and health authorities. While some epidemiological data, such as the number of symptomatic reported cases or vector population changes, are readily observable, other key aspects of the epidemic remain hidden. This includes the number of asymptomatic individuals or the extent of underreporting among those with mild symptoms \cite{CDC}. As a result, modeling ZIKV dynamics with only partial information becomes inherently uncertain, complicating the accurate estimation of vital epidemic parameters like the effective reproduction number  and the total number of infections \cite{CDC1,anderson2011continuous}).

To address these uncertainties, it is crucial to adopt stochastic epidemic models that incorporate both the observable and unobservable elements of transmission. Models that integrate the effects of undetected cases allow for more reliable predictions and can inform more effective control strategies. This is particularly relevant for vector-borne diseases like ZIKV, where the transmission process is influenced not only by human hosts but also by environmental conditions and vector biology \cite{Davi}. Accounting for hidden infections, whether asymptomatic or underreported symptomatic cases, can lead to better estimates of disease spread and improve response strategies. Such stochastic models can support decision making in epidemic control, such as evaluating and adjusting intervention strategies, such as vector control, based on updated data.

In this study, we develop a stochastic epidemic model for ZIKV that explicitly accounts for the partial information regarding undetected cases. These include both asymptomatic individuals and those with mild symptoms who do not report their illness. We explore how these unobserved cases affect the overall transmission dynamics and propose methods for improving the estimation of the true burden of the epidemic. By integrating stochastic control approaches, our model will allow for more effective public health interventions in the context of ZIKV and similar infectious diseases as in  \cite{Njiasse}.

\paragraph{Literature Review on Zika Disease}
Zika virus which causes Zika fever is a \textit{Flaviviridae} virus predominantly transmitted by female mosquito bites from  \textit{ Aedes aegypti} and \textit{Aedes albopictus} \cite{Peter}. It poses a serious public health hazard due to its potential to cause severe birth abnormalities (such as microcephaly) and neurological disorders. Other arboviruses transmitted by mosquitoes include Chikungunya and Dengue. While these diseases are primarily transmitted by mosquitoes, research indicates that Zika, unlike other insect-transmitted infections, can also be transmitted through sexual contact, primarily from men to women \cite{Maga}.

In human-to-human transmission, Zika can spread before, during, and after symptoms appear. There is also a high probability of a mother transmitting the disease to her fetus during pregnancy, delivery, or by blood transfusion and breastfeeding \cite{Blohm, Greg}. Asymptomatic carriers may also transmit the virus \cite{CDC}. Studies reveal that the Zika virus can persist in semen for extended periods, possibly up to six months \cite{Mead}, meaning that even months after recovery, the disease might still be transmissible.
The incubation period of the disease is approximately 3-14 days. Infected individuals typically experience no symptoms or moderate symptoms such as muscle and joint soreness, fever, rash, and headaches \cite{WHO1}. Due to these reasons, many infected persons may go undetected, posing a significant risk in transmission of the virus.

Zika virus continues to spread globally, posing a significant hazard to public health as there are currently no identified vaccinations or therapies for treatment and prevention \cite{Song}. Consequently, Zika fever is one of the priority diseases in the WHO's Blueprint for Research and Development \cite{WHO2}, due to its likely persistence as a significant threat in the future.
Therefore, understanding the mechanisms of Zika transmission and the effectiveness of intervention efforts is crucial in mitigating its impact. In many scenarios, a portion of infected individuals may go undetected through testing or surveillance systems, leading to undetected infections. This study thus proposes a classic epidemiological model that can be used to understand the spread of infectious diseases within a population, particularly considering undetected cases. Understanding the dynamics of undetected infections is crucial for effective disease control and public health interventions.

Numerous mathematical models have been developed to predict the spread of Zika virus disease \cite{Pad, Bra}. A stage-structured model was developed to examine the impact of sexual transmission \cite{Sas}. Another compartmental model of Zika propagation, considering vector-borne and sexual transmission, was proposed by \cite{Gao}, which used an SEIR-type model for humans and an SEI-chain for vectors, distinguishing between asymptomatic and symptomatic infected individuals. Additionally, a compartmental model accounting for vector-borne and sexual transmission, differentiating between sexes and migration, was developed by \cite{Baca}. Transmission of Zika fever using two vector management strategies, lowering mosquito biting rates and population size, was proposed by \cite{Supa}.
In \cite{Denes}, we developed a non-autonomous model that incorporated many important features of Zika transmission, including weather seasonality, sexual and vector-borne transmission, the prolonged period of infectiousness after recovery, and the role of asymptomatically infected humans. However, in the study, we did not address the fact that within the susceptible and exposed population,  there may be those who were infected but undetected. 

No  Zika model has addressed the fact that not all state variables can be directly observed and others go undetected, even though the inclusion of non-detected compartments has been proposed \cite{Kamkumo}. In the study, they used the stochastic epidemic model to explore state variables that could not be directly observed and termed them as the ``dark figures'' problem. The application of the stochastic model in epidemics is not novel. Stochastic epidemic modeling, with a focus on compartmental models, has been used to determine the distribution of ultimate epidemic size, the impact of various infectiousness patterns, and the quantification of stochastically maintained oscillations \cite {green}.
Hence, in the present work, based on our model in \cite{Denes}, we propose a stochastic epidemic model by incorporating partial information on Zika infections, distinguishing between detected and undetected cases. 

\paragraph{Literature Review on General Filtering Methods}
In the context of stochastic epidemic models with partial observability, filtering techniques serve as indispensable tools for inferring unobservable states and estimating unknown parameters. These models typically classify system components into \emph{observable states}, which can be measured or reported, and \emph{latent states}, which remain hidden from direct observation. A variety of filtering methods have been employed in the literature to recover these hidden dynamics and calibrate model parameters using noisy, incomplete, or indirect observations.

A significant challenge in epidemic modeling is the scarcity or complete absence of data regarding the true infection process. This problem is addressed using a Bayesian framework  \cite{oneill}, by implementing Markov chain Monte Carlo (MCMC) techniques to jointly estimate missing data and infer model parameters. Their approach facilitates inference even when key epidemiological events, such as infection times, are unobserved.

Extending this methodology, \cite{britton} introduced a model in which the population structure is described by a random graph, resembling an SIR-type process. Their work integrates MCMC methods for simultaneous inference of disease transmission rates and characteristics of the social contact network, allowing a coherent estimation of both epidemic and structural parameters.

Further developments include state-space models incorporating modified SEIR dynamics. \cite{calvetti} proposed a Bayesian particle filtering algorithm for estimating the temporal evolution of both hidden states and unknown parameters from noisy daily infection data related to the COVID-19 pandemic. Their method enables robust estimation even in the presence of measurement noise and under-reporting.

Similarly, \cite{lal} employed an ensemble Kalman filter (EnKF) within a SIRD modeling framework to describe COVID-19 progression. The EnKF is particularly well-suited for high-dimensional or nonlinear systems, offering an efficient recursive estimation of both model parameters and unobserved epidemiological compartments.

A discrete-time stochastic SIR model with latent infectious states was considered by \cite{colaneri}, where the transmission rate and the actual number of infectious individuals were assumed to be unobservable. Their approach utilizes a hidden Markov model (HMM) formulation, with nested particle filtering techniques employed to recover the hidden states and estimate the reproduction number and other model parameters.

To address the presence of non-Gaussian features such as skewness and outliers in epidemic data, \cite{alyami} introduced a skew Kalman filter (SKF). Traditional Kalman filters often assume Gaussianity and can be sensitive to extreme values. The SKF generalizes the Kalman framework to accommodate asymmetric distributions, leading to more robust Bayesian inference, particularly for state estimation in the presence of irregular or noisy data.

In general, Bayesian filtering frameworks, including MCMC and sequential Monte Carlo (SMC) methods like particle filters, offer powerful probabilistic tools for estimation under uncertainty \cite{cappe,doucet,sarkka}. These methods approximate posterior distributions of hidden variables through sampling, making them flexible and suitable for a wide range of nonlinear and non-Gaussian models. However, they often come with a computational cost: convergence can be slow and the algorithms may be computationally intensive, especially when applied to complex models or large-scale datasets.

\paragraph{Literature Review on Kalman Filtering in Epidemic Modeling}
Increasing research has used Kalman filtering techniques to improve state space estimation in compartmental epidemic models. A stochastic SEIR (R) DSD model has been introduced by \cite{Zhu}.  An Extended Kalman Filter (EKF) was employed to infer both model parameters and unobserved states, enhancing short‐term forecasting accuracy.
 \cite{Hasan} follow a comparable route, employing an EKF within a classical SIRD model to dynamically track disease progression, enabling refined temporal parameter estimation.

On the other hand ,\cite{Zeng} implemented a switching Kalman Filter based on linear Gaussian assumptions. However, such flexibility comes at the cost of increased computational demand and challenges in calibrating transition probabilities among regimes. In an alternative direction,  \cite{Njiasse} apply EKF within a stochastic optimal‐control framework. Under partial information, their methodology addresses decision‐making problems faced by policymakers aiming to contain an epidemic efficiently while balancing public health and economic costs.
In a related study, \cite {Kamkumo} estimated the unobservable component from the observations using the extended Kalman filter approach to account for nonlinearity of the state dynamics.

Further expanding the filtering toolbox, \cite{Chenli} explore conditional nonlinear Gaussian systems (CNGS). These systems support fast joint estimation of both latent states and model parameters in nonlinear settings under partial observability. A major advantage lies in the availability of closed‐form expressions for conditional distributions, enabling efficient data assimilation and uncertainty quantification.
Building on this, \cite{Chenma} highlighted the suitability of conditional Gaussian frameworks for multiscale stochastic systems. Their work emphasizes the capacity of such models to capture non-Gaussian, non-linear dynamics characteristic of epidemics and other natural phenomena.

Several studies have successfully applied ensemble Kalman filters to dengue transmission models, which is another mosquito-based disease such as the Zika virus. For example, \cite{Zhou} utilized an SIR-type compartmental model coupled with the Ensemble Adjustment Kalman Filter (EAKF) to forecast the weekly incidence of dengue in Guangzhou, China, over the 2011-2017 period. Their model incorporated mosquito density and meteorological variables (such as temperature and rainfall) as covariates. The EAKF approach allowed for dynamic estimation of effective transmission rate, population susceptibility, and peak timing, generating accurate short-term outbreak predictions.

In a related effort, \cite{Yang} developed a metapopulation network model, also driven by EAKF, to predict dengue spread in cities in Guangdong province. By assimilating case reports and climatic data across multiple nodes, the model improved forecasts of spatial and temporal disease dynamics, outperforming single-city models in epidemic peak prediction up to 10 weeks in advance.
Furthermore, a comparative analysis by \cite{Indriani} evaluated both EnKF and EKF methods within a SIRS model framework. They reported high forecast accuracy, with the EKF slightly outperforming the EnKF in terms of computational efficiency and convergence speed.

\paragraph{Our Contribution}
This study provides a novel contribution to the field of infectious disease modeling by introducing, for the first time, an extended Kalman filter framework for tracking Zika virus transmission under partial observability. Unlike existing approaches, which predominantly rely on deterministic or simple stochastic compartmental models, the proposed methodology is rooted in stochastic epidemic models derived from large-population limits of continuous-time Markov chains. These models are specifically designed to incorporate unobserved epidemic states, enabling the systematic estimation of hidden infections. A central innovation is the formulation of a cascade of partially hidden compartments, where either the inflow or outflow is observed, but not both. This allows the model to leverage all available data without overestimating unobservable quantities. To capture nonlinear dynamics inherent in vector-borne disease transmission, we apply the extended Kalman filter (EKF), which facilitates real-time inference of latent states. Through comprehensive simulation studies based on a Zika transmission model, we demonstrate that although initial uncertainty in estimates is high, the EKF quickly assimilates information and accurately tracks the evolution of hidden states. These results highlight the potential of Kalman filtering techniques in improving public health decision making in the presence of incomplete data and under-reporting, particularly in diseases such as Zika, where traditional models fall short.

\paragraph{Notation} 
Let $x \in \mathbb{R}^d$ be a vector with components denoted by $x^1, \ldots, x^d$, and let $\|x\|$ represent its Euclidean norm. For a matrix $A \in \mathbb{R}^{d \times d}$, the entry in the $i$-th row and $j$-th column is written as $A^{ij}$, and $\|A\|$ denotes the Frobenius norm. The $d \times d$ identity matrix is denoted by $\mathbb{I}_d$, and the zero vector in $\mathbb{R}^d$ is denoted by $0_d$.

\paragraph{Paper organization}
The rest of this paper,is organized as follows. Section \ref{sec:Stochastic} presents the stochastic epidemic modeling framework underpinning the analysis. In Section~3, we develop the Zika virus transmission model, beginning with a compartmental representation, followed by a simplified version, a base model, and finally an extended model capturing more detailed dynamics. Section~4 addresses the estimation of unobservable states using Kalman filtering techniques for conditionally Gaussian state-space models. Section~5 discusses the numerical results, including parameter specification and initialization procedures for the extended model and Kalman filter. An appendix contains a  list of  notations and provides technical details and proofs that were removed from the main text.

\section{Stochastic Epidemic Modeling Framework}
\label{sec:Stochastic}

Following the framework introduced in \cite{Kamkumo2}, we briefly recapitulate the core components of the stochastic compartmental model relevant to the current work. We consider a structured population model in which a closed population of constant size $N \in \mathbb{N}$ is subdivided into $d \in \mathbb{N}$ compartments, each representing a subpopulation with a distinct epidemiological status. The model incorporates $K \in \mathbb{N}$ distinct transition mechanisms between these compartments.

Let $T > 0$ denote a fixed time horizon. The stochastic processes describing the model are defined on a complete filtered probability space $(\Omega, \mathcal{F}, \mathbb{F}, \mathbb{P})$, where the filtration $\mathbb{F} = (\mathcal{F}_t)_{t \in [0,T]}$ satisfies the usual conditions of right-continuity and completeness. Let $X = (X(t))_{t \in [0,T]}$ denote a $\mathbb{Z}^d$-valued  process, where each component $X^i(t) \in \{0, \ldots, N\}$ represents the number of individuals in compartment $i \in \{1, \ldots, d\}$ at time $t$. The filtration $\mathbb{F}$ is assumed to be generated by $X$, i.e., $\mathcal{F}_t = \sigma(X(s): 0 \leq s \leq t)$.

\paragraph{Microscopic Dynamics}

Let $\Count_k(t)$ denote the cumulative number of transitions of type $k \in \{1,\ldots,K\}$ that have occurred in the interval $[0, t]$. The evolution of the process $X(t)$ can then be expressed as
\begin{align} \label{stateX}
    X(t) = X(0) + \sum_{k=1}^{K} \xi_k \, \Count_k(t),
\end{align}
where each $\xi_k \in \mathbb{Z}^d$ is a transition vector encoding the change in state caused by a single transition of type $k$. Specifically, $\xi_k^i = +1$ if the transition causes an inflow into compartment $i$, $\xi_k^i = -1$ if it causes an outflow from compartment $i$, and $\xi_k^i = 0$ otherwise. This formulation presumes that at most one individual is transferred between compartments during a single transition event.

A standard assumption in many compartmental epidemic models is that the transitions are independent and memoryless. Accordingly, we model the transition counting processes $\Count_1, \ldots, \Count_K$ as mutually independent, time-inhomogeneous Poisson processes with intensity functions $\lambda_k(t, X(t))$, which depend on the current time and system state.
To express these processes in terms of homogeneous Poisson processes, let $\Pi_k(t)$ denote a standard unit-rate Poisson process for each $k = 1, \ldots, K$. Then, define the stochastic time-change
\[
\TimCh_k(t) := \int_0^t \lambda_k(s, X(s))\,ds.
\]
The non-homogeneous counting process can then be represented as
$\Count_k(t) = \Pi_k(\TimCh_k(t))$. 
Substituting into \eqref{stateX}, the state dynamics of the system can be rewritten as
\begin{align} \label{stateX_Poisson}
	X(t) = X(0) + \sum_{k=1}^{K} \xi_k \, \Pi_k\left( \int_0^t \lambda_k(s, X(s))\,ds \right), \quad { X(0) = x_0.}
\end{align}

This stochastic representation defines a continuous-time Markov chain (CTMC) over the discrete state space $\{0, \ldots, N\}^d$, with time-dependent transition intensities governed by the system’s current state.

\paragraph{Macroscopic  Dynamics  -- Large Population Limits}

For large population sizes $N$, the stochastic dynamics of the continuous-time Markov chain (CTMC) described in \eqref{stateX_Poisson} can be approximated by a continuous-state diffusion process.  This approach is based on the study of the asymptotic behavior of a properly rescaled CTMC  for $N\to\infty$  and the application of a functional  law of large numbers and  central limit theorem, see Britton and Pardoux ~\cite[Chapter 2, Section 2.2-2.3]{BrittonPardoux2019}, Anderson and Kurtz \cite[Chapter 1, Section 3.2]{anderson2011continuous}, Ethier and Kurtz \cite[Chapter 4, Section 7]{ethier2009markov},  Guy et al.~\cite{guy2015approximation}, and \cite[Section 3.4]{Kamkumo2}. 
There it is shown that under suitable regularity conditions, the CTMC process converges weakly, as $N \to \infty$, to a continuous-state diffusion process $X^D = (X^D(t))_{t \in [0 , T]}$ governed by the following  system of stochastic differential equations (SDEs) 
\begin{equation} \label{DA1}
	dX^D(t) = f_X(t, X^D(t))\,dt + \sigma_X(t, X^D(t))\,dW(t), \quad X^D(0) = x_0,
\end{equation}
where $W(t)$ is a $K$-dimensional standard Brownian motion. The drift and diffusion terms are given by
\begin{align} \label{DA_drift_diffusion}
	\begin{split}
		f_X(t, x) &= \sum_{k=1}^{K} \xi_k\, \lambda_k(t, x), \quad \text{and}\quad 
		\sigma_X(t, x) = \left( \xi_1 \sqrt{\lambda_1(t, x)},\, \ldots,\, \xi_K \sqrt{\lambda_K(t, x)} \right).
	\end{split}
\end{align}

This limiting system captures the macroscopic (aggregate) behavior of the epidemic under stochastic fluctuations and is commonly referred to as the diffusion approximation of the underlying CTMC.

In practical applications such as statistical inference or real-time filtering, it is often preferable to work with a time-discretized version of the dynamics, since observations (e.g., incidence data, case reports) are typically available at discrete time intervals (e.g., daily or weekly). To this end, we partition the time interval $[0,T]$ into $N_t \in \mathbb{N}$ equally spaced intervals of length $\Delta t = T / N_t$, and define the discrete time points $t_n = n \Delta t$ for $n = 0, \ldots, N_t$.Applying the Euler–Maruyama discretization scheme to the diffusion system \eqref{DA1} yields the following stochastic recursion with the initial value $X^D_0 = x_0$:
\begin{align} \label{Discrete_Dyn1}
	X^D_{n+1} = X^D_n + f_X(t_n, X^D_n)\, \Delta t + \sigma_X(t_n, X^D_n)\, \sqrt{\Delta t}\, \mathcal{E}_{n+1},
\end{align}
where $X^D_n \approx X^D(t_n)$ denotes the discrete-time approximation of the state at time $t_n$, and $\{\mathcal{E}\}_{n=1}^{N_t}$ is a sequence of independent standard normally distributed vectors in $\mathbb{R}^K$, i.e., $\mathcal{E}_n \sim \mathcal{N}(0_K, \mathbb{I}_K)$.

This discrete-time representation will serve as the foundation for the development of inference procedures under partial observation in subsequent sections.

\paragraph{Models with Hidden States}
\label{sec:Hiddenstates}
In realistic Zika epidemic scenarios, not all components of the population can be directly observed. For instance, asymptomatic human infections and unmonitored mosquito infections are typically hidden from surveillance systems. To incorporate this partial observability, we partition the full epidemic state vector into hidden and observable components.

Suppose the total system is described by a $d$-dimensional state vector $X^D_n$ at discrete time $t_n$, as in \eqref{Discrete_Dyn1}. We assume that $d_1 < d$ of these compartments are unobservable (latent, hidden), while the remaining $d_2 = d - d_1$ compartments are observable .

We decompose the state as
\[
X^D_n = \begin{pmatrix}
Y_n \\ Z_n
\end{pmatrix},
\]
where $Y_n \in \mathbb{R}^{d_1}$ contains the hidden states (e.g., asymptomatic human carriers ), and $Z_n \in \mathbb{R}^{d_2}$ represents the observable states (e.g., confirmed symptomatic human infections).

The evolution of this system can be described by a coupled stochastic recursion with initial values $Y_0= y_0, Z_0 = z_0$:
\begin{align} \label{state_YZ_Zika}
\begin{split}
Y_{n+1} &= f(n, Y_n, Z_n) + \sigma(n, Y_n, Z_n) \, \mathcal{E}^1_{n+1} + g(n, Y_n, Z_n) \, \mathcal{E}^2_{n+1}, \\
Z_{n+1} &= h(n, Y_n, Z_n) + \ell(n, Y_n, Z_n) \, \mathcal{E}^2_{n+1},
\end{split}
\end{align}
where $\mathcal{E}^1_{n}$ and $\mathcal{E}^2_{n}$ are independent sequences of standard normally distributed  vectors:
\[
\mathcal{E}^1_n \sim \mathcal{N}(0_{k_1}, \mathbb{I}_{k_1}), \quad \mathcal{E}^2_n \sim \mathcal{N}(0_{k_2}, \mathbb{I}_{k_2}),
\]
with $k_1 + k_2 = K$, the total number of stochastic drivers in the system.The noise vector $\mathcal{E}^1_n$ captures randomness exclusively affecting hidden compartments , while $\mathcal{E}^2_n$ influences both hidden and observable states. The coefficient functions  are defined on $\mathcal{G}=\{0,\ldots,N_t\} \times \mathbb{R}^{d_1} \times \mathbb{R}^{d_2}$ with 
$f: \mathcal{G} \to \mathbb{R}^{d_1}$, 
$h: \mathcal{G} \to \mathbb{R}^{d_2}$, 
$\sigma: \mathcal{G} \to \mathbb{R}^{d_1 \times k_1}$, 
$g: \mathcal{G} \to   \mathbb{R}^{d_1 \times k_2}$,  
$\ell: \mathcal{G} \to  \mathbb{R}^{d_2 \times k_2}$.
and encode the drift and diffusion terms  of the system. They will be specified in Appendix B based on the structure of the Zika transmission model presented in the next section.

\section{Zika Models}
\label{sec:model}
 \subsection{Introduction}
This section develops a class of stochastic models designed to describe the dynamics of Zika virus transmission. 
As a starting point, Subsection  \ref{compartmental} discusses the development of a stochastic compartmental model that integrates both human and vector populations. The model extends classical epidemic formulations by explicitly accounting for asymptomatic and undocumented infections, waning immunity, and vector-host interactions, which are key features influencing the spread of ZIKV. The modelling assumptions are also captured in this subsection.

In Subsection \ref{sec:symplistic}, we introduce a foundational compartmental model capturing key features of Zika transmission dynamics. A notable challenge in Zika epidemiology is the high proportion of asymptomatic or mildly symptomatic cases, many of which remain undetected due to limited testing or clinical presentation. In addition, some people who experience symptoms may not seek a formal diagnosis or healthcare services, further contributing to under-reporting. This leads to a substantial discrepancy between reported and actual infection counts, which must be addressed within the model structure. The base model incorporates these hidden processes explicitly.

Among the hidden compartments are those representing individuals who have recovered from infection and are assumed to be temporarily immune. The transitions into these recovered compartments (e.g., from symptomatic infection) are often observable, as recovery is typically recorded. However, the transition out of these states, specifically, the gradual loss of immunity and return to susceptibility is typically unobserved in surveillance systems. For Zika virus, evidence suggests that recovery confers full immunity for a limited period (see, e.g., \cite{griffin}), after which immunity wanes and susceptibility may return.

To incorporate this biological reality and leverage the available information on recovery timing, we refine the base model in Subsection \ref{sec:Extended}. The extended model introduces  cascades of recovered compartments, which stratify recovered individuals by the duration since recovery, often referred to as ``recovery age''. These cascade compartments allow the model to track the gradual transition from full immunity to susceptibility in a structured manner, based on the time elapsed since infection resolution. By modeling this delay explicitly, the extended model supports more reliable estimation of underreported cases which are commonly referred to as ``dark figures'' by capturing both the observable inflow into recovery and the unobservable outflow back to susceptibility.

\subsection{Compartmental Modeling of Zika Transmission}  \label{compartmental}  

In this section, we introduce a stochastic compartmental model designed to capture the transmission dynamics of the ZIKV, with particular emphasis on unobserved epidemiological states. Motivation comes from two critical challenges associated with ZIKV surveillance and control. First, a substantial fraction of infections are asymptomatic or present with only mild symptoms, resulting in low test-seeking behavior and, consequently, underreporting. Second, a subset of individuals can undergo an informal or unofficial diagnosis, for example, through self-testing or peer confirmation, without the cases being included in official health records. These factors contribute to significant uncertainty about the true burden of infection within the population.

Given the often benign course of ZIKV in symptomatic cases and the negligible associated mortality, a large pool of undetected cases can accumulate without triggering major public health responses. However, if this hidden infectious population grows unchecked, it can become a significant driver of transmission, undermining disease control efforts and possibly seeding new outbreaks. Therefore, accounting for these unobservable compartments is critical for effective modeling, forecasting, and policy formulation.
Empirical studies have also shown that immunity to ZIKV is not necessarily life-long. Although initial infection induces a neutralizing antibody response, this immune protection tends to decrease over time, approximately within 24 months. \cite{magal,griffin}.

\paragraph{Model Structure}
To model the above features, we start with a classic SIRS (Susceptible-Infected-Recovered-Susceptible) model, which is often used to study recurrent viral infections where immunity wanes over time. However, to accurately reflect the biological transmission mechanism of ZIKV,  we extend the standard SIRS model to include the mosquito vector population to account for vector–host interactions. Furthermore, this model exclusively considers sexual transmission from men to women, as the number of transmissions from women to men is low \cite{CDC1,Davi}.

The human population is stratified into male and female subgroups, each subdivided into compartments reflecting disease progression and immunity status: susceptible, exposed,  detected and undetected asymptomatic infected, symptomatic infected, and recovered. The mosquito vector population is modeled similarly  using susceptible, exposed and infectious compartments, incorporating seasonality.

\paragraph{Model Assumptions}
Our proposed Zika epidemic models operate under the following key assumptions with important practical justifications. The \textit{human population size is held constant at $N$}, neglecting births and natural deaths - a reasonable simplification for short-term epidemic modeling where Zika-related mortality is negligible. The \textit{mosquito birth rate  varies seasonally} according to a time-dependent function $B_v(t)$ capturing realistic temporal fluctuations in vector abundance driven by climate and environmental conditions.
Both the \textit{human and vector populations are homogeneously mixed} within and between disease compartments, which means, in particular, that every susceptible individual has the same probability of encountering infectious individuals — a standard epidemiological assumption that simplifies the transmission dynamics. Finally, all model parameters  are \textit{time-homogeneous}, with the exception of mosquito birth rate which captures seasonal variation concentrated in vector population dynamics where it matters most. These assumptions create a computationally feasible framework suitable for epidemic forecasting and control optimization in resource-limited settings while capturing the essential biological features of Zika transmission.

\subsection{Simplified Zika Model}
\label{sec:symplistic}
Zika virus presents a unique epidemiological challenge due to its dual modes of transmission: vector and direct human-to-human transmission through sexual contact \cite{Lessler, Mead}. Traditional vector host models are insufficient to capture the full scope of ZIKV dynamics, particularly in settings where sexual transmission plays a non-negligible role in sustaining outbreaks \cite{Yakob}. To address this, we construct a transmission framework that integrates vector-borne and sexual transmission pathways of the Zika virus.  Here we start with a simplified model capturing only the key properties.  In subsequent subsections, we include a more refined compartmental structure and introduce explicit \emph{exposed} compartments for male and female individuals and the vector to capture latency period between infection and infectiousness.  Furthermore, in Section \ref{sec:Extended}, we extend the model to include so-called cascade compartments in order to capture the transition from complete to waning immunity after recovery depending on the time elapsed since recovery. 

\begin{figure}[h]
	\centering
	\includegraphics[width= 0.7 \textwidth]{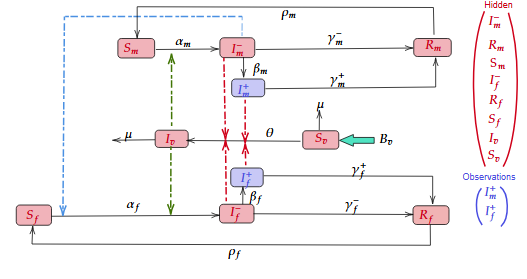}
    \label{fig:model}
	\caption{Simplified Zika virus transmission dynamics considering gender (male human$ ( {m})$, female human$ ({f}$)) and vector (${v}$).The blue compartments indicate the observable states, while the red compartments are the unobserved states (dark figures). The green dashed lines indicate the vector-to-human transmission while red lines indicate the human-to-vector transmission and blue lines indicate the human-to-human transmission. A black solid arrow shows the transition from one compartment to the next. }
	\label{model1}
\end{figure} 

The human compartments are subdivided into susceptible (\(S_m, S_f\)), infected but undetected (\(I^-_m, I^-_f\)), confirmed infected (\(I^+_m, I^+_f\)), and recovered classes, which are hidden (\(R_m, R_f\))  as illustrated in Figure \ref{model1}. The mosquito population is modeled with susceptible (\(S_v\)) and infected (\(I_v\)) classes. The model tracks the transitions between these compartments driven by infection, recovery, sexual contact, and mosquito bites. Crucially, only a subset of infections are observed, reflecting the underreporting common in real-world surveillance data \cite{Oduyebo, Turmel}. This hybrid transmission structure and observability framework enable a more realistic representation of ZIKV transmission, supporting improved inference and evaluation of control strategies.

Susceptible males (\(S_m\)) and females (\(S_f\)) become infected by bites from infected mos\-qui\-toes  (\(I_v\)). Using the placeholder notation $\dagger=m,f$ the infection rate is modeled by 
\begin{align}
	\label{infect_rate1}
	\lambda_\dagger^{S\to I}=\lambda_\dagger^{S\to I}(S_\dagger,I_v)=\alpha_\dagger S_\dagger \frac{I_v}{\halfsat+I_v},
\end{align}
with some positive constant parameters $\alpha_\dagger,\kappa_\dagger$. The reasoning behind this choice is as follows. In the CTMC modeling framework sketched in Section \ref{sec:Stochastic} the expected number of transitions from $S_\dagger$ to $I_\dagger$, i.e., the expected number of infections,  within a small interval of time of length $\Delta t$ is given by  $\lambda_\dagger^{S\to I} \Delta t$. Denote  by $A_\dagger$ the random event  that a single susceptible person from $S_\dagger$ becomes infected within  $\Delta t$ time units, and its probability by $\P(A_\dagger)$. Assuming a homogeneously mixed population and that bites and the subsequent development of infection among susceptible individuals are stochastically independent, we are within the framework of a Bernoulli scheme. Thus, the  random number of infections in $\Delta t$ is binomially distributed with parameters $S_\dagger$ and $\P(A_\dagger)$, and expectation $S_\dagger \P(A_\dagger)$. Therefore, we find for the above considered expectation   
\begin{align}
	\label{infect_rate2}
	\lambda_\dagger^{S\to I} \Delta t=S_\dagger \P(A_\dagger).
\end{align}
$\P(A_\dagger)$ is proportional to $\Delta t$ and the probability that a single susceptible is bitten by an infected  mosquito. For the latter we use Holling's Type II function $ \frac{I_v}{\halfsat+I_v}$, in which $\halfsat>0$ is the so-called ``half-saturation parameter''.   That means, for  $I_v=\halfsat$ we have $\frac{I_v}{\halfsat+I_v}=\frac{1}{2}$. This function approaches $0$ and $1$ for $I_v\to 0$ and $I_v\to\infty$, respectively. For small numbers of infected mosquitoes $I_v$, it increases almost linearly with $I_v$, and for large  $I_v$ it captures saturation effects, since it tends to $1$  for $I_v\to\infty$. 

Summarizing, we have $\P(A) \sim \frac{I_v}{\halfsat+I_v} \Delta t$. Denoting the factor of proportionality by $\alpha_\dagger$ and substituting into \eqref{infect_rate2} we obtain \eqref{infect_rate1}.

  Following infection, individuals may remain undetected (\(I^-_m, I^-_f\)), either due to an asymptomatic presentation or limited access to tests. Some asymptomatic infections are detected through surveillance or clinical diagnosis  controlled by the parameters \(\beta_m, \beta_f\)  and recorded as confirmed cases (\(I^+_m, I^+_f\)). Recovery occurs in both states, such that both observed and hidden infected individuals move to the recovery compartment (\(R_m, R_f\)), with transitions governed by rates \(\gamma^-_m, \gamma^+_m, \gamma^-_f,\) and \(\gamma^+_f\). Individuals who recover from infection (\(R_m, R_f\)) are assumed to acquire temporary immunity. However, over time, immunity can wane due to a natural decline in protective antibodies or lack of long-term immunity, thus returning individuals to the susceptible pool (\(S_m, S_f\)). 
In the proposed model, the parameter \(\rho\) represents the rate at which individuals lose immunity and return to the susceptible class. This parameter applies separately to male (\(\rho_m\)) and female (\(\rho_f\)) populations, reflecting possible differences in immune waning dynamics between the sexes. 
A notable feature of this model is the inclusion of unidirectional sexual transmission from males to females. This mechanism reflects empirical findings that ZIKV can persist for a long time in semen \cite{Turmel, Mead}. Vector dynamics involves recruitment of susceptible mosquitoes (\(S_v\)) at rate \(B_v\)  which is time-dependent to capture seasonality effects. 
 $B_v $ is considered as the relative birth rate at which new mosquitoes are produced compared to their mortality rate.

Mosquitoes can be infected when biting infectious humans. The associated transmission rate is modeled by the usual contact term which is the product of $S_v$, the fraction of infectious humans and and a scaling factor $\theta>0$, and given by     $\theta S_v{(I^-_m+I^+_m+I^-_f+I^+_f)}/{N}$, where N is the total population size.
 Infected vectors (\(I_v\)) can then transmit the virus to new human hosts before dying at rate \(\mu\).  
   
Table \ref{model1} outlines the structure and all possible transitions of the  model. The system is composed of both observable and unobservable components, which are grouped into two state vectors. The vector of latent (hidden) states is given by 
$Y = (I^{-}_m, R_m,S_m,I_f^- R_f,S_f, I_v, S_v)^\top$, while the vector of observable states is $Z = (I^{+}_m, I^{+}_f)^\top$. This results in a total of $d = 10$ distinct states and $K = 14$ possible transition pathways.\\
Within the hidden states, we have the \textit{fully hidden} states ($S_m$, $S_f $ $I_v$, and $S_v$) which are not directly observable at either entry or exit. There are also the \textit{partially hidden} states ($I^{-}_m$ and $I^-_f$) where the  inflow from $S_m,S_f$ and the outflow to $R_m,R_f$ cannot be tracked, but  the outflow to  $I^{+}_m,I^+_f$ is observed.  $R_m$ and $R_f$  are also partially hidden, since the inflow from $I_m^+,I_f^+$ is observable whereas the outflow to $S_m,S_f$ and the inflow from   $I_m^-,I_f^-$ cannot be observed.

\begin{table}[th]
\footnotesize
\begin{tabularx}{\linewidth}{ |r| l|| l|| l|}
\hline
     K &Transition & Transition vectors  &   Intensities    \\ 
\hline
\hline
    1   & Infection  to $I^{-}_m$ &\phantom{-}1, 0,-1,\phantom{-}0,\phantom{-}0,\phantom{-}0,\phantom{-}0,\phantom{-}0,\phantom{-}0,\phantom{-}0 &  $\alpha_m S_m \frac{I_v}{\halfsat+I_v} = \alpha_m  Y^3\frac{Y^7}{\halfsat+Y^7}$ \\ 
   2 & Testing from $I^{-}_m$ & -1,\phantom{-}0,\phantom{-}0,\phantom{-}0,\phantom{-}0,\phantom{-}0,\phantom{-}0,\phantom{-}0,\phantom{-}1,\phantom{-}0 &  $\beta_{m}^- I_m = \beta_{m}Y^1 $   \\
    3 &Recovery from $ I ^{-}_m$ & -1,\phantom{-}1,\phantom{-}0,\phantom{-}0,\phantom{-}0,\phantom{-}0,\phantom{-}0,\phantom{-}0,\phantom{-}0,\phantom{-}0 & $ \gamma_m^- I^{-}_m = \gamma_m^- Y^1 $\\ 
      4 &Recovery from $ I ^{+}_m$ & \phantom{-}0,\phantom{-}1,\phantom{-}0,\phantom{-}0,\phantom{-}0,\phantom{-}0,\phantom{-}0,\phantom{-}0,-1,\phantom{-}0 &  $ \gamma_m^+ I^{+}_m =\gamma_m^+ Z^1 $  \\
 5  & Loss of immunity i $R_m$ &\phantom{-}0,-1,\phantom{-}1,\phantom{-}0,\phantom{-}0,\phantom{-}0,\phantom{-}0,\phantom{-}0,\phantom{-}0,\phantom{-}0  & $ \rho_m R_m = \rho_m Y^2  $    \\ 
 \midrule
    6 &Infection  to $I^{-}_f$  & \phantom{-}0,\phantom{-}0,\phantom{-}0,\phantom{-}1,\phantom{-}0,-1,\phantom{-}0,\phantom{-}0,\phantom{-}0,\phantom{-}0 &  $\big(\alpha_f   \frac{I_v}{\halfsat+I_v} + \omega  \frac{I^-_m+I^+_m}{N_m}\big)S_f $  \\ 
     &  &  &  $ =(  \alpha_f   \frac{Y^7}{\halfsat+Y^7} + \omega  \frac{Y^1+Z^1}{N_m}) Y^3$  \\ 
    7 & Testing from $I^{-}_f$     &\phantom{-}0,\phantom{-}0,\phantom{-}0,-1,\phantom{-}0,\phantom{-}0,\phantom{-}0,\phantom{-}0,\phantom{-}0,\phantom{-}1& $\beta_{f}^- I_f = \beta_{f}Y^4 $ \\
    8 & Recovery from $ I ^{-}_f$  &\phantom{-}0,\phantom{-}0,\phantom{-}0,-1,\phantom{-}1,\phantom{-}0,\phantom{-}0,\phantom{-}0,\phantom{-}0,\phantom{-}0 &  $\gamma_f^- I^{-}_f = \gamma_f^- Y^4 $ \\ 
    9 &Recovery from $ I ^{+}_f$ & \phantom{-}0,\phantom{-}0,\phantom{-}0,\phantom{-}0,\phantom{-}1,\phantom{-}0,\phantom{-}0,\phantom{-}0,\phantom{-}0,-1 & $\gamma_f^+ I^{+}_f =\gamma_f^+ Z^2 $  \\    
  10 &Loss of immunity in $R_f $   & \phantom{-}0,\phantom{-}0,\phantom{-}0,\phantom{-}0,-1,\phantom{-}1,\phantom{-}0,\phantom{-}0,\phantom{-}0,\phantom{-}0&  $ \rho_m R_f = \rho_f Y^5  $    \\
  \midrule
    11 &Infection of  vector  &\phantom{-}0,\phantom{-}0,\phantom{-}0,\phantom{-}0,\phantom{-}0,\phantom{-}0,\phantom{-}1,-1,\phantom{-}0,\phantom{-}0& 
     { $\theta S_v{(I^-_m+I^+_m+I^-_f+I^+_f)}/{N}$} \\[0.5ex]
      &  &&
     $= \theta Y^8 {  {(Y^1+Z^1)+ Y^4_f+Z^2)/{N}}}$ \\
     12 &Birth of  vector  &\phantom{-}0,\phantom{-}0,\phantom{-}0,\phantom{-}0,\phantom{-}0,\phantom{-}0,\phantom{-}0,\phantom{-}1,\phantom{-}0,\phantom{-}0 & $   B_v(I_v+S_v)  =B_v(Y^{9}+Y^{10})$\\
      13 &Death of susceptible vector  &\phantom{-}0,\phantom{-}0,\phantom{-}0,\phantom{-}0,\phantom{-}0,\phantom{-}0,\phantom{-}0,-1,\phantom{-}0, \phantom{-}0& $ \mu S_v = \mu Y^{10}$ \\
       14 &Death of infected   vector  &\phantom{-}0,\phantom{-}0,\phantom{-}0,\phantom{-}0,\phantom{-}0,\phantom{-}0,-1,\phantom{-}0,\phantom{-}0,\phantom{-}0& $ \mu I_v = \mu Y^9 $ \\
\hline

\caption{Transition vectors and intensities of the   simplified Zika model, with state processes $Y = (I^{-}_m, R_m, S_m, I^{-}_f, R_f, S_f, I_v, S_v)^\top$ and $Z = (I^{+}_m, I^{+}_f)^\top$;  total number of compartments $d= 10 $ and number of transitions, $K=14$ } 
\label{Tab1}
\end{tabularx}
\end{table}

  The stochastic evolution of the system is described by a non-homogeneous continuous-time Markov chain (CTMC), whose transition directions and corresponding intensity functions are specified in Table \ref {Tab1}. The CTMC is approximated by a diffusion process derived via a functional central limit theorem applied to the underlying Poisson processes driving the transitions. This results in a set of stochastic differential equations (SDE) system.  
\footnotesize{
\begin{align*}
d I_m^{-} &= \left( { \alpha_m S_m \frac{I_v}{\halfsat+I_v}} - \beta_m I_m^{-} - \gamma_m^- I_m^{-} \right) dt 
+ \sqrt{{ \alpha_m S_m \frac{I_v}{\halfsat+I_v}} } \, dW_1 
- \sqrt{\beta_m I_m^{-}} \, dW_2 
- \sqrt{\gamma_m^- I_m^{-}} \, dW_3, \\[1ex]
d R_m &= \left( \gamma_m^- I_m^{-} + \gamma_m^+ I_m^{+} - \rho_m R_m \right) dt 
+ \sqrt{\gamma_m^- I_m^{-}} \, dW_3 
+ \sqrt{\gamma_m^+ I_m^{+}} \, dW_4 
- \sqrt{\rho_m R_m} \, dW_5, \\[1ex]
d S_m &= \left( -{\alpha_m S_m \frac{I_v}{\halfsat+I_v}}  + \rho_m R_m \right) dt 
- \sqrt{{\alpha_m S_m \frac{I_v}{\halfsat+I_v}} } \, dW_1 
+ \sqrt{\rho_m R_m} \, dW_5, \\[1ex]
d I_f^{-} &= \left( \alpha_f S_f { \frac{I_v}{\halfsat+I_v}}  + \omega S_f I_m^{-} - \beta_f I_f^{-} - \gamma_f^- I_f^{-} \right) dt \\
& \hspace*{1em}+ \sqrt{\alpha_f S_f { \frac{I_v}{\halfsat+I_v}} } \, dW_6 
+ \sqrt{\omega S_f I_m^{-}} \, dW_7 
- \sqrt{\beta_f I_f^{-}} \, dW_8 
- \sqrt{\gamma_f^- I_f^{-}} \, dW_9, \\[1ex]
d R_f &= \left( \gamma_f^- I_f^{-} + \gamma_f^+ I_f^{+} - \rho_f R_f \right) dt 
+ \sqrt{\gamma_f^- I_f^{-}} \, dW_9 
+ \sqrt{\gamma_f^+ I_f^{+}} \, dW_{10} 
- \sqrt{\rho_f R_f} \, dW_{11}, \\[1ex]
d S_f &= \left( -\alpha_f S_f { \frac{I_v}{\halfsat+I_v}} - \omega S_f I_m^{-} + \rho_f R_f \right) dt 
- \sqrt{\alpha_f S_f { \frac{I_v}{\halfsat+I_v}}} \, dW_6 
- \sqrt{\omega S_f I_m^{-}} \, dW_7 
+ \sqrt{\rho_f R_f} \, dW_{11}, \\[1ex]
d I_v &= \Big( \theta S_v {  \frac{I^-_m+I^+_m+I^-_f+I^+_f}{N}} - \mu I_v \Big) dt 
+ \sqrt{\theta S_v{  \frac{I^-_m+I^+_m+I^-_f+I^+_f}{N}}} \, dW_{12} 
- \sqrt{\mu I_v} \, dW_{13}, \\[1ex]
d S_v &= \Big( B_v (S_v+I_v) - \theta S_v {  \frac{I^-_m+I^+_m+I^-_f+I^+_f}{N}}  - \mu S_v \Big) dt 
\\& \hspace*{1em}
+ \sqrt{B_v} \, dW_{14} 
- \sqrt{\theta S_v {  \frac{I^-_m+I^+_m+I^-_f+I^+_f}{N}} } \, dW_{12} 
- \sqrt{\mu S_v} \, dW_{15}, \\[1ex]
d I_m^{+} &= \left( \beta_m I_m^{-} - \gamma_m^+ I_m^{+} \right) dt 
+ \sqrt{\beta_m I_m^{-}} \, dW_2 
- \sqrt{\gamma_m^+ I_m^{+}} \, dW_4, \\[1ex]
dI_f^{+} &= \left( \beta_f I_f^{-} - \gamma_f^+ I_f^{+} \right) dt 
+ \sqrt{\beta_f I_f^{-}} \, dW_8 
- \sqrt{\gamma_f^+ I_f^{+}} \, dW_{10}.
\end{align*}}

\normalsize
Subsequently, these diffusion approximations are discretized in time using the Euler–Maruyama scheme, as in Equations \ref{Discrete_Dyn1} leading to a discrete-time recursive formulation which governs the evolution of the state vectors $Y$ and $Z$ as described in Section \ref{sec:Stochastic}. The functions and coefficients used in these formulations are provided in Appendix \ref{app:sym_coeff}.

\subsection{Base Zika Model}
\label{sec:Base}
In this section, we present an expanded version of the  ZIKV transmission model. It includes a more refined compartmental structure to cater for the biological and epidemiological characteristics of the disease. In contrast to the simplified model in Section \ref{sec:Stochastic}, this formulation introduces explicit \emph{exposed} compartments for both male and female individuals $(E_m,E_f)$ and the vector ($ E_v $), capturing the latency period between infection and infectiousness. Additionally, the infected population is subdivided into symptomatic ($I^s_m, I^s_f$) and asymptomatic classes. The asymptomatic class is further stratified by whether or not the infection is detected ($I^+_m, I^+_f$) or not ($I^-_m, I^-_f$). This separation allows for a clearer distinction between observed and hidden dynamics. The recovery process is similarly disaggregated, with distinct compartments for individuals recovering from observed compartments  $(R^{2-}_m, R^{2-}_f)$ which are partially hidden) versus those recovering from hidden compartments ($R^{1-}_m, R^{1-}_f$). These additions enhance the model’s ability to account for underreporting, sex-specific disease progression, and vector-host transmission dynamics.

\begin{figure}[ht]
	\centering
	\includegraphics[width= 0.8 \textwidth]{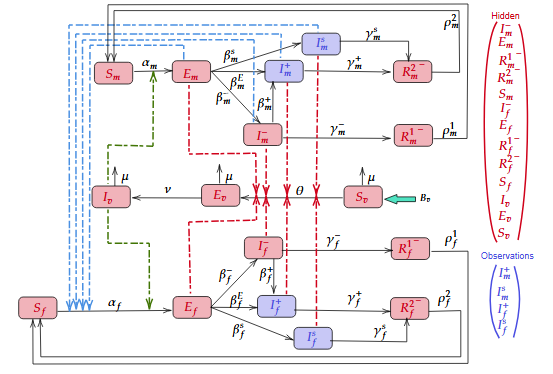}
    
	\caption{ Base Zika model:  virus transmission dynamics considering gender (male human$ ({m})$, female human $ ({f})$ and vector $({v})$.The blue compartments indicate the observable states for the males and females, while the red compartments are the unobserved states (dark figures).  The green dashed lines indicate the vector-to-human transmission while red lines indicate the human-to-vector transmission and blue lines indicate the human-to-human transmission. The transition from one compartment to the next is shown by the black solid arrow. }
	\label{model2}
\end{figure}

In the base model, males and females are also represented in parallel structures, each beginning in the susceptible class ($S_m$, $S_f$).  When an infectious mosquito comes into contact with susceptible human, the male and females, become exposed ($E_m$, $E_f$) at rates  $\alpha_\dagger S_\dagger \frac{I_v}{\halfsat+I_v}, \dagger=m,f$.  These exposed individuals at earlier stages  are not yet infectious to both the human and vectors before virema \cite{CDC2}. At later stages of exposure, shortly before the onset of the symptoms, the exposed individual becomes infectious to both human

 and vector.   The transition to the infected state after an intrinsic incubation period, 
ranges from 3 to 14 days, with an average estimate of approximately 7.5 days \cite{Lessler}.

\begin{table}[H]
	\footnotesize
	\begin{tabularx}{\linewidth}{ |r| l|| l|| l||}
		\hline
		K &Transition & Transition vectors  &   Intensities    \\ 
		\hline
		\hline
		1  & Exposed from $S_m $ &\phantom{-}0,\phantom{-}1,\phantom{-}0,\phantom{-}0,-1,\phantom{-}0,\phantom{-}0,\phantom{-}0,\phantom{-}0,\phantom{-}0,\phantom{-}0,\phantom{-}0,\phantom{-}0,\phantom{-}0,\phantom{-}0,\phantom{-}0,\phantom{-}0 &  $\alpha_m S_m \frac{I_v}{\halfsat+I_v} $ \\ 
		2 & Infection of  male to $I^{-}_m$ &\phantom{-}1 -1,\phantom{-}0,\phantom{-}0,\phantom{-}0,\phantom{-}0,\phantom{-}0,\phantom{-}0,\phantom{-}0,\phantom{-}0,\phantom{-}0,\phantom{-}0,\phantom{-}0,\phantom{-}0,\phantom{-}0,\phantom{-}0,\phantom{-}0 & $\beta_{m}^- E_m $\\
		3 & Recovery from $I^{-}_m$ & -1,\phantom{-}0,\phantom{-}1,\phantom{-}0,\phantom{-}0,\phantom{-}0,\phantom{-}0,\phantom{-}0,\phantom{-}0,\phantom{-}0,\phantom{-}0,\phantom{-}0,\phantom{-}0,\phantom{-}0,\phantom{-}0,\phantom{-}0,\phantom{-}0 &  $\gamma_{m}^-  I^{-}_m  $\\
		4 &Loss of  immunity in $ R ^{1-}_m$ &\phantom{-}0,\phantom{-}0,-1,\phantom{-}0,\phantom{-}1,\phantom{-}0,\phantom{-}0,\phantom{-}0,\phantom{-}0,\phantom{-}0,\phantom{-}0,\phantom{-}0,\phantom{-}0,\phantom{-}0,\phantom{-}0,\phantom{-}0,\phantom{-}0 & $ \rho^1 _m R^{1-}_m  $\\
		5 &Loss of immunity in $R ^{2-}_m$ & \phantom{-}0,\phantom{-}0,\phantom{-}0,-1,\phantom{-}1,\phantom{-}0,\phantom{-}0,\phantom{-}0,\phantom{-}0,\phantom{-}0,\phantom{-}0,\phantom{-}0,\phantom{-}0,\phantom{-}0,\phantom{-}0,\phantom{-}0,\phantom{-}0 &  $\rho^2 _m R^{2-}_m $\\
		\hline
		6  & Exposed from $S_f$ &\phantom{-}0,\phantom{-}0,\phantom{-}0,\phantom{-}0,\phantom{-}0,\phantom{-}0,\phantom{-}1,\phantom{-}0,\phantom{-}0,-1,\phantom{-}0,\phantom{-}0,\phantom{-}0,\phantom{-}0,\phantom{-}0,\phantom{-}0,\phantom{-}0  &  $ (\alpha_f    \frac{I_v}{\halfsat+I_v}  + \omega  {neu \frac{\Isum_m}{N_m}}) S_f $\\
		7 &Infection of female to $I^{-}_f$ & \phantom{-}0,\phantom{-}0,\phantom{-}0,\phantom{-}0,\phantom{-}0,\phantom{-}1,-1,\phantom{-}0,\phantom{-}0,\phantom{-}0,\phantom{-}0,\phantom{-}0,\phantom{-}0,\phantom{-}0,\phantom{-}0,\phantom{-}0,\phantom{-}0  & $\beta_{f}^- E_f $\\
		8 & Recovery from $I^{-}_f $   &\phantom{-}0,\phantom{-}0,\phantom{-}0,\phantom{-}0,\phantom{-}0,-1,\phantom{-}0,\phantom{-}1,\phantom{-}0,\phantom{-}0,\phantom{-}0,\phantom{-}0,\phantom{-}0,\phantom{-}0,\phantom{-}0,\phantom{-}0,\phantom{-}0& $\gamma_{f}^-  I^{-}_f$\\
		9 & Loss of immunity in  $ R ^{1-}_f$  &\phantom{-}0,\phantom{-}0,\phantom{-}0,\phantom{-}0,\phantom{-}0,\phantom{-}0,\phantom{-}0,-1,\phantom{-}0,\phantom{-}1,\phantom{-}0,\phantom{-}0,\phantom{-}0,\phantom{-}0,\phantom{-}0,\phantom{-}0,\phantom{-}0 &  $\rho^1 _f R^{1-}_f   $\\
		10 &Loss of immunity in $R ^{2-}_f$  & \phantom{-}0,\phantom{-}0,\phantom{-}0,\phantom{-}0,\phantom{-}0,\phantom{-}0,\phantom{-}0,\phantom{-}0,-1,\phantom{-}1,\phantom{-}0,\phantom{-}0,\phantom{-}0,\phantom{-}0,\phantom{-}0,\phantom{-}0,\phantom{-}0 & $\rho^2_f R^{2-}_f $\\
		\hline 
		11 &Exposure of vector   & \phantom{-}0,\phantom{-}0,\phantom{-}0,\phantom{-}0,\phantom{-}0,\phantom{-}0,\phantom{-}0,\phantom{-}0,\phantom{-}0,\phantom{-}0,\phantom{-}1,\phantom{-}0,-1,\phantom{-}0,\phantom{-}0,\phantom{-}0,\phantom{-}0 &  $\theta  S_v    \frac\Isum{N}  $\\
		12 &Infection of  vector  &\phantom{-}0,\phantom{-}0,\phantom{-}0,\phantom{-}0,\phantom{-}0,\phantom{-}0,\phantom{-}0,\phantom{-}0,\phantom{-}0,\phantom{-}0,-1,\phantom{-}1,\phantom{-}0,\phantom{-}0,\phantom{-}0,\phantom{-}0,\phantom{-}0 & $ \nu E_v $\\
		
		13 &Birth of susceptible  vector  &\phantom{-}0,\phantom{-}0,\phantom{-}0,\phantom{-}0,\phantom{-}0,\phantom{-}0,\phantom{-}0,\phantom{-}0,\phantom{-}0,\phantom{-}0,\phantom{-}0,\phantom{-}0,\phantom{-}1,\phantom{-}0,\phantom{-}0,\phantom{-}0,\phantom{-}0  & $ B_v ( S_v  + E_v + I_v)$\\
		14 &Death of susceptible vector  &\phantom{-}0,\phantom{-}0,\phantom{-}0,\phantom{-}0,\phantom{-}0,\phantom{-}0,\phantom{-}0,\phantom{-}0,\phantom{-}0,\phantom{-}0,\phantom{-}0,\phantom{-}0,-1,\phantom{-}0,\phantom{-}0,\phantom{-}0,\phantom{-}0 & $ \mu S_v $\\
		15 &Death of exposed vector  &\phantom{-}0,\phantom{-}0,\phantom{-}0,\phantom{-}0,\phantom{-}0,\phantom{-}0,\phantom{-}0,\phantom{-}0,\phantom{-}0,\phantom{-}0,\phantom{-}0,-1,\phantom{-}0,\phantom{-}0,\phantom{-}0,\phantom{-}0,\phantom{-}0 & $ \mu E_v $\\
		16 &Death of infected vector  &\phantom{-}0,\phantom{-}0,\phantom{-}0,\phantom{-}0,\phantom{-}0,\phantom{-}0,\phantom{-}0,\phantom{-}0,\phantom{-}0,\phantom{-}0,-1,\phantom{-}0,\phantom{-}0,\phantom{-}0,\phantom{-}0,\phantom{-}0,\phantom{-}0 & $ \mu I_v $\\
		\hline
		17 & Infection of  male to $I^s_m$ & \phantom{-}0,-1,\phantom{-}0,\phantom{-}0,\phantom{-}0,\phantom{-}0,\phantom{-}0,\phantom{-}0,\phantom{-}0,\phantom{-}0,\phantom{-}0,\phantom{-}0,\phantom{-}0,\phantom{-}1,\phantom{-}0,\phantom{-}0,\phantom{-}0 & $\beta_{m}^s E_m $\\
		18 &  Testing of male to $I^{+}_m $&\phantom{-}0,-1,\phantom{-}0,\phantom{-}0,\phantom{-}0,\phantom{-}0,\phantom{-}0,\phantom{-}0,\phantom{-}0,\phantom{-}0,\phantom{-}0,\phantom{-}0,\phantom{-}0,\phantom{-}0,\phantom{-}1,\phantom{-}0,\phantom{-}0 &  $\beta_{m}^+ E_m  $\\
		19 &  Testing male from $I^-_m$ to $I^{+}_m $& -1,\phantom{-}0,\phantom{-}0,\phantom{-}0,\phantom{-}0,\phantom{-}0,\phantom{-}0,\phantom{-}0,\phantom{-}0,\phantom{-}0,\phantom{-}0,\phantom{-}0,\phantom{-}0,\phantom{-}1,\phantom{-}0,\phantom{-}0,\phantom{-}0 &  $\beta_{m}^E I^-_m  $\\
		20 & Recovery from $I^{s}_m $  & \phantom{-}0,\phantom{-}0,\phantom{-}0,\phantom{-}1,\phantom{-}0,\phantom{-}0,\phantom{-}0,\phantom{-}0,\phantom{-}0,\phantom{-}0,\phantom{-}0,\phantom{-}0,\phantom{-}0,-1,\phantom{-}0,\phantom{-}0,\phantom{-}0 &   $\gamma_{m}^s  I^{s}_m   $\\
		21 & Recovery from $I^{+}_m$    &\phantom{-}0,\phantom{-}0,\phantom{-}0,\phantom{-}1,\phantom{-}0,\phantom{-}0,\phantom{-}0,\phantom{-}0,\phantom{-}0,\phantom{-}0,\phantom{-}0,\phantom{-}0,\phantom{-}0,\phantom{-}0,-1,\phantom{-}0,\phantom{-}0 &  $\gamma_{m}^+  I^{+}_m $\\
		\hline
		22 & Infection of female to $ I^s_f$ & \phantom{-}0,\phantom{-}0,\phantom{-}0,\phantom{-}0,\phantom{-}0,\phantom{-}0,-1,\phantom{-}0,\phantom{-}0,\phantom{-}0,\phantom{-}0,\phantom{-}0,\phantom{-}0,\phantom{-}0,\phantom{-}0,\phantom{-}1,\phantom{-}0 & $\beta_{f}^s E_f $\\
		23 & Testing of female to $I^{+}_f$ & \phantom{-}0,\phantom{-}0,\phantom{-}0,\phantom{-}0,\phantom{-}0,\phantom{-}0,-1,\phantom{-}0,\phantom{-}0,\phantom{-}0,\phantom{-}0,\phantom{-}0,\phantom{-}0,\phantom{-}0,\phantom{-}0,\phantom{-}0,\phantom{-}1 & $\beta_{f}^+ E_f $\\
		24 &  Testing female from $I^-_f$ to $I^{+}_m $& \phantom{-}0,\phantom{-}0,\phantom{-}0,\phantom{-}0,\phantom{-}0,-1, \phantom{-}0,\phantom{-}0,\phantom{-}0,\phantom{-}0,\phantom{-}0,\phantom{-}0,\phantom{-}0,\phantom{-}0,\phantom{-}0,\phantom{-}1,\phantom{-}0 &  $\beta_{f}^E I^-_f  $\\
		25 & Recovery from $I^{s}_f$    & \phantom{-}0,\phantom{-}0,\phantom{-}0,\phantom{-}0,\phantom{-}0,\phantom{-}0,\phantom{-}0,\phantom{-}0,\phantom{-}1,\phantom{-}0,\phantom{-}0,\phantom{-}0,\phantom{-}0,\phantom{-}0,\phantom{-}0,-1,\phantom{-}0   &   $\gamma_{f}^s  I^{s}_f $\\
		26 & Recovery from $I^{+}_f $ &\phantom{-}0,\phantom{-}0,\phantom{-}0,\phantom{-}0,\phantom{-}0,\phantom{-}0,\phantom{-}0,\phantom{-}0,\phantom{-}1,\phantom{-}0,\phantom{-}0,\phantom{-}0,\phantom{-}0,\phantom{-}0,\phantom{-}0,\phantom{-}0,-1  &   $\gamma_{f}^+  I^{+}_f  $\\
		\hline
	\end{tabularx}
	
	\medskip
	\caption  {Transition vectors and intensities of the   base Zika model, with hidden state $Y = (I^{-}_m,E_m, R_m^{1-}, R_m^{2-} S_m,I_f^- ,E_f, R_f^{1-},R_f^{2-},S_f, I_v,E_v, S_v)^\top$, and the observable state $Z = (I^{+}_m, I^{s}_m,I^{+}_f , I^{s}_f)^\top$,  total number of compartments $d= 17 $ and number of transitions $K=20$.  
	}
	\label{Tab2}
\end{table}

 As in the simplified model, there is also a possible transition from undetected  $I_m^{-}$ and $I_f^{-}$  to  detected  $I_m^{+}$ and $I_f^{+}$ via testing rates \(\beta ^+_m ,\beta^+_f\).   Cohort and surveillance studies indicate that Zika viremia is reliably detectable in serum shortly before or at the start of the symptom and is most reliably detected by  Reverse Transcription Polymerase Chain Reaction (RT-PCR) tests  \cite {Fontaine}. For base model and later, the extended model, this dynamic is represented by the transition from $E_m, E_f$ to $I^+_m, I^+_f$ via testing at the rate $\beta^E_m, \beta ^E_f$.  
Recovery from infection also occurs through distinct pathways. Detected symptomatic or asymptomatic individuals transition into partially hidden recovery states $R_m^{2-}$, $R_f^{2-}$ for men and women respectively. These compartments are considered partially hidden due to observed inflows into the compartment and unobserved outflows, which will be discussed in the next section on the extended Zika model.
The undetected individuals enter fully hidden recovery compartments ($R_m^{1-}$, $R_f^{1-})$. These states reflect differing levels of immunity awareness and tracking. Individuals in all recovery compartments may eventually lose immunity and re-enter the susceptible state, due to waning immunity over time.

In the vector population, mosquitoes are classified into susceptible ($S_v$), exposed ($E_v$), and infectious ($I_v$) compartments. Susceptible mosquitoes become exposed after biting an infectious human in the compartments $I^-_\dagger,I^+_\dagger,I^s_\dagger$ or  $E_\dagger$, and then progress to the infectious stage at the rate $\nu$ following an incubation period of approximately 7.5 days.  It is worth noting that only a fraction $\psi$ of the human population is  infectious. 

Figure \ref{model2} gives the structure of all possible transitions in the base model.  
 For simplifying the notation we denote the sum of all infectious men and women by $\Isum_m$ and $\Isum_f$, respectively, and by $\Isum=\Isum_m+\Isum_f$ the total number of infectious  in the human population. It holds $\Isum_\dagger=I^-_\dagger+I^+_\dagger+ I^s_\dagger+ \propexposed E_\dagger$ for $\dagger=m,f$. Here, the scaling factor $\propexposed\in[0,1]$ takes into account the fact that not all exposed individuals are infectious, but only those who are about to develop symptoms and transition to one of the infected compartments. Then the fraction $\Isum_m/N_m$ contributes to the rate of the transition $S_f\to E_f$ in the female population due to a contact with an infectious male, for which we have the rate $\omega S_f \Isum_m/N_m$. Note that  the total male  subpopulation size $N_m$  and female population $N_f$ are constant, since there is no birth and death in the human population.

We denote by $\beta^{+}_m$ , $ \beta^{+}_f$  and  $\beta^{E}_m,\beta^{E}_f$ the \emph{testing rates}, which quantify the rates at which infected  and exposed individuals respectively, are identified by diagnostic screening or clinical confirmation. This parameter governs the flow of individuals from undetected infection classes into observed compartments. 
 The parameters $\beta^{s}_m$ and $ \beta^{s}_f$ denote the \emph{rates of symptomatic infection}, which describe the rate at which  exposed  individuals  transition to the observed symptomatic infectious states. By contrast, $\beta^{-}_m,\beta^{-}_f$ represents \emph{rates of unobserved infection}, corresponding to infections that escape detection due to asymptomatic progression, limited health-seeking behavior, or surveillance inefficiencies. 
Recovery is likewise partitioned between observable and unobservable processes. The rates $\gamma^{s}_m , \gamma^{s}_f$ and $\gamma^{+}_m, \gamma^{+}_f$ represent the \emph{rates of observed recovery}, with $\gamma^{s}_m,\gamma^{s}_f$ capturing recovery among symptomatic infections while $\gamma^{+}_m, \gamma^{+}_f$ representing recovery following confirmed positive tests. On the other hand, $\gamma^{-}_m,\gamma^{-}_f$ corresponds to the \emph{rates of unobserved recovery}, encompassing recoveries among individuals whose infections were never reported. The rates $\rho_{m}^1 , \rho_{f}^1$ denotes the transition from the unobserved recovery  classes $R^{1-}_m,R^{1-}_f$ back to susceptibility, representing the waning immunity among those who recovered from undetected infections. Similarly, $\rho_{m}^2, \rho_{f}^2$  describe the \emph{rate of waning immunity} from the partially observed compartments   $R^{2-}_m,R^{2-}_f$, which contains individuals who recovered from infections in the observed classes $I^{+}$ and $I^{s}$ but their immunity eventually diminishes.  

  The recruitment of susceptible mosquitoes $(S_v)$ is given by the rate $B_v$, which varies over time to account for seasonal fluctuations. The parameter $B_v$ is the relative birth rate,  at which new mosquitoes emerge in relation to their natural mortality rate.
	The rate for the transition $S_v\to E_v $ depends on the total proportion of infectious humans $\Isum/N$ 
	and is given by 	$\theta S_v\Isum/N$ where $\theta$ is  the rate of exposure of the vector. The exposed vector then becomes infectious and transition to $I_v$ after approximately 7.5 days at the rate  $\nu$.

 Figure \ref{model2} is also composed of both observable and unobservable components.The vector of hidden states is given by $Y = (I^{-}_m,E_m, R_m^{1-}, R_m^{2-} S_m,I_f^- ,E_f, R_f^{1-},R_f^{2-},S_f, I_v,E_v, S_v)^\top$, while the vector of observable states is $Z = (I^{+}_m, I^{s}_m,I^{+}_f, I^{s}_f)^\top$. This results in a total of $d = 17$ distinct states and $K = 26$ possible transition pathways.

Summarizing, within the hidden states, we have the \textit{fully hidden} states ($S_m$, $S_f $ $I_v$, and $S_v$) which are not directly observable at either inflow or outflow. There are also the \textit{partially hidden} states ($I^{-}_m$ and $I^-_f$) where the inflow is not visible, but the outflow is observed to the compartments $I^{+}_m$ and $I^+_f$ respectively. $R_m$ and $R_f$  also remains partially hidden with observed inflows and unobserved outflows. 

As described earlier, Equations (\ref{stateX}) and (\ref{stateX_Poisson}) are used to model the dynamics as a continuous-time Markov chain that is discretized as in (\ref{Discrete_Dyn1}). The specific transition directions and corresponding Poisson intensities are detailed in Table \ref{Tab2}. As population sizes become large, the dynamics of CTMC can be approximated by a system of stochastic differential equations (SDEs), presented in Equation (\ref{DA1}). For numerical analysis, the system is discretized, leading to the recursive formulation in Equation (\ref{state_YZ_Zika}). The functions and coefficients used in these formulations ( $f$, $h$, $g$, $\sigma$, and $\ell$) are provided in Appendix \ref{app:Base}.
\subsection{Extended Zika Model }
\label{sec:Extended}
Building upon the foundational structure of the base Zika transmission model developed in the preceding section, we introduce a refined and extended formulation that captures immunological memory through compartments organized by recovery age. Although the base model incorporated key epidemiological interactions between human (male and female) and mosquito populations, it assumed a homogeneous recovered class without accounting for the dynamics of periods during which an individual can have perfect immunity.

To incorporate the effect of temporary full immunity into the model and to improve the estimation of dark figures, we introduce the notion of \emph{recovery age}, defined as the elapsed time since an individual has recovered from infection. Individuals are stratified into subcompartments according to their recovery age, so that each subcompartment contains only those with the same duration since recovery.

\begin{figure}[h]
	\centering
	\includegraphics[width= 0.9 \textwidth]{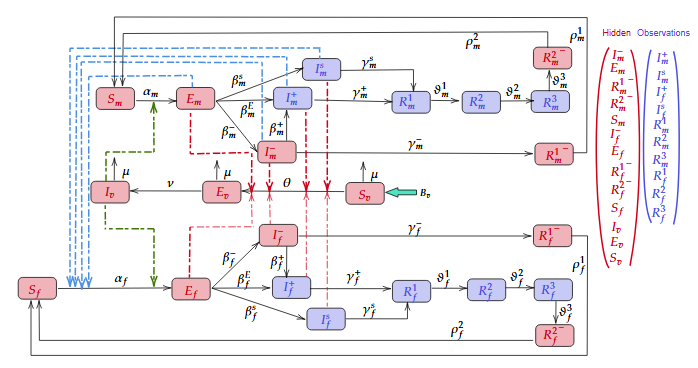}
	\caption{{{ Extended Zika virus} transmission model highlighting the inclusion of $L = 3 $ cascaded recovery compartments \(R^1_m \to R^2_m \to R^3_m\) for males and \(R^1_f \to R^2_f \to R^3_f\) for females. All these cascade compartments are included in the observed states and other compartmental transitions remain as defined in the base model. }}
	\label{model3}
\end{figure}

	All previously defined variables and transitions retain their original meaning, and only the immunity cascade and its associated transitions represent new components in this extension. To incorporate information from observable random transitions into the estimation of unobserved (hidden) states, we adopt a cascade framework from \cite{Kamkumo, Kamkumo2} and  decompose compartments with observable  inflows into a sequence of intermediate stages. 
The recovered population in $R^{2-}_m, R^{2-}_f$ is stratified according to their \emph{recovery age}, i.e., the time elapsed since recovery. For brevity, we use the placeholder notation $\dagger=m,f$ and simply write $R^{2-}_\dagger$.

Let  $L \in \mathbb{N}$ denote the number of full-immunity periods of length $\Delta t$, the step size of the  time discretization,   such that complete immunity persists for the time  $L \Delta t$ after recovery. The recovered class is then represented as a sequence of subcompartments $(R^1_\dagger,\ldots,R_\dagger^{L})$, followed by a residual class $R^{2-}_\dagger$,where $R^j_\dagger$ contains individuals with recovery age in $((j-1)\Delta t,j\Delta t]$.  The last compartment  $R^{2-}_\dagger$ includes those with a recovery age greater than $L \Delta t$ who have lost their full immunity and can transition to the susceptible compartments at a rate\footnote{This rate is different to the one in the base model!} of $\rho^2_\dagger$.
Transitions within the cascade are deterministic, with $R^j_{\dagger,n+1}=R^{j-1}_{\dagger,n}$ for $j \geq 2$, while the first compartment receives a random inflow $R^1_{\dagger,n+1}=\Rin_{\dagger,n}$. Here,   $\Rin_{\dagger,n}$ counts newly recovered individuals who enter the first cascade compartment  during the interval $(t_n, t_{n+1}]$, it is given by  $\Rin_{\dagger,n}=( \gamma^s_\dagger I^s_{\dagger,n} + \gamma^+_\dagger I^+_{\dagger,n} )\Delta t$ and 
is observable since it depends on the observable states $I^s_\dagger,I^+_\dagger$. For the last compartment the dynamics reads $R^{2-}_{\dagger, n+1}=R^{2-}_{\dagger, n}+R^{L,n}-\Rout_{\dagger,n}$  with the outflow to $S_\dagger$ given by $\Rout_{\dagger,n}=\rho^2_\dagger R^{2-}_{\dagger,n}\Delta t$.

 When $L$ is large, this construction may introduce too many new compartments that add little statistical information because only $R^1_\dagger$ represents  the new information  from the observable inflow that was not yet been  captured in the base base model. To limit the number of compartments and reduce model complexity, we aggregate the $L$ cascade stages into $\KR \leq L$ compartments $R^1_\dagger, \dots, R^{\KR}_\dagger$. Each $R^j_\dagger$ groups $P_j\in\N$ adjacent recovery-age compartments with $\sum_{j=1}^{\KR} P_j = L$, and represents recovery ages in $\big(\sum_{k<j}P_k,\sum_{k\le j}P_k\big]\Delta t$. This reduction leads to some loss of resolution, but it is computationally more efficient and still captures the essential dynamics of immunity.
 
  The transition between aggregated classes is approximated by assuming that, at each time step, a fraction $\vartheta_j=1/P_j$ of the individuals in $R^j_\dagger$ progress to $R^{j+1}_\dagger$, which is an approximation that is accurate under an  uniform distribution of recovery ages within $R^j_\dagger$.  This leads to the following recursions for the new cascade states. 
 \begin{align}\label{cascade_dynamics_R}
 	\begin{split}
 		R^1 _{\dagger,n+1} &= (1 - \vartheta_1) R^1_{\dagger,n} + \Rin_{\dagger,n}, \\
 		R^j_{\dagger,n+1} &= (1 - \vartheta_j) R^j_{\dagger,n} + \vartheta_{j-1} R^{j-1}_{\dagger,n}, \quad \text{for } j = 2, \dots, \KR, \\
 		R^{2-}_{\dagger,n+1} &= R^{2-}_{\dagger,n} + \vartheta_{\KR} R^{\KR}_{\dagger,n} - \Rout_{\dagger,n}.
 	\end{split}
 \end{align}
 We refer to  Appendix \ref{app:Extended} for a complete description of the extended Zika model dynamics.

\section{Estimation of Unobservable States}
\label{sec:filter}

Recall that in the context of Zika virus transmission modeling, not all compartments of the population are directly observable. For example, the number of asymptomatic infections or the vector infection level may be hidden from direct measurement. The goal of filtering is to estimate the hidden state variable \( Y_n \) at each discrete time point \( n \), based on all available observations \( Z_0, \ldots, Z_n \) and prior knowledge of the initial hidden state distribution. These variables evolve according to the stochastic model defined in \eqref{state_YZ_Zika}.

Formally, the filtering problem consists in computing the optimal mean-square estimate of \( Y_n \) given the filtration
\[
\mathcal{F}^Z_n := \sigma \{ Z_k : 0 \leq k \leq n \} \vee \mathcal{F}^I_0,
\]
which captures all information available up to time \( n \), including the observation history and prior information encoded in \( \mathcal{F}^I_0 \)  such as estimates from past outbreaks, historical surveillance data, or expert knowledge. The optimal estimate is the conditional mean
\begin{align}
\label{def_condmean}
\condmean_n := \mathbb{E} \left[ Y_n \,\big|\, \mathcal{F}^Z_n \right],
\end{align}
which minimizes the mean-square error \( \mathbb{E}[\|Y_n - \hat{Y}_n\|^2] \) among all \( \mathcal{F}^Z_n \)-measurable estimators \( \hat{Y}_n \in L^2(\Omega, \mathcal{F}^Z_n) \).The uncertainty associated with the estimate \( \condmean_n \) is quantified by the conditional covariance matrix
\begin{align}
\label{def_condvar}
\condvar_n := \operatorname{Var}(Y_n \mid \mathcal{F}^Z_n) 
= \mathbb{E} \left[ (Y_n - \condmean_n)(Y_n - \condmean_n)^\top \,\big|\, \mathcal{F}^Z_n \right].
\end{align}

The recursive filtering process is initialized at time \( n = 0 \) by:
\begin{align}
\label{def_filter_ini}
\condmean_0 = m_0 := \mathbb{E}[Y_0 \mid \mathcal{F}^Z_0], 
\quad \condvar_0 = q_0 := \operatorname{Var}(Y_0 \mid \mathcal{F}_0^Z),
\end{align}

A central challenge in the Zika model is the presence of nonlinear dynamics in both the hidden and observed components. The drift term \( f(n, Y_n, Z_n) \) and the diffusion terms \( \sigma(n, Y_n, Z_n), g(n, Y_n, Z_n), \ell(n, Y_n, Z_n) \) are nonlinear functions of the states. Consequently, classical Kalman filtering theory, which assumes linear-Gaussian dynamics, is not applicable in this setting.

To address this, appropriate filtering methods must be employed. In this work, we use the \textit{Extended Kalman Filter} (EKF), which relies on local linearization of the nonlinear system around the current estimate. This approach provides a computationally efficient approximation to the optimal filter and is particularly suitable when the process noise remains approximately Gaussian.

The formulation and implementation of the EKF tailored to our epidemic model are developed in Subsections \ref{sec:EKF} and \ref{ACGS}, following the general filtering framework introduced in the next subsection.

\subsection{Kalman Filtering for Conditionally Gaussian State-Space Models}  
\label{Sect_Cond_gaussian}

Consider a partially observed stochastic process $\binom{Y_n}{Z_n}$
where \( Y_n \in \mathbb{R}^{d_1} \) denotes the unobservable (hidden) state vector and \( Z_n \in \mathbb{R}^{d_2} \) denotes the observed process, for discrete times \( n = 0, \dots, N_t \), with \( N_t \in \mathbb{N} \), \( d_1, d_2 \in \mathbb{N} \). The evolution of this system is governed by the following stochastic recursions:
\begin{align}
\begin{split}
	Y_{n+1} &= \cf_0(\Zpathn) + \cf_1(\Zpathn) Y_n + \csigma(\Zpathn) \mathcal{E}^1_{n+1} + \cg(\Zpathn) \mathcal{E}^2_{n+1}, \\
	Z_{n+1} &= \ch_0(\Zpathn) + \ch_1(\Zpathn) Y_n + \cell(\Zpathn) \mathcal{E}^2_{n+1},
\end{split}
\label{kfcgm2}
\end{align}
with given initial conditions \( Y_0 \in \mathbb{R}^{d_1} \), \( Z_0 \in \mathbb{R}^{d_2} \), and where \( \Zpathn := (Z_0, \dots, Z_n) \) denotes the trajectory of the observed process up to time \( n \). The noise processes \( \{\varepsilon^1_n\}_{n=1}^{N_t} \subset \mathbb{R}^{k_1} \) and \( \{\varepsilon^2_n\}_{n=1}^{N_t} \subset \mathbb{R}^{k_2} \) are independent sequences of standard Gaussian random vectors, and are independent of the initial states \( Y_0 \), \( Z_0 \). The coefficient functions \( \cf_0, \cf_1, \csigma, \cg, \ch_0, \ch_1, \cell \) are measurable functions of the observation path \( \Zpathn \), ensuring that all terms in \eqref{kfcgm2} are well-defined and dimensionally consistent.

We impose the following technical assumptions:

\begin{assumption}~	\label{ass_filter}
Let \( \cb = \widetilde{b}(\Zpathn) \) denote any of the measurable coefficient functions \( \cf_0, \cf_1, \csigma, \cg, \ch_0, \ch_1, \cell \). Then for each \( n = 0, \dots, N_t - 1 \), the following hold:
\begin{enumerate}[label=(A\arabic*), leftmargin=2.5em]
	\item \label{kfcgm_A3} \textbf{Square integrability:}  
	\[
	\mathbb{E} \left[ \| \cb(\Zpathn) \|^2 \right] < \infty.
	\]
	
	\item \label{kfcgm_A4} \textbf{Uniform boundedness (almost surely):}  
	\[
	\| \cf_1(\Zpathn) \|, \; \| \ch_1(\Zpathn) \| \leq C < \infty \quad \mathbb{P}\text{-a.s.}
	\]
	
	\item \label{kfcgm_A5} \textbf{Integrable initial state:}  
	\[
	\mathbb{E}[\| Y_0 \|^2 + \| Z_0 \|^2] < \infty.
	\]
	
	\item \label{kfcgm_A6} \textbf{Gaussian prior:}  
	The conditional distribution of \( Y_0 \) given \( \mathcal{F}_0^Z = \sigma(Z_0) \vee \mathcal{F}_{0}^I \) is Gaussian.
\end{enumerate}
\end{assumption}

\begin{remark}
In contrast to the nonlinear stochastic epidemic models such as \eqref{state_YZ_Zika}, where the drift and diffusion terms are nonlinear functions of both the hidden and observed variables (e.g., due to nonlinear infection dynamics or seasonality in vector populations), the signal process in \eqref{kfcgm2} evolves according to a \textit{conditionally linear and Gaussian} structure.  Specifically, the drift of \( Y_{n+1} \) is affine in the signal \( Y_n \), and 
    the diffusion coefficients \( \csigma(\Zpathn) \), \( \cg(\Zpathn) \), and \( \cell(\Zpathn) \) depend only on the observation history \( \Zpathn \), not on the current or past values of the hidden state \( Y_n \).
Such structure allows closed-form filtering recursions analogous to the classical Kalman filter, albeit in a non-autonomous setting where the coefficients vary with the observation history.
\end{remark}

Given that only the observation sequence \( \{Z_n\} \) is available, the filtering problem consists of computing, for each time \( n = 0, \dots, N_t \), the conditional expectation
$\condmean_n := \mathbb{E}[Y_n \mid \mathcal{F}_n^Z]$,and the associated conditional covariance matrix $\condvar_n := \operatorname{Var}(Y_n \mid \mathcal{F}_n^Z)$,which together describe the \textit{optimal mean-square estimate} and its uncertainty, based on the observation path \( \Zpathn \) and prior information \( \mathcal{F}_{0}^I \).

Under Assumption \ref{ass_filter}, it is well known (cf. \cite[Chapter 13]{LiptserShiryaevVolII2001}) that
  \( \condmean_n \) is the unique mean-square optimal estimate of \( Y_n \) given \( \mathcal{F}_n^Z \), and 
    the total expected squared estimation error is given by $
    \operatorname{tr}(\mathbb{E}[\condvar_n]) = \sum_{i=1}^{d_1} \mathbb{E}[(Y_n^i - \condmean_n^i)^2]$.

While in general nonlinear filtering problems, recursive computation of \( (\condmean_n, \condvar_n) \) requires approximations (e.g., particle filters, extended Kalman filters), the conditional Gaussian structure of \eqref{kfcgm2} allows an \textit{exact recursive solution} analogous to the classical Kalman filter. These recursions are derived below.

\begin{theorem}[Liptser \& Shiryaev (2001) {\cite[Theorem 13.3]{LiptserShiryaevVolII2001}}] \ \\
\label{thm_condGauss}%
Let Assumption~\ref{ass_filter} hold. Consider the stochastic system defined by the recursive state-space model \eqref{kfcgm2}, where $\{Y_n\}$ is the hidden state sequence and $\{Z_n\}$ is the corresponding sequence of observations. Then, the joint process $\{(Y_n, Z_n)\}_{n=0}^{N_t}$ is conditionally Gaussian with respect to the filtration $\{\mathcal{F}_n^Z\}_{n \geq 0}$. That is, for every $n \in \{0,1,\ldots,N_t\}$, the conditional distribution of \( Y_n \) given \( \mathcal{F}_n^Z \) is multivariate Gaussian.
\end{theorem}

\begin{proof}
A proof is provided in \cite[Theorem 13.3]{LiptserShiryaevVolII2001} and also in \cite{chen1989kalman}. The result follows by mathematical induction, relying on the recursive linear-Gaussian structure of the model \eqref{kfcgm2}. \qed
\end{proof}

\noindent The conditional Gaussianity established in Theorem~\ref{thm_condGauss} permits an explicit recursive computation of the conditional mean and conditional covariance of the hidden states $\{Y_n\}$, given observations $\{Z_n\}$. These recursive filtering equations, which generalize the Kalman filter to time-varying and nonlinear models, are stated next.

\begin{theorem}[Liptser \& Shiryaev (2001) {\cite[Theorem 13.4]{LiptserShiryaevVolII2001}}]\ \\
\label{Theo_EKF}%
Let Assumption~\ref{ass_filter} be satisfied, and suppose the sequences $\{Y_n\}$ and $\{Z_n\}$ evolve according to the conditionally Gaussian model \eqref{kfcgm2}. Then the conditional distribution of \( Y_n \) given \( \mathcal{F}_n^Z \) is the multivariate Gaussian distribution $\mathcal{N}(\condmean_n, \condvar_n)$, with parameters computed recursively as follows:
\begin{align}
\condmean_{n+1} &= \cf_0 + \cf_1 \condmean_n \nonumber \\
&\quad + \left( \cg \cell^\top + \cf_1 \condvar_n \ch_1^\top \right) \left( \cell \cell^\top + \ch_1 \condvar_n \ch_1^\top \right)^+ \left( Z_{n+1} - \ch_0 - \ch_1 \condmean_n \right), 
\label{eq:EKF_mean_update} \\[1em]
\condvar_{n+1} &= - \left( \cg \cell^\top + \cf_1 \condvar_n \ch_1^\top \right) \left( \cell \cell^\top + \ch_1 \condvar_n \ch_1^\top \right)^+ \left( \cg \cell^\top + \cf_1 \condvar_n \ch_1^\top \right)^\top \nonumber \\
&\quad + \cf_1 \condvar_n \cf_1^\top + \csigma \csigma^\top + \cg \cg^\top,
\label{eq:EKF_cov_update}
\end{align}

\noindent with initial conditions $\condmean_0 = m_0$ and $\condvar_0 = q_0$, where all functions are evaluated at the observed path $\Zpathn$.
\end{theorem}

\begin{proof}
See \cite[Theorem 13.4]{LiptserShiryaevVolII2001} or \cite{chen1989kalman} for a complete derivation. \qed
\end{proof}

\begin{remark}[Use of the Moore–Penrose Pseudoinverse]
The notation $\left[ A \right]^+$ denotes the Moore–Penrose pseudoinverse of a matrix $A$. This generalized inverse ensures the recursive equations are well-defined even when $A$ is singular or not full rank. It satisfies the following identities:
\[
A [A]^+ A = A, \quad [A]^+ A [A]^+ = [A]^+, \quad (A [A]^+)^\top = A [A]^+, \quad ([A]^+ A)^\top = [A]^+ A.
\]
The pseudoinverse provides the minimum-norm solution in least-squares problems and is particularly advantageous in state estimation for models with low-rank noise structure or ill-conditioned observation matrices. For example, in Zika transmission models where the observation matrix $\ell$ may be low-rank due to sparse or aggregated reporting, the pseudoinverse guarantees numerical stability of the filter.
\end{remark}

\begin{remark}[Online Covariance Updates in Nonlinear Systems]
In contrast to the classical Kalman filter for linear time-invariant systems, where the covariance update equation (a Riccati equation) can be solved offline, the filtering equations \eqref{eq:EKF_mean_update}--\eqref{eq:EKF_cov_update} require real-time updates. This is because the model coefficients depend explicitly on the observation path $\Zpathn$, and thus the conditional covariance $\condvar_n$ is non-deterministic and observation-driven. This online updating is essential in real-time epidemic tracking of diseases such as Zika, where model parameters and surveillance data evolve dynamically over time.
\end{remark}

\subsection{Extended Kalman Filter}
\label{sec:EKF}

In the context of Zika virus transmission dynamics, the underlying epidemic processes are inherently nonlinear due to the complex interactions between human and mosquito populations, seasonally varying transmission rates, and the stochastic nature of disease spread. These nonlinearities manifest both in the deterministic drift and in the stochastic components of the system. Consequently, classical linear filtering techniques such as the standard Kalman filter are inadequate for accurately estimating the unobservable epidemic states.

To address this challenge, we employ the \emph{Extended Kalman Filter (EKF)}, a recursive filtering technique tailored for nonlinear stochastic systems. The EKF generalizes the Kalman filter framework to accommodate nonlinear state dynamics by performing a first-order linearization of the system at each discrete time step. This approach allows for tractable, approximate inference of hidden states even when the underlying model deviates from linearity. Foundational references for this methodology are included in \cite{Burkholder}, and \cite{Bain}.

In the Zika epidemic setting, the drift function $f$ captures the nonlinear transmission dynamics influenced by vector-host interactions, while the diffusion terms $\sigma$ and $\ell$ may depend on $Y_n$ to reflect state-dependent variability such as seasonally fluctuating mosquito populations or reporting delays. This model structure induces a filtering problem that is nonlinear both in the system dynamics and in the observation equation.

To implement the EKF, we proceed as follows: at each time step $n$, the drift function $f$ is linearized via a first-order Taylor expansion about a reference point $\overline{Y}_n$, which is typically chosen to be the current state estimate $\condmean_n = \mathbb{E}[Y_n|\mathcal{F}_n^Z]$. The linearized form of $f$ is given by
\[
f(Y_n) \approx  f(n, \overline{y}, z) + \nabla_y f(n, \overline{y}, z)(y - \overline{y}),
\]
where $\nabla_y f$ denotes the Jacobian matrix of $f$ evaluated at $\overline{Y}_n$. Similarly, the diffusion matrices $\sigma(Y_n)$ and $\ell(Y_n)$ are approximated by their evaluations at $\overline{Y}_n$, i.e.,
\[
\sigma(Y_n) \approx \sigma(\overline{Y}_n), \quad \ell(Y_n) \approx \ell(\overline{Y}_n),
\]
as proposed in \cite{Picard}. These approximations yield a locally linear Gaussian model, to which the Kalman filtering recursion for conditional Gaussian sequences (Theorem~\ref{Theo_EKF}) can be applied.

\subsection{Approximation by Conditional Gaussian Sequences}
\label{ACGS}

In this section, we construct an approximate formulation of the nonlinear state-observation system described by the recursions \eqref{state_YZ_Zika}, with the aim of recasting it into a form suitable for applying Kalman filtering results for conditionally Gaussian models. Specifically, we seek to approximate the nonlinear dynamics of the state process \( (Y_n) \) and the observation process \( (Z_n) \) by a system of linear recursions of the form \eqref{kfcgm2}.

Our method is based on a local linearization of the drift function \( f \) with respect to the hidden signal \( Y_n \), performed via a first-order Taylor expansion around a carefully chosen reference point \( \overline{Y}_n \) at each discrete time step. In the Zika transmission models considered, the drift of the observation process is already linear with respect to \( Y_n \), and thus does not require any further approximation. Accordingly, we restrict attention to models where the observation drift satisfies the affine structure:
\begin{align}
\label{h_linear}
h(n,y,z) = h_0(n,z) + h_1(n,z)y.
\end{align}

Denoting the reference point at time \( n \) by \( \overline{Y}_n = \overline{y} \), the drift term \( f(n, y, z) \) is approximated by its first-order Taylor expansion:
\begin{align}
\label{linearization_f}
f(n, y, z) \approx f(n, \overline{y}, z) + \nabla_y f(n, \overline{y}, z)(y - \overline{y}),
\end{align}
where \( \nabla_y f \) denotes the Jacobian matrix of \( f \) with respect to \( y \). In addition, signal-dependent diffusion coefficients \( \sigma, g, \ell \) are approximated by evaluating them at the reference point \( \overline{Y}_n \), thus freezing their dependence on the unknown state.

This procedure leads to the following approximated state-observation system:

\begin{lemma}
\label{lem:linearized_system}
Let the system \eqref{state_YZ_Zika} satisfy the linear observation drift condition \eqref{h_linear}. Then the linearized approximation of the nonlinear dynamics, based on the expansion \eqref{linearization_f} and freezing of the diffusion coefficients at the reference point \( \overline{Y}_n \), yields the following recursions:
\begin{align}
\label{ekf_5}
\begin{split}
\widetilde{Y}_{n+1} &= f_0(n, \overline{Y}_n, \widetilde{Z}_n) + f_1(n, \overline{Y}_n, \widetilde{Z}_n) \widetilde{Y}_n + \sigma(n, \overline{Y}_n, \widetilde{Z}_n) \mathcal{E}^1_{n+1} + g(n, \overline{Y}_n, \widetilde{Z}_n) \mathcal{E}^2_{n+1}, \\
\widetilde{Z}_{n+1} &= h_0(n, \widetilde{Z}_n) + h_1(n, \widetilde{Z}_n)\widetilde{Y}_n + \ell(n, \overline{Y}_n, \widetilde{Z}_n)\mathcal{E}^2_{n+1}, \\
\widetilde{Y}_0 &= y, \quad \widetilde{Z}_0 = z.
\end{split}
\end{align}
Here, the functions \( f_0 \) and \( f_1 \) are defined as:
\begin{align}
f_0(n, \overline{y}, z) = f(n, \overline{y}, z) - \nabla_y f(n, \overline{y}, z)\overline{y}, \quad f_1(n, \overline{y}, z) = \nabla_y f(n, \overline{y}, z).
\end{align}
\end{lemma}

\begin{proof}
Substituting the linearized drift \eqref{linearization_f} and evaluating the diffusion coefficients at \( \overline{Y}_n \), the dynamics of \( \widetilde{Y}_n \) become:
\begin{align}
\label{ekf_3}
\widetilde{Y}_{n+1} &= f(n, \overline{Y}_n, \widetilde{Z}_n) + \nabla_y f(n, \overline{Y}_n, \widetilde{Z}_n)(\widetilde{Y}_n - \overline{Y}_n) + \sigma(n, \overline{Y}_n, \widetilde{Z}_n)\varepsilon^1_{n+1} + g(n, \overline{Y}_n, \widetilde{Z}_n)\varepsilon^2_{n+1}.\quad \quad \quad
\end{align}
Similarly, the observation recursion becomes:
\begin{align}
\label{ekf_4}
\widetilde{Z}_{n+1} &= h_0(n, \widetilde{Z}_n) + h_1(n, \widetilde{Z}_n)\widetilde{Y}_n + \ell(n, \overline{Y}_n, \widetilde{Z}_n)\varepsilon^2_{n+1}.
\end{align}
Rewriting \eqref{ekf_3} using the identities that define \( f_0 \) and \( f_1 \) produces the recursion \eqref{ekf_5}. \qed
\end{proof}

The system \eqref{ekf_5} now takes the form of a linear state-space model with Gaussian innovations, conditional on the current approximating point \( \overline{Y}_n \). As such, it is amenable to analysis via the Kalman filter for conditional Gaussian sequences, as described in Theorem~\ref{Theo_EKF}. 

In practice, the EKF method iteratively selects the reference point \( \overline{Y}_n \) at each time step as the current estimated mean of the hidden state $\overline{Y}_n := \condmeanEKF_n.$.Moreover, the actual observed sequence \( (Z_n) \) is interpreted as being generated by the observation equation in \eqref{ekf_5}, allowing the recursive computation of the filter approximations \( (\condmeanEKF_n), (\condvarEKF_n) \). The full algorithmic implementation of this approximation procedure is given in Algorithm \ref{Algo_EKF}

\begin{algorithm}[!ht] 
	\DontPrintSemicolon
	
	\KwIn{$Z_{0},\ldots,Z_{N_{t}}$; model parameters,   prior information $\Fprior$ } 
	\KwOut{Approximations $\condmeanEKF_n, \condvarEKF_n$ of $\condmean_n:=\mathbb{E}[Y_n|\mathcal{F}_{n}^{Z}]$ and $\condvar_n:=\var(Y_n|\mathcal{F}_{n}^{Z})$ for $n=0,\ldots,N_t$}

    \vspace{0.2cm}
	\textbf{Initialization :}~~$n:=0$, \quad   $\condmeanEKF_0:=\condmean_0=\mathbb{E}[Y_0|\mathcal{F}_{0}^{Z}]$,  $\condvarEKF_0:=\condvar_0=\var(Y_0|\mathcal{F}_{0}^{Z})$ 
	\begin{enumerate}
		\item[(i)] State prediction 
		\begin{align}
			\begin{split}
				\condmeanEKF_{n+1}=   f_0+f_1  \condmeanEKF_{n}   ~+ ~\big( g\ell^\top+f_1 \condvarEKF_{n} h_1^\top \big)\big[\ell\ell^\top + h_1 \condvarEKF_{n} h_1^\top\big]^{+} \nonumber  \big( {\widetilde{Z}_{n+1}}- \big( h_0+h_1 \condmeanEKF_{n} \big) \big) 
			\end{split}
		\end{align}			
		\item[(ii)] Error measurement
		\begin{align}	
			\condvarEKF_{n+1}=  -\big( g\ell^\top+
			f_1 \condvarEKF_{n} h_1^\top \big)\nonumber \big[\ell\ell^\top + h_1 \condvarEKF_{n} h_1^\top\big]^{+}  \big( g\ell^\top+ f_1 \condvarEKF_{n} h_1^\top  \big)^\top 
			+f_1 \condvarEKF_{n}f_1^\top+ \sigma\sigma^T + gg^{\top} 
		\end{align}
		All coefficient functions are evaluated at the point $(n,\condmeanEKF_n, Z_n)$.
		\item[(iii)] Repeat (i) and (ii) for the next time step until all samples
		are processed.
	\end{enumerate}
	
	\caption{EKF Algorithm \label{Algo_EKF}}
\end{algorithm}

In the EKF approximation of the nonlinear filtering problem, the system coefficients \( \cf_0, \cf_1, \cg, \csigma, \cell \) are defined as functions of both time and the observed data. At each time step \( n \), they are evaluated using the current observation \( Z_n \) and the EKF-estimated signal mean \( \condmeanEKF_n \), according to the rule
\[
\cb(\Zpathn) = \cb((Z_0,\ldots,Z_n)) = b(n, \condmeanEKF_n, Z_n), \quad \text{for } \cb \in \{ \cf_0, \cf_1, \cg, \csigma, \cell \}.
\]
For the coefficients \( \cb = \ch_0, \ch_1 \) appearing in the observation equation, the dependence is restricted to the current observation:
\[
\cb(\Zpathn) = \cb((Z_0,\ldots,Z_n)) = b(n, Z_n).
\]

The EKF recursion is initialized at \( n=0 \) with the conditional distribution of \( Y_0 \), given the initial observation \( Z_0 \) and the prior information \( \Fprior \). This distribution is assumed to be Gaussian, \( \mathcal{N}(\condmean_0, \condvar_0) \), where \( \condmean_0 \) and \( \condvar_0 \) represent the prior mean and covariance matrix, respectively.It is important to note that the EKF estimate \( \condmeanEKF_n \) is defined recursively and its computation requires access to the full observation path \( \Zpathn \). 

\begin{remark}\label{rem_error_analysis_summary}
Under suitable smoothness and regularity conditions on the model coefficients \( f \) and \( h \), the EKF yields a valid first-order approximation of the true filter associated with the nonlinear system. For rigorous error analysis and convergence guarantees, we refer to the results in \cite{Picard1} and their extension in the continuous-time setting in \cite{Njiasse2}.
\end{remark}

\section{Numerical Results}
\label{sec:NumericalResults}

In this section, we present numerical experiments to illustrate the performance of the discrete-time stochastic Zika transmission model introduced earlier, particularly the version incorporating cascade compartments for temporary immunity and hidden states. The simulations are designed to demonstrate the impact of partial observability and immune waning on the system dynamics, as well as to assess the identifiability of key hidden compartments under noisy observations.

We recall that the extended Zika model consists of thirteen hidden compartments collected in the state vector  
\[
Y = (I_m^{-}, E_m, R_{m}^{1-}, R_{m}^{2-}, S_m, I_f^{-}, E_f, R_{f}^{1-}, R_{f}^{2-}, S_f, I_v, E_v, S_v)^{\top},
\]
and ten observable components forming the observation vector
\[
Z = (I_m^+, I_m^s, I_f^+, I_f^s, R_m^1, R_m^2, R_m^3, R_f^1, R_f^2, R_f^3)^{\top}.
\]

\subsection{Model parameters}
\label{sec:modelparameters}

The reproductive dynamics of Aedes aegypti mosquitoes, the primary vectors for Zika virus transmission, are strongly shaped by environmental conditions, particularly rainfall and humidity. Rainfall creates standing water that serves as breeding sites, thereby increasing mosquito birth rates. To account for this ecological effect, we introduce a time-dependent mosquito birth rate \( B_v(t) \) into the population dynamics.

\begin{figure}[ht]
	\centering
	\includegraphics[width=0.9\textwidth]{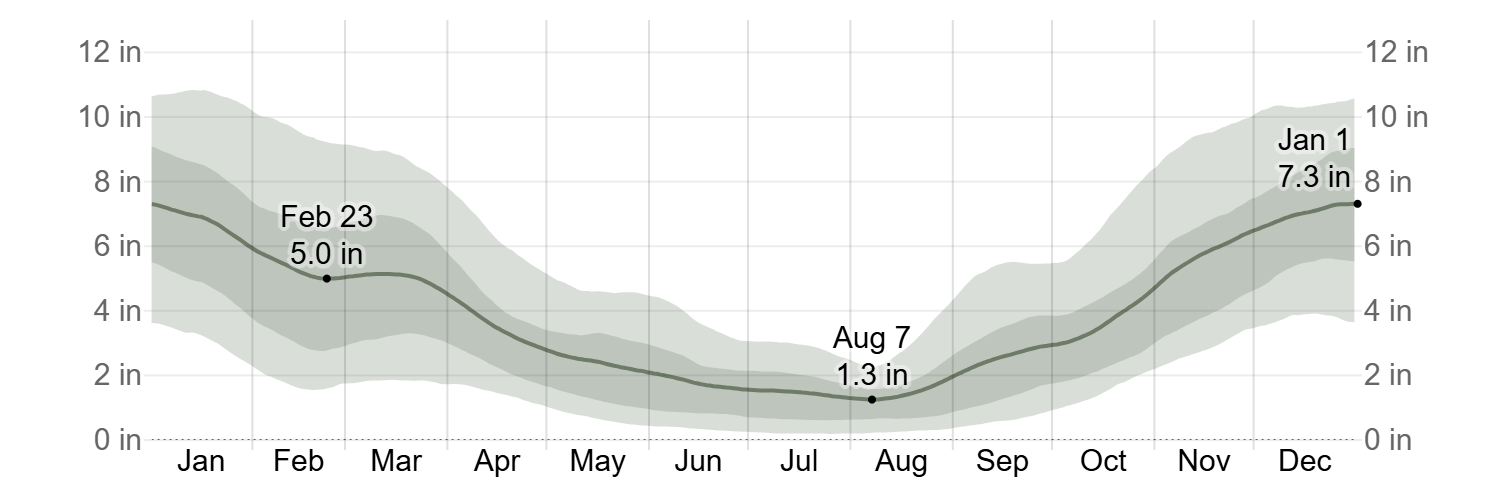 }
	\caption{Monthly rainfall pattern in Rio de Janeiro, Brazil, showing clear seasonality with peaks during the summer months (December–March) and lows in the winter (June–August)${}^{\ref{fn_weather}}$. The solid line represents the mean rainfall computed over a 31-day sliding window, while the shaded regions indicate empirical percentile bands: the darker band corresponds to the interquartile range (25th–75th percentiles), and the lighter band spans the 10th–90th percentiles.
	}
	\label{fig:rainfall_rio}
\end{figure}

Figure~\ref{fig:rainfall_rio} illustrates the seasonal rainfall pattern in Rio de Janeiro, chosen here as a representative case study due to its tropical climate and its central role in the 2015–2016 Zika epidemic. The city exhibits strong seasonal fluctuations in temperature, humidity, and precipitation, which directly affect mosquito life cycles. Rainfall, in particular, plays a critical role by generating standing water that accelerates mosquito reproduction. Incorporating this rainfall-driven seasonality into the model is therefore essential to capture vector population dynamics and improve epidemic estimates under partial observations.

The rainfall profile peaks in January–February, averaging 7–8 inches per month, and reaches a minimum in July–August, with averages below 2 inches. Accordingly, we assume the mosquito birth rate \( B_v(t) \) is strongly correlated with seasonal rainfall. Instead of using a purely parametric cosine-type function, we adopt a data-driven approach based on historical monthly rainfall records from Rio de Janeiro.\footnote{\label{fn_weather}\scriptsize\url{https://weatherspark.com/y/30563/Average-Weather-in-Rio-de-Janeiro-Brazil-Year-Round}} This provides a more realistic calibration of mosquito abundance.

To incorporate this information, we proceed as follows:
\begin{itemize}
	\item Monthly average rainfall (in inches) is converted to millimeters (1 inch = 25.4 mm).
	\item The rainfall values are transformed into the interval \( (a, \bar{B}_v) \subset (0,1) \), with \( a = 0.01 \) as a positive lower bound and \( \bar{B}_v = 0.1 \) as the biologically (assumed) plausible upper bound for the mosquito birth rate. This ensures that the birth rate term  \( B_v(t)(S_v+E_v+I_v) \) remains stable and interpretable.
	\item The normalized monthly values are extended periodically and interpolated using a cubic spline to obtain a smooth daily function \( B_v(t) \).
\end{itemize}

The normalization is given by 
\[
B_v({t_i}) = a + (\bar{B}_v - a)\cdot \frac{R_i - \min(R)}{\max(R) - \min(R)},
\]
where \( R = \{R_1, R_2, \dots, R_{12}\} \) denotes the monthly average rainfall (mm), \( a=0.01 \), and \( \bar{B}_v=0.1 \). This rescaling ensures \( B_v(t) \) remains within  reasonable limits while reflecting seasonal variability.

This calibrated function \( B_v(t) \) is used in the numerical experiments that follow to simulate mosquito population dynamics and their effect on Zika transmission under partial information. Parameter values and their epidemiological interpretations are summarized in Table~\ref{table : Extend_Zika_Param}.

\begin{figure}[ht]
	\centering
	\subfloat[  Transformed mosquito birth rate \( B_v(t) \) based on monthly rainfall in Rio de Janeiro over one year.\label{fig:birth_rate_1year}]{
		\includegraphics[width=0.75\textwidth]{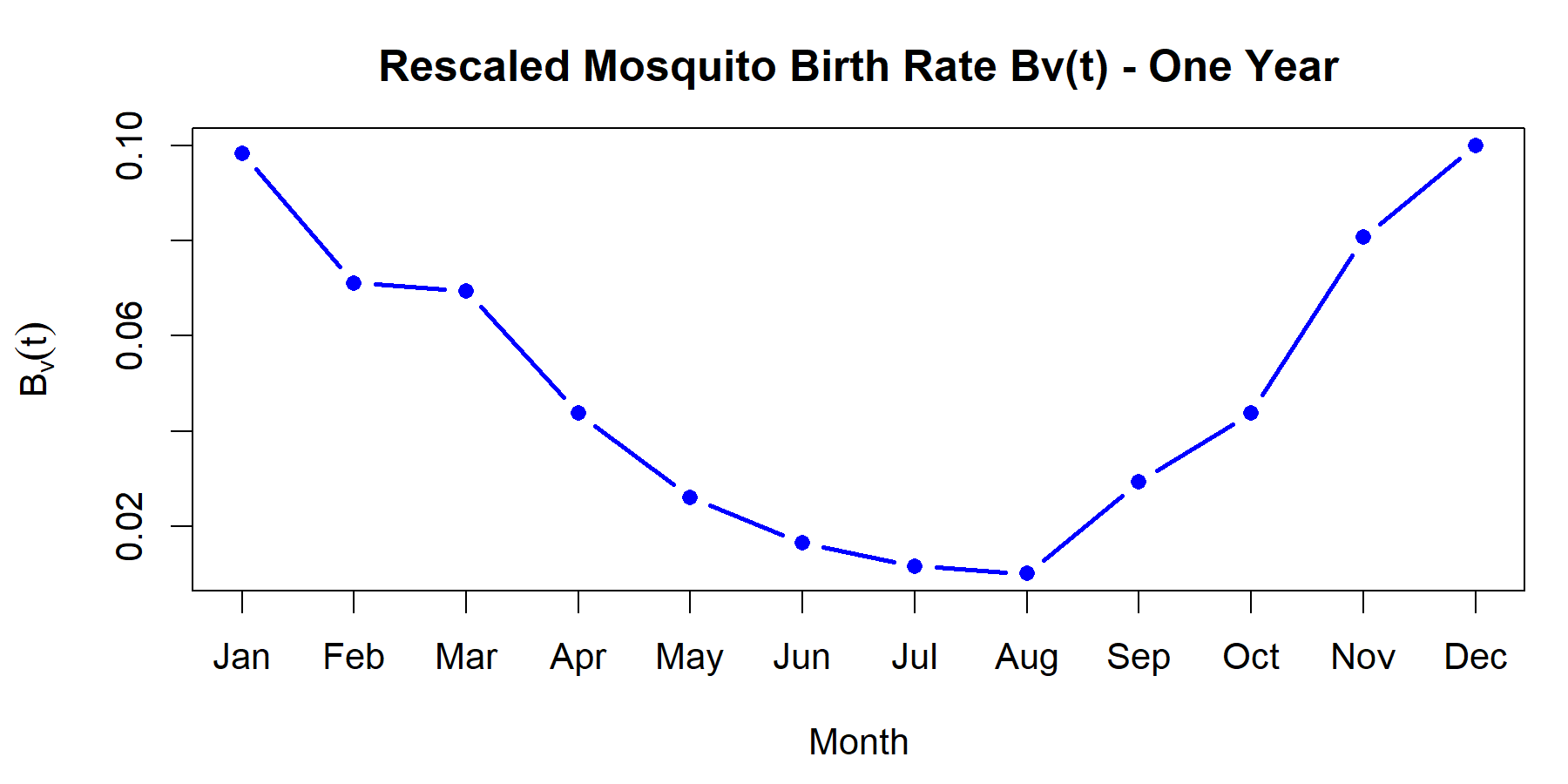}
	}
	
	\vspace{-2ex}
	
	\subfloat[Smoothed mosquito birth rate \( B_v(t) \) extended over four years using periodic spline interpolation.\label{fig:birth_rate_4year}]{
		\includegraphics[width=0.75\textwidth]{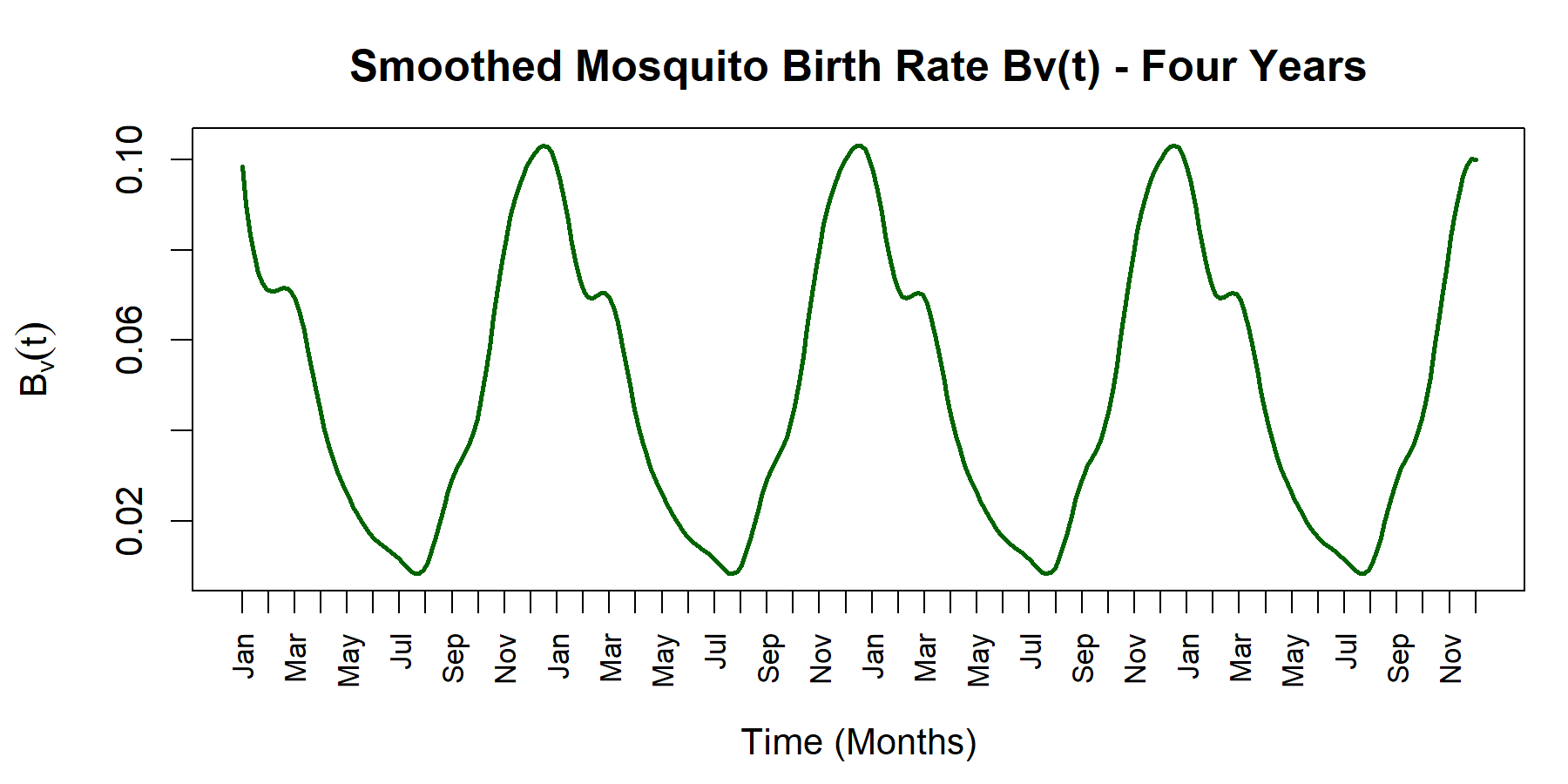}
	}
	
	\caption{Calibration of the mosquito birth rate \( B_v(t) \) from  transformed rainfall data in Rio de Janeiro. The function reflects seasonal variability in mosquito reproduction driven by environmental conditions.}
	\label{fig:calibration_birth_rate}
\end{figure}


{\small 
\begin{table}[ht]
	\centering
	\begin{tabular}{lll}
		\toprule
		\textbf{Parameter} & \textbf{Description} & \textbf{Value}  \\
		\midrule
		\( \alpha_m \) & Male infection rate & 1.5   \\
		\( \alpha_f \) & Female infection rate & 2.5  \\		
		\( \beta_m^-, \beta_m^+, \beta_m^s \) & Transition from exposed to \( I^- \), \( I^+ \), \( I^s \) (male) & 0.005, 0.005, 0.004  \\
		\( \beta_f^-, \beta_f^+, \beta_f^s \) & Transition from exposed to \( I^- \), \( I^+ \), \( I^s \) (female) & 0.005, 0.005, 0.004  \\
		\( \gamma_m^-, \gamma_f^- \) & Recovery from \( I^- \) & 0.1  \\
		\( \gamma_m^+, \gamma_f^+ \) & Recovery from \( I^+ \) & 0.1  \\
		\( \gamma_{m}^s, \gamma_{f}^s \) & Recovery from symptomatic \( I^s \) & 0.1  \\
		\( \rho_m^1, \rho_m^2, \rho_f^3, \rho_f^4 \) & Immunity loss (recovered) & 1/760,1/30,1/760,1/30 \\
		$\halfsat$ & half-saturation parameter & 200,000\\
		\( \theta \) & Exposure rate of susceptible vectors & 0.02 \\
		\( \nu \) & Infection rate of exposed vectors & 0.10  \\
		\( \mu \) & Vector death rate & 0.05  \\
		\( B_v(t) \) & Vector birth rate (seasonal) & calibrated  \\
		\( \omega \) & Male-to-female transmission probability & 0.05 \\
        \( N_m , N_f, N_{v,0} \) & Population of males, females, and vectors & 10{,}000, 10{,}000, 100{,}000 \\
				\( \vartheta_m^{1,2,3}, \vartheta_f^{1,2,3} \) & Transition between cascade states & 3/730  \\
		\midrule
		\(  \dfcdmean^I, \dfcdmean^E,\dfcdmean^R \) &  DFC Mean of $I, E, R$  & 7,10,7  \\
		\(  \dfcdvar^I, \dfcdvar^E, \dfcdvar^R  \) &  DFC Variance of $I, E,  R$ & 1,1,1 \\
		\( \eccmean^{S},\eccmean^{E},\eccmean^{I}  \) &  Mean  of $S_{v,0}, E_{v,0} I_{v,0} $  & 89,000,\  6,000,\  5,000 \\
		\( \eccvar^{S},\eccvar^{E},\eccvar^{I}  \) &  Variance   of $S_{v,0}, E_{v,0} I_{v,0} $  & $9,500^2,\  800^2,\  700^2$ \\
		\bottomrule
	\end{tabular}
	\caption{Model Parameters for Zika Transmission Dynamics}
	\label{table : Extend_Zika_Param}
\end{table}
}

\subsection{Initialization of the Extended Zika Model}

We present a strategy for initializing the extended Zika epidemic model under partial observation, tailored to the epidemiological context of the 2015–2016 Brazilian outbreak. The model distinguishes between a hidden state vector $Y = (Y^1, Y^2, \dots, Y^{13})^\top$, comprising unobserved compartments such as undetected infections, exposures, undetected recoveries, and susceptibles across male, female, and mosquito populations, and an observable vector $Z = (Z^1, Z^2, \dots, Z^{10})^{\top}$, which includes reported symptomatic and asymptomatic infections as well as confirmed recoveries.

A major challenge in modeling Zika dynamics is the substantial underreporting of cases, largely due to mild or asymptomatic infections and limited diagnostic capacity. To account for this, we introduce dark figure coefficients (DFCs),  which constructs the initial filter estimates $\condmean_0,\condvar_0$. These coefficients are informed by epidemiological evidence and allow us to incorporate prior knowledge about the hidden epidemic burden into the initialization.
 The true initial values of the hidden state $Y_0$ are drawn from the Gaussian distribution \( \mathcal{N}(\condmean_0, \condvar_0). \) 
 
 For each detected infected male ($I_m^+$), we assume approximately four undetected infections ($I_m^-$) and five exposures ($E_m$). Similarly, for each observed infected female ($I_f^+$), we assume three undetected infections ($I_f^-$) and four exposures ($E_f$). For recoveries, each observed case in stage one or two is associated with two undetected recoveries in the same stage. These assumptions provide a structured and epidemiologically plausible initialization of the hidden compartments.


Applying the dark figure coefficients to the initial observations provides consistent estimates for the hidden compartments, as summarized in Table~\ref{tab:initial_state_estimates}. These estimates define the initial values of the hidden state vector $Y$ at $t=0$ and serve as the prior for simulating the Zika epidemic dynamics. The resulting trajectories are then used to assess the performance of the proposed filtering method.

{  \begin{table}[h!]
	\centering
	
	\begin{tabular}{|c|c|c||c|c|}
		\hline
		\multicolumn{5}{|c|}{\textbf{Hidden States}}\\\hline
        Index & State Variable &{Initial value} & $M_0^{(i)}$ (Initial Mean) & $(Q_{0}^{ii})^2$ (Initial Variance) \\
		\hline
        1  & $I_m^{-}$ & 80  & $210$      & $30^2$ \\
			2  & $E_m$  &  100    & $300$      & $30^2$ \\
			3  & $R_m^{1-}$& 36   & $126$      & $18^2$ \\
			4  & $R_m^{2-}$ & 32  & $0$       & $0$ \\
			5  & $S_m$ &   9{,}752    & $9{,}364$ & $30^2 + 18^2$ \\
			\midrule
			6  & $I_f^{-}$ & 75  & $259$      & $37^2$ \\
			7  & $E_f$  &   100   & $370$      & $30^2$ \\
			8  & $R_f^{1-}$ & 36  & $126$      & $18^2$ \\
			9  & $R_f^{2-}$ & 32  & $0$       & $0$ \\
			10 & $S_f$ &  9{,}757     & $9{,}245$ & $37^2 + 17^2$ \\
			\midrule
			11 & $I_v$ &  2{,}000     & $5{,}000$   & $700^2$ \\
			12 & $E_v$ &  3{,}000     & $6{,}000$   & $800^2$ \\
			13 & $S_v$ &   95{,}000    & $89{,}000$  & $9500^2$ \\	
		\hline
	\end{tabular}
	\\[2ex]
	\begin{tabular}{|c|c|c||c|c|c|}
		\hline
		\multicolumn{6}{|c|}{\textbf{Observable States}}\\\hline
		$i$& {Variable} & {Initial value}& $i$& {Variable} & {Initial value} \\
		& $Z^{i}$ & $Z_{0}^{i}$ &  &  $Z^{i}$ & $Z_{0}^{i}$\\
		\hline
		$1$ & \( I^{+}_m \)  & 10  & 6 &  \( R^{2}_m \) &  0 \\
		$2$ & \( I^{s}_m  \)      & 20  & 7 &  \( R^{3}_m \) &  0 \\
		$3$ &\(I^{+}_f \)        & 12  & 8 &  \(R^{1}_f \) &   18 \\	
		$4$ & \( I^{s}_f \)      &  25  & 9 &  \( R^{2}_f \) &  0 \\
        $5$ & \( R^{1}_m \)      &   18  & 10 &  \( R^{3}_f \) &  0 \\	
			
		\hline
	\end{tabular}	
	\caption{Initial conditional expectations and variances for hidden state variables, consistent with Proposition~\ref{prop:init_law}, used to construct $M_0$ and the diagonal blocks of $Q_0$.
		\label{tab:initial_state_estimates}
		Top: Hidden states $Y_1,\ldots,Y_{13}$  with their initial estimates used for filtering. Bottom: Observable states $Z_1,\ldots,Z_{4}$	}
	\label{tab:init_vals}
\end{table}}

This Table (\ref{tab:initial_state_estimates}) reflects the initialization of the hidden state vector $Y_0$ used to simulate epidemic progression in a population. The dark-figure coefficient for compartments such as $I_m^-$, $E_m$, $I_f^-$, and $E_f$ account for unreported cases. Recovered individuals are similarly adjusted to reflect potential underreporting. The susceptible compartments are calculated by removing all known and estimated compartments from the total. The vector compartments ($I_v$, $E_v$, $S_v$) are set using biologically plausible proportions to maintain realistic initial transmission pressure.

\subsection{Initialization of the Filter Estimate and Covariance Matrix}
\label{sec:InitializationFilter}
Computing the filter estimates requires the initialization of the filter processes  $ \condmeanEKF, \condvarEKF $ at time $n=0$.  
In the following, we present for the extended stochastic Zika model with partial observations how  $ \condmeanEKF_0$  and $\condvarEKF_0 $
can be constructed from the prior information contained in $\Fprior$ and the first observation $Z_0$.  
Analogously to the Covid-19 case described in \cite{Kamkumo2}  Section 2.3, the initialization is not yet affected by linearization errors that occur for $n > 0$ in the EKF approximation of the ``true'' filter processes $\condmean_n,\condvar_n$. Therefore, we omit the tilde and simply write $\condmean_0,\condvar_0$ in the sequel.

\paragraph{Dark Figure Coefficients (DFC)}
We employ DFC to integrate expert knowledge on the ratio of undetected to detected cases into the initialization.  
Using the placeholder notation $\dagger=m,f$ for male and female compartments, we define  
\begin{align}\label{def:DFC_Zika}
	\dfc_{n,\dagger}^{I} &= \frac{I_{\dagger,n}^{-}}{I_{\dagger,n}^+ + I_{\dagger,n}^s}, \quad   
	\dfc_{n,\dagger}^{E} = \frac{E_{\dagger,n}}{I_{\dagger,n}^+ + I_{\dagger,n}^s},\quad 
	\dfc_{\dagger,n}^{R} = \frac{R_{\dagger,n}^{1-}}{R_{\dagger,n}^{1}  + R_{\dagger,n}^2 + R_{\dagger,n}^3   }
\end{align}
so that $\dfc_{\dagger,n}^{I}$ denotes the number of undetected infected  per detected or symptomatic infected in the male ($\dagger=m$) and female ($\dagger=f$) population. Similarly,  $\dfc_{\dagger,n}^{E}$ represents the number of  exposed  per  detected or symptomatic infected. The DFC $\dfc_{\dagger,n}^{R}$ relates 
the recovered individuals with waning immunity to the observable recovered in the cascade states.  

Throughout the following, the symbol $\#$ is used as a generic placeholder that may refer to any of the compartments $I$, $E$, or $R$, respectively.

	\begin{assumption}\label{ass:dfc_zika}
		\quad At initial time $n=0$ we assume:
		\begin{enumerate}
			\item Given the prior information $\Fprior$, the initial DFCs $\dfc_{0,\dagger}^{\#}$, are conditionally independent and Gaussian-distributed:
			\[
			\dfc_{0,\dagger}^{\#} \sim \mathcal{N}(\dfcdmean^{\#}, (\dfcdvar^{\#})^2). 
			\]
			\item $R_{\dagger,0}^{2-}$  start with zero initial values, i.e. $R_{\dagger,0}^{2-}=0.$
			\item The susceptible compartments $S_\dagger$  are determined by normalization within each subpopulation:
			\[
			S_{\dagger,0} = N_{\dagger} - (I_{\dagger,0}^{-} + E_{\dagger,0} + R_{\dagger,0}^{1-} + R_{\dagger,0}^{2-} + I_{\dagger,0}^{+} + I_{\dagger,0}^{s} ).
			\]
			\item Given the prior information $\Fprior$, the initial vector compartments \((I_{v,0}, E_{v,0}, S_{v,0})\) are assumed to be conditionally independent, independent of $Z_0$ 	and Gaussian:
			\[
			I_{v,0}\sim \mathcal N\!\big(\eccmean^{I},(\eccvar^{I})^{2}\big),\qquad
			E_{v,0}\sim \mathcal N\!\big(\eccmean^{E},(\eccvar^{E})^{2}\big),\qquad
			S_{v,0}\sim \mathcal N\!\big(\eccmean^{S},(\eccvar^{S})^{2}\big).
			\]

		\end{enumerate}
	\end{assumption}

\noindent 	
The second assumption, $R_{\dagger,0}^{2-}=0$, is consistent with the epidemiological interpretation of these states. It only considers individuals with waning immunity after they have passed through the cascade states, in which complete immunity is assumed.

Based on this assumption we derive the following initial filter estimates $M_0,Q_0$.
	
	\begin{proposition}\label{prop:init_law}
		For the extended Zika model and under  Assumption~\ref{ass:dfc_zika}, the conditional distribution of the   the initial state $Y_0=(Y_0^1,\ldots,Y_0^{13})$ given $\mathcal{F}_0^Z$, i.e., the prior information $\Fprior$ and the initial observation \(Z_0\) is Gaussian. The   mean and the covariance matrix of this distribution are given by 
		\[
		M_0=\mathbb{E}[Y_0\mid \mathcal{F}_0^Z]=M_0 = \begin{pmatrix}
			M_0^{m} \\[2mm]
			M_0^{f} \\[2mm]
			M_0^{v}
		\end{pmatrix},\quad
		Q_0=\cov(Y_0\mid \mathcal{F}_0^Z) =\footnotesize
		\begin{pmatrix}
			\boxed{Q_0^m} & \mathbf{0}_{5\times 5} & \mathbf{0}_{5\times 3}\\
			\mathbf{0}_{5\times 5} & \boxed{Q_0^f} & \mathbf{0}_{5\times 3}\\
			\mathbf{0}_{3\times 5} & \mathbf{0}_{3\times 5} & \boxed{Q_0^v}
		\end{pmatrix}.
		\]
		Here, 
		{\small 
			\begin{align}
				M_0^{\dagger}&=
				\Big(
				\dfcmean^{I} A_\dagger,\;
				\dfcmean^{E} A_\dagger,\;
				\dfcmean^{R} B_\dagger,\;
				0,\;
				N_\dagger-(\dfcmean^{I} A_\dagger+\dfcmean^{E} A_\dagger +\dfcmean^{R} B_\dagger +A_\dagger)
				\Big)^\top, M_0^{v}=(\eccmean^{I},\,\eccmean^{E},\,\eccmean^{S})^\top,\\
				Q_0^{\dagger}
				&=
				\begin{pmatrix}
					a_I & 0   & 0   & 0 & -a_I\\
					0   & a_E & 0   & 0 & -a_E\\
					0   & 0   & b_R & 0 & -b_R\\
					0   & 0   & 0   & 0 & 0\\
					-a_I & -a_E & -b_R & 0 & a_I+a_E+b_R
				\end{pmatrix}, ~~
				Q_0^{v} =
				\begin{pmatrix}
				(\eccvar^{I})^{2} & 0 & 0\\[2pt]
				0 & (\eccvar^{E})^{2} & 0\\[2pt]
				0 & 0 & (\eccvar^{S})^{2}
				\end{pmatrix},
			\end{align}
		}
		$A_{\dagger}:=I_{\dagger,0}^{+}+I_{\dagger,0}^{s},\quad B_{\dagger}:=R_{\dagger,0}^{1}+R_{\dagger,0}^{2}+R_{\dagger,0}^{3},$ \(\quad C_{\dagger}:=N_{\dagger}-\big(E_{\dagger,0}+I_{\dagger,0}^{+}+I_{\dagger,0}^{s}\big)\)\\
		\(a_I:=A_{\dagger}^2\,(\dfcdvar^{I})^2,\quad
		a_E:=A_{\dagger}^2\,(\dfcdvar^{E})^2,\quad
		b_R:=B_{\dagger}^2\,(\dfcdvar^{R})^2.\)
	\end{proposition}
	
	\begin{proof}
		The result follows by writing each random component as an affine transformation 
		of independent Gaussian variables and applying the standard properties of expectation, variance 
		and covariance. A detailed derivation of the conditional mean and covariance is provided 
		in Appendix~\ref{Proof_details_IFE}.
	\end{proof}

	\subsection*{Initial Values for the Filter Estimate}
	
	The construction of the initial mean vector $M_0$ and the associated conditional covariance matrix $Q_0$ follows Proposition~\ref{prop:init_law}.

	The resulting conditional means and variances for each hidden state variable are summarized in Table~\ref{tab:initial_state_estimates}. 
	The entries are computed directly from the explicit formulas for $M_0^m,M_0^f,M_0^v$ and $Q_0^m,Q_0^f,Q_0^v$ given in Proposition~\ref{prop:init_law}. 
	By construction, the covariance matrix $Q_0$ is block diagonal, encoding both the independence structure across male, female, and vector compartments and the correlations implied within each block through the population balance constraints.

	This initialization provides a structured and interpretable starting point for deploying the extended Kalman filter in the context of stochastic Zika epidemic modeling under partial observation.

	\begin{figure}[h]
		\centering
		
		\begin{minipage}{0.48\textwidth}
			\centering
			\includegraphics[width=\linewidth]{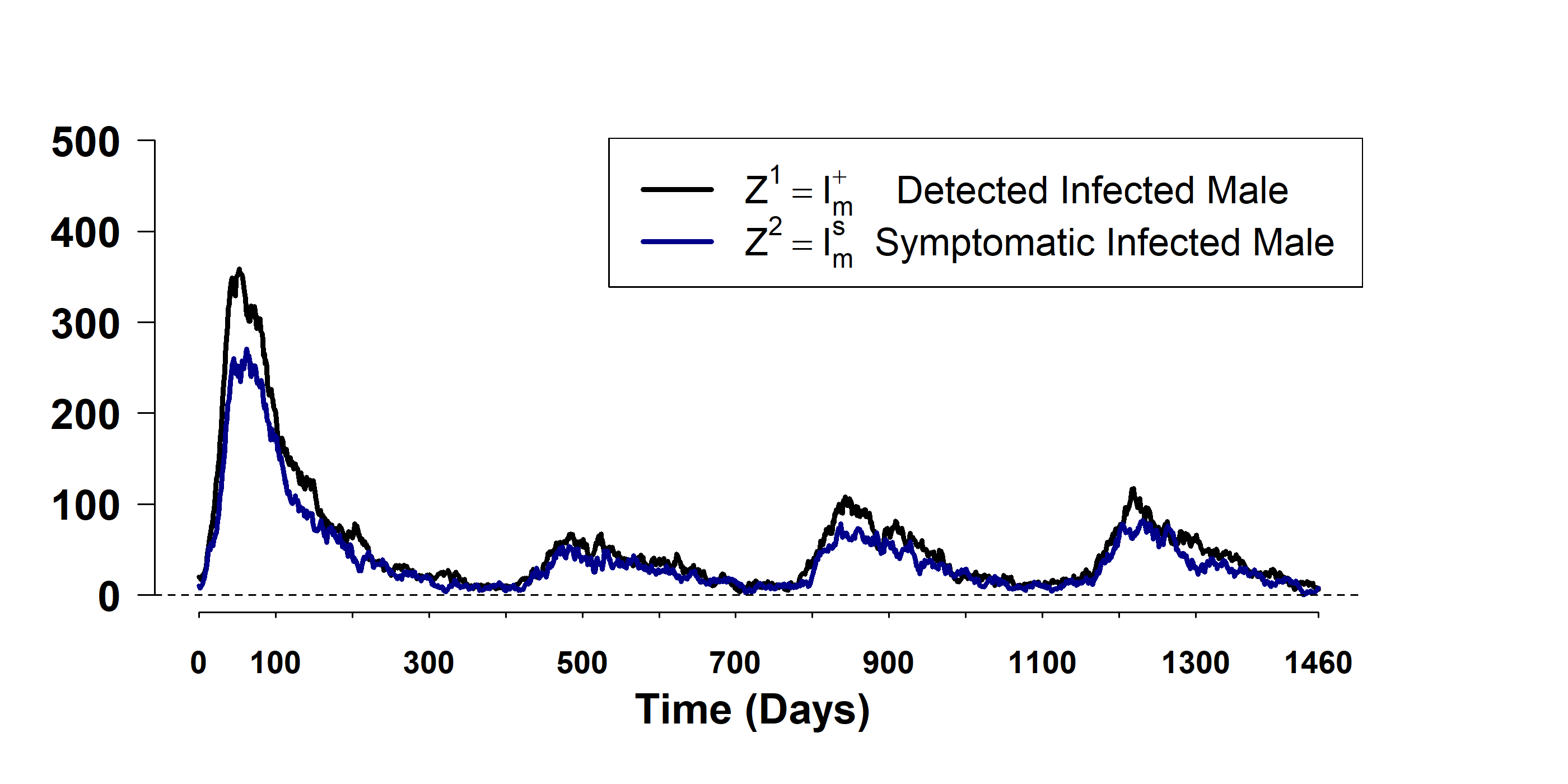}
			\caption*{Detected infected male and detected symptomatic male.}
		\end{minipage}
		\hfill
		\begin{minipage}{0.48\textwidth}
			\centering
			\includegraphics[width=\linewidth]{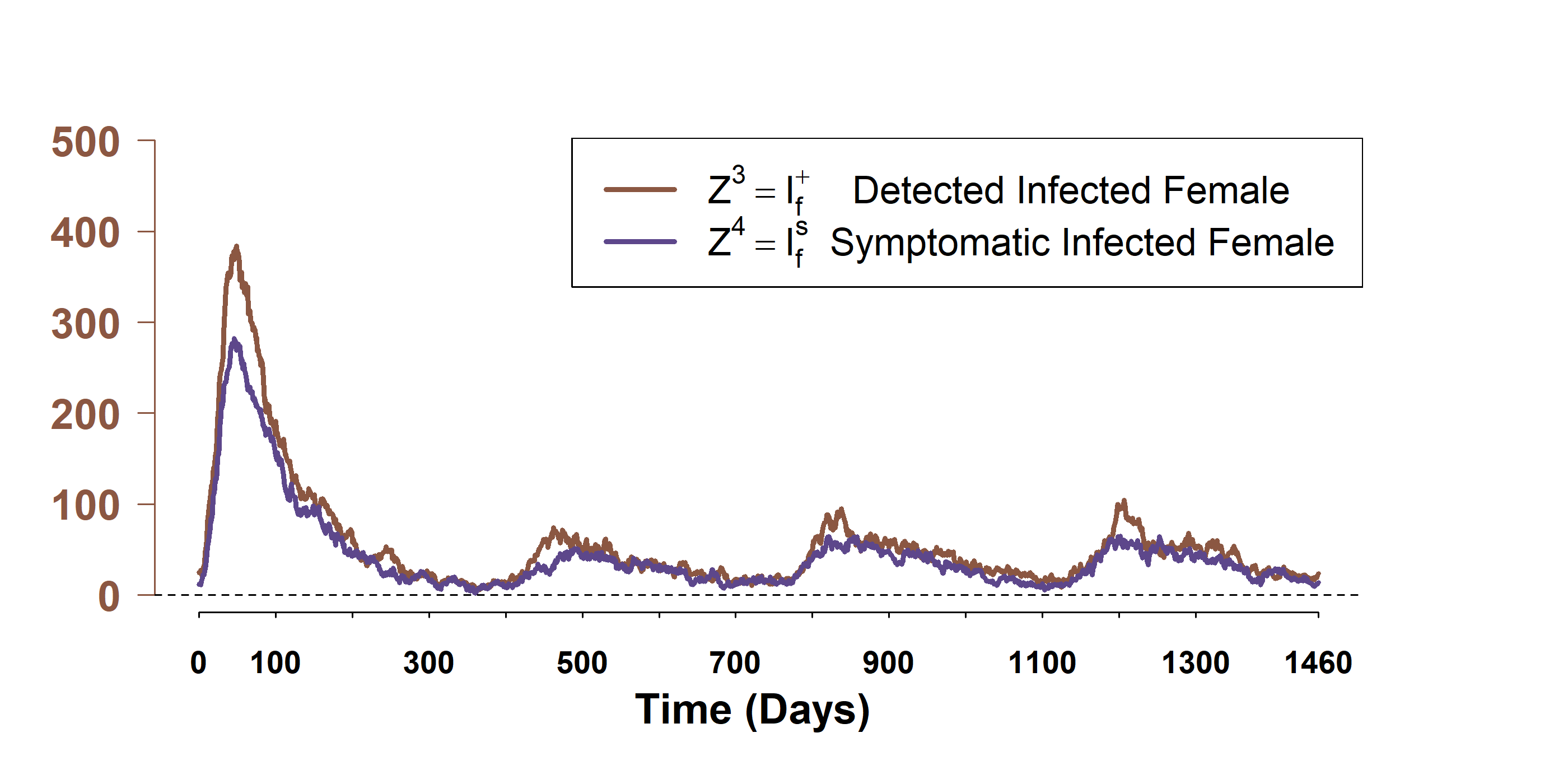}
			\caption*{Detected infected female and detected symptomatic female.}
		\end{minipage}
		
		\vspace{0.3cm}
		
		\begin{minipage}{0.48\textwidth}
			\centering
			\includegraphics[width=\linewidth]{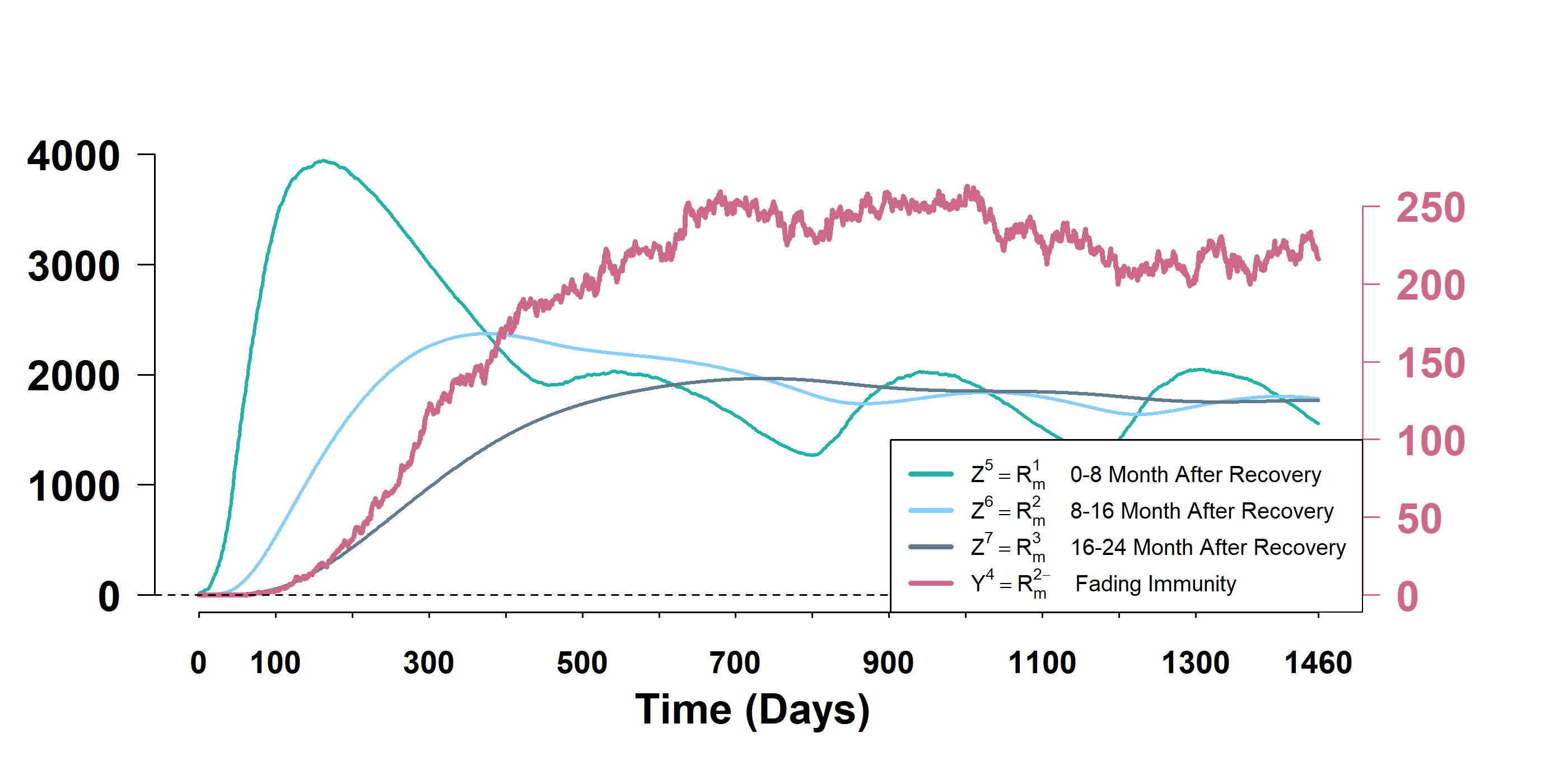}
			\caption*{ Recovery Cascade States for Male .}
		\end{minipage}
		\hfill
		\begin{minipage}{0.48\textwidth}
			\centering
			\includegraphics[width=\linewidth]{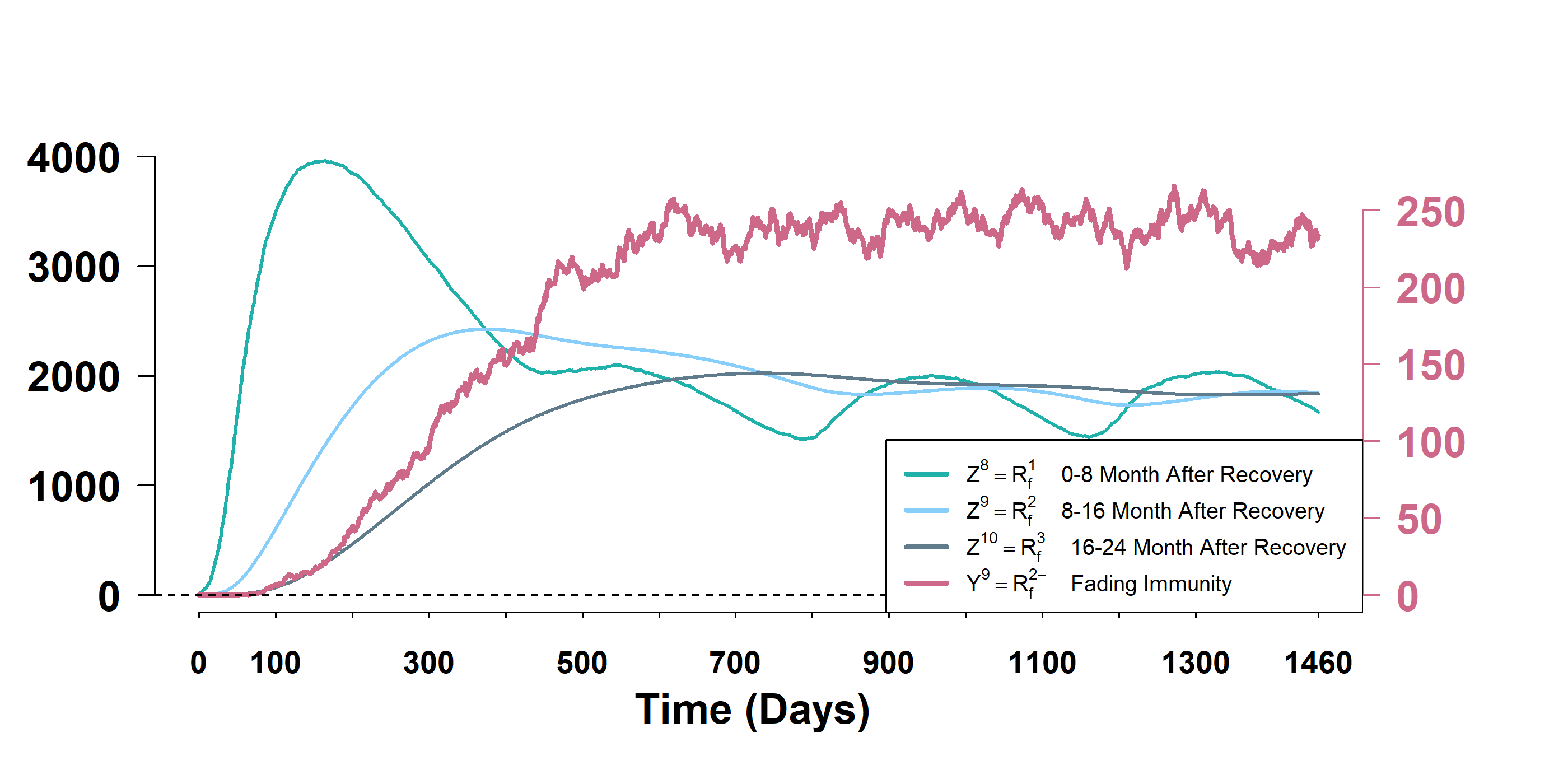}
			\caption*{Recovery Cascade States for Female.}
		\end{minipage}
		
		\caption{Observations of the 10 observable compartments of the extended Zika model. Each figure shows the observable compartment.}
		\label{fig:Observation_estimates_zika}
	\end{figure}

	\begin{figure}[h]
		\centering
		\subfloat[Undetected Infected Male ($I_m^-$)]{
			\includegraphics[width=0.45\textwidth]{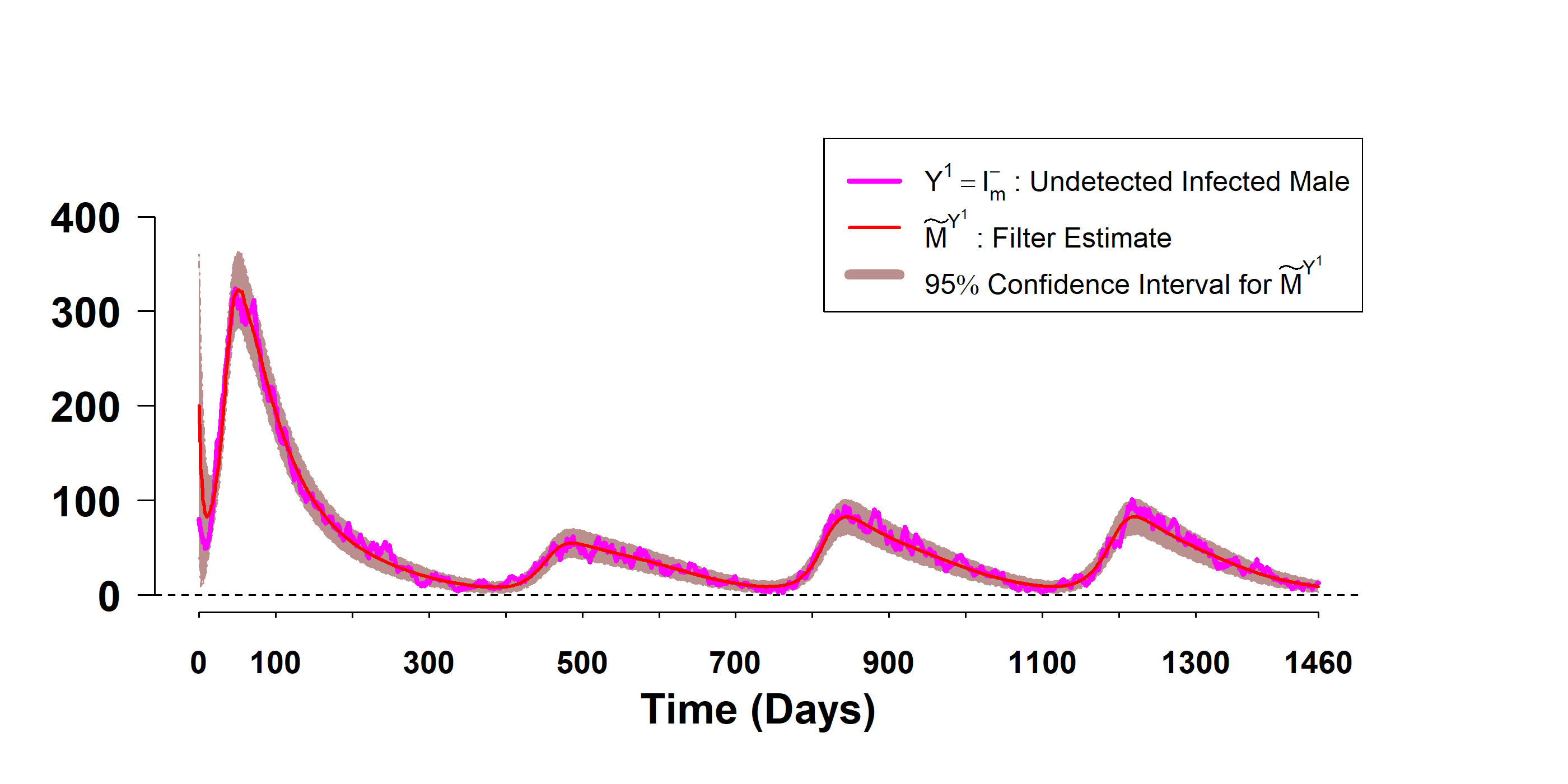}
		}
		\hfill
		\subfloat[Undetected Infected Female ($I_f^-$)]{
			\includegraphics[width=0.45\textwidth]{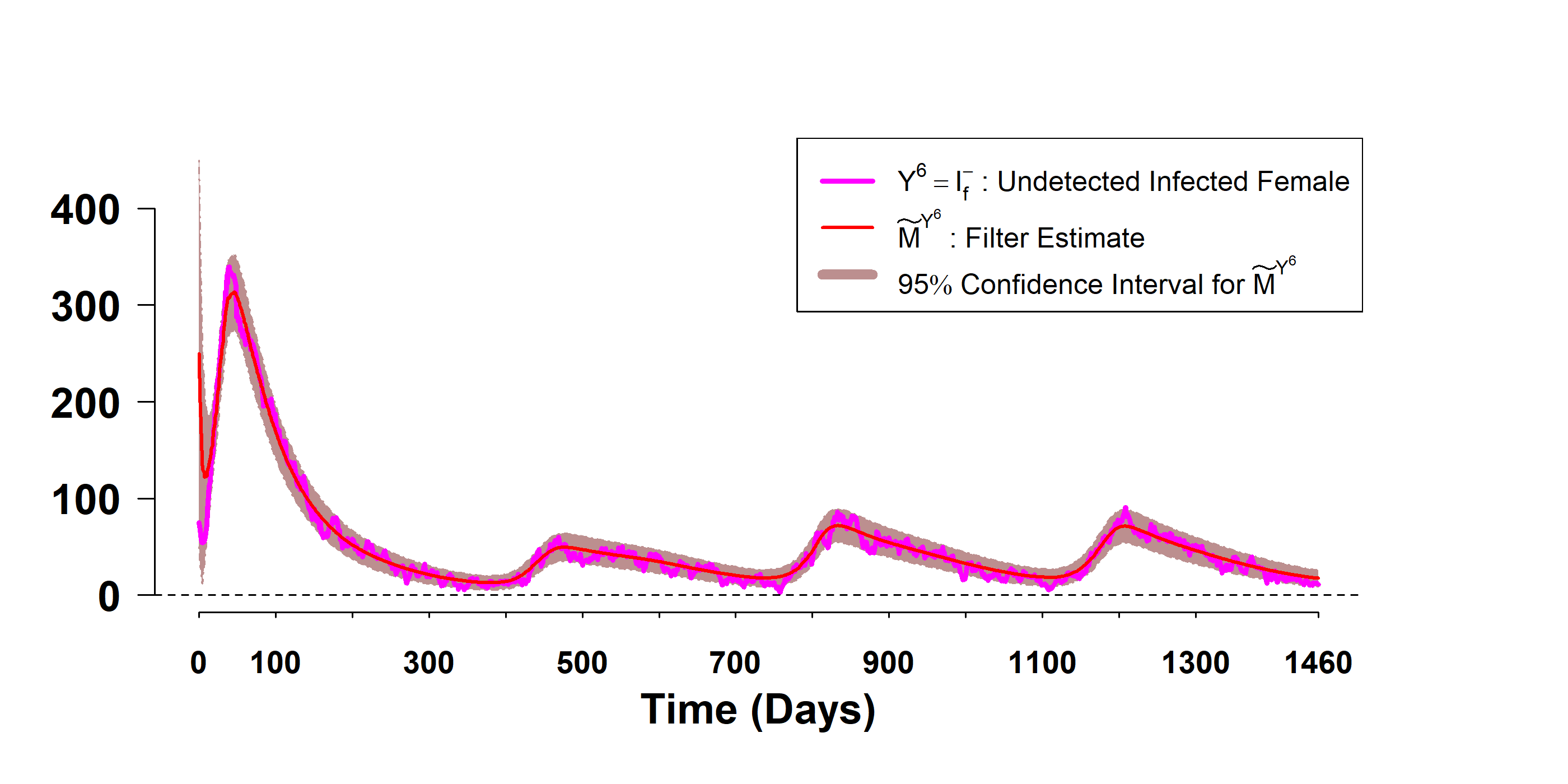}
		}
		
		\subfloat[Exposed Male ($E_m$)]{
			\includegraphics[width=0.45\textwidth]{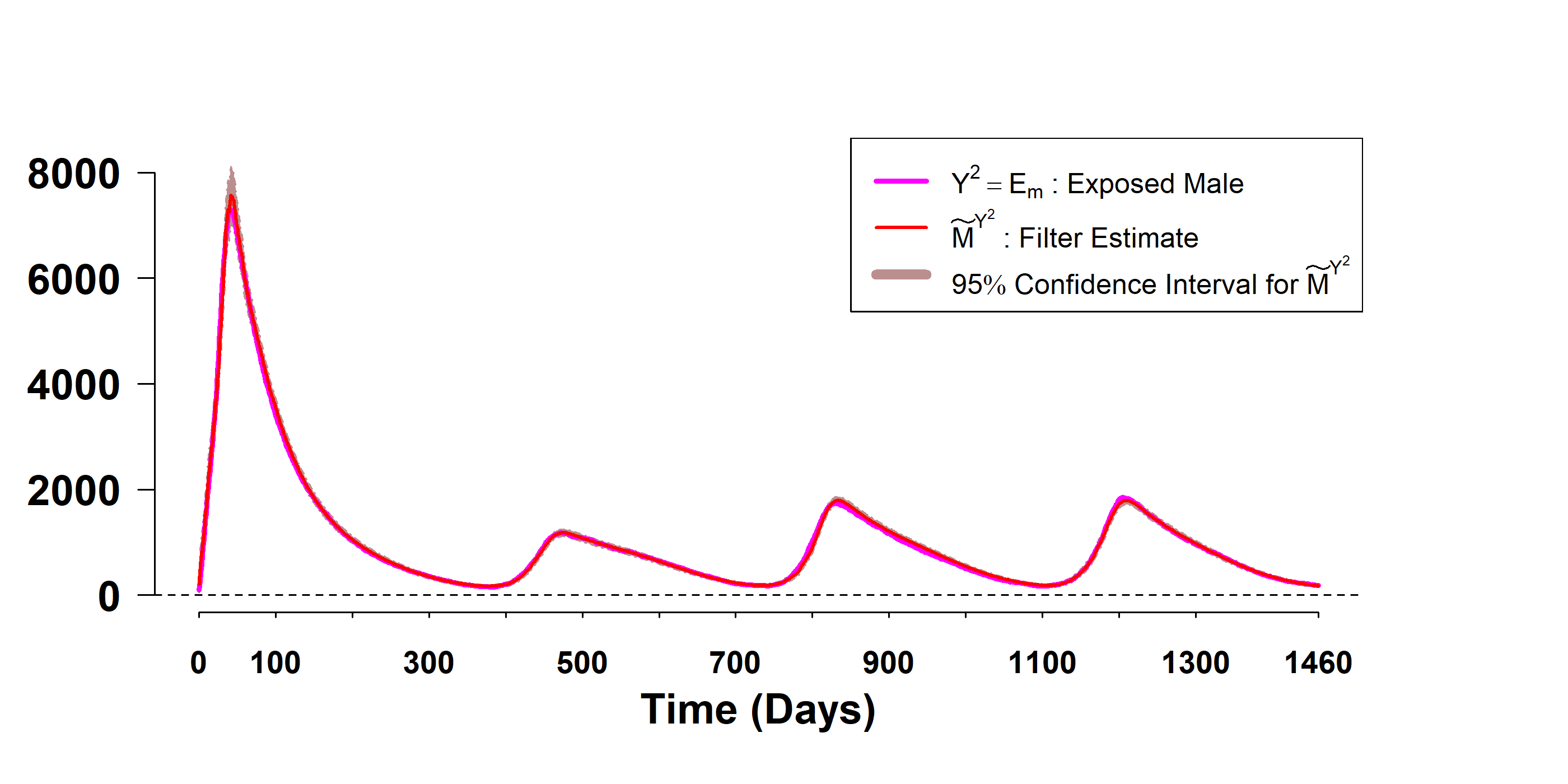}
		}
		\hfill
		\subfloat[Exposed Female ($E_f$)]{
			\includegraphics[width=0.45\textwidth]{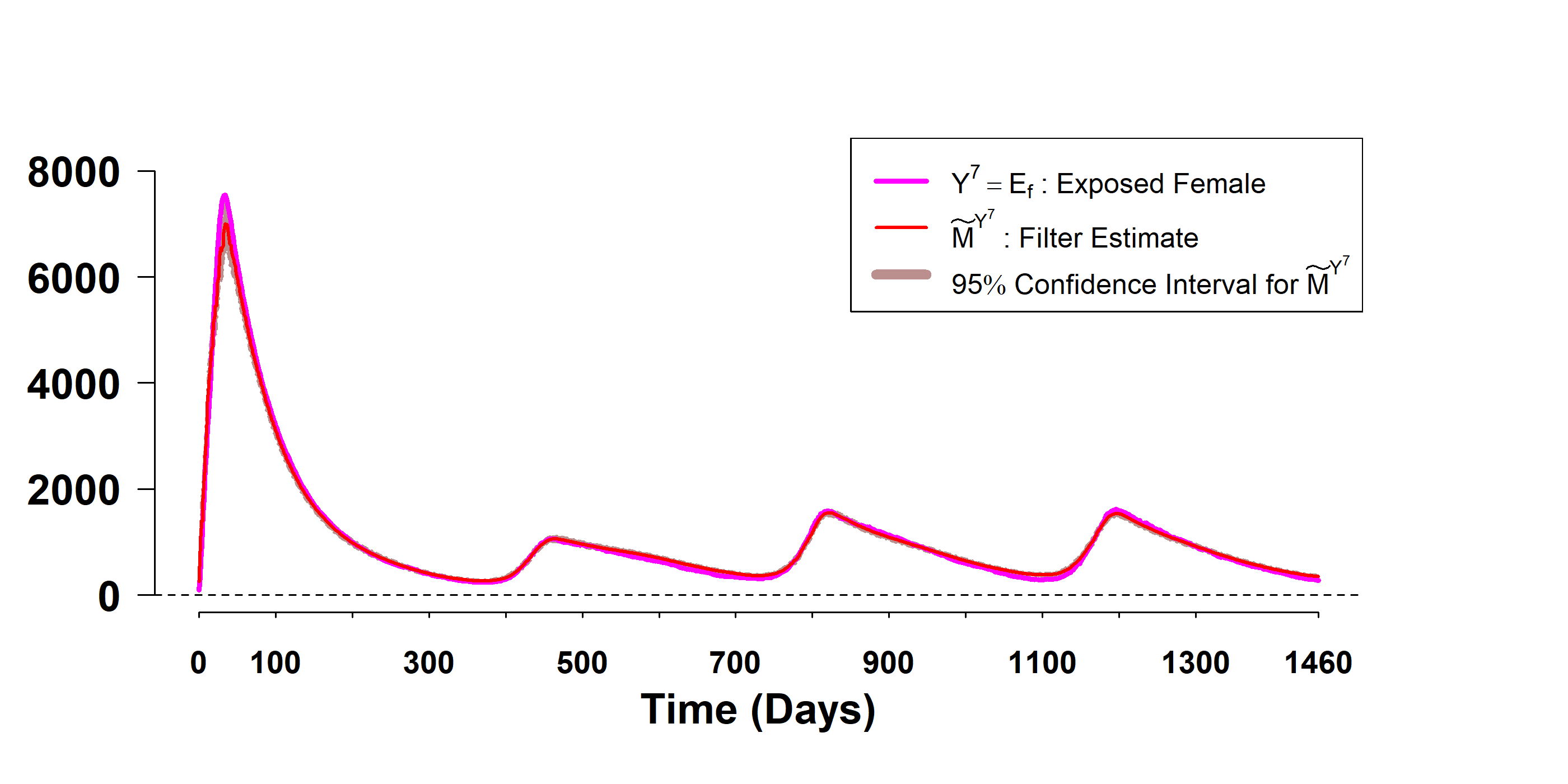}
		}
		
		\subfloat[Recovered Male (1) ($R_m^{1-}$)]{
			\includegraphics[width=0.45\textwidth]{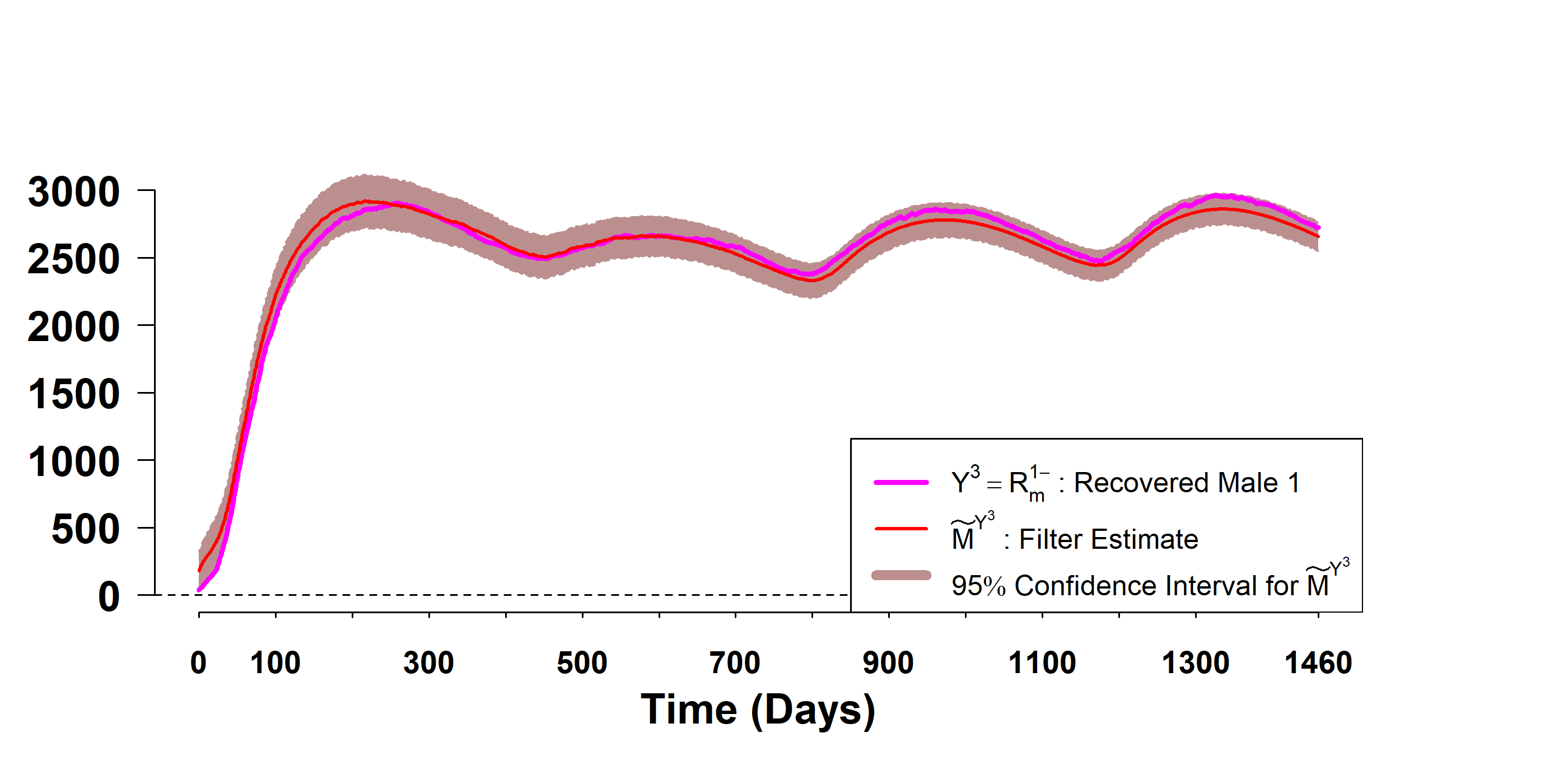}
		}
		\hfill
		\subfloat[Recovered Female (1) ($R_f^{1-}$)]{
			\includegraphics[width=0.45\textwidth]{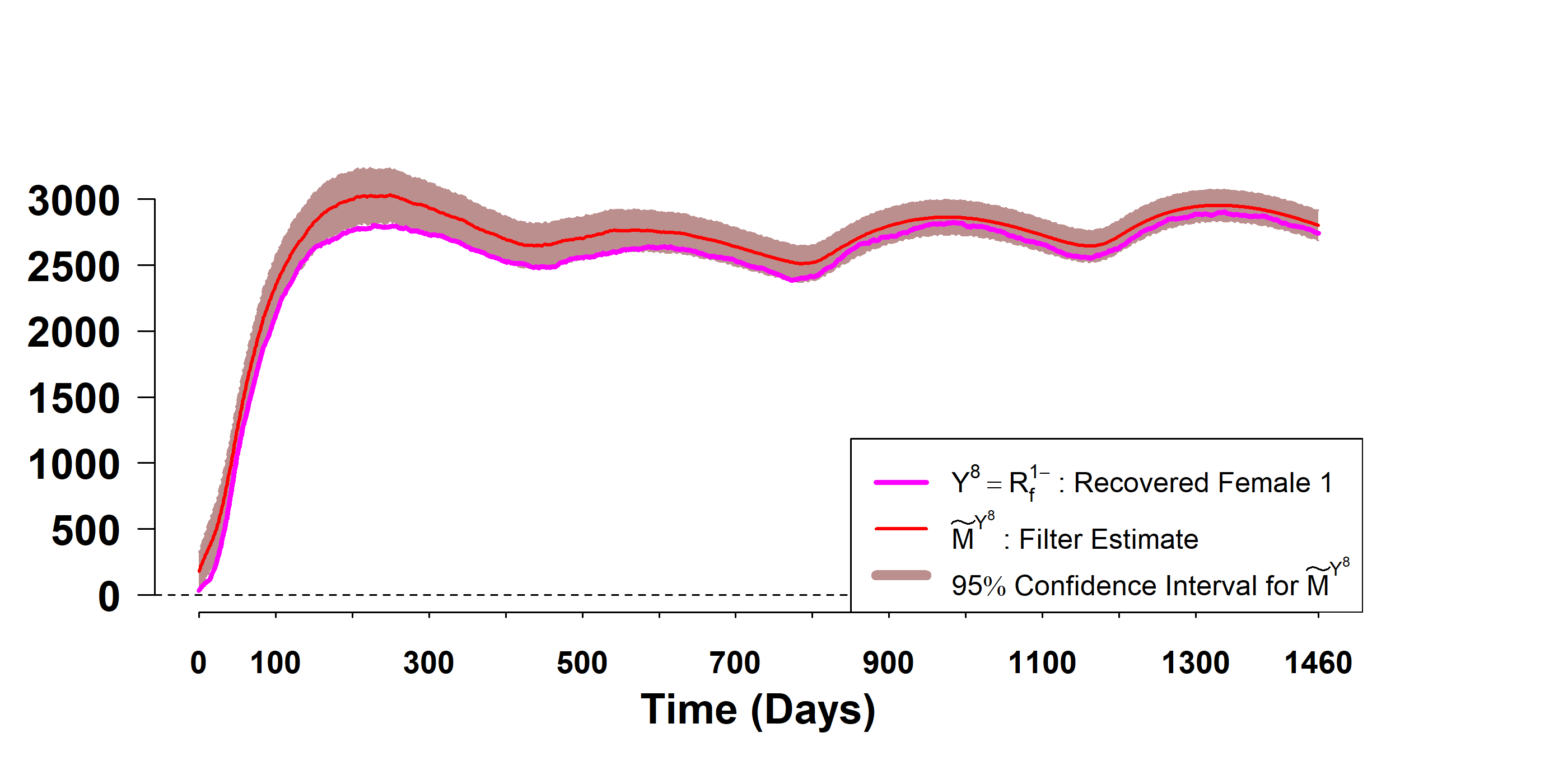}
		}
		
		\subfloat[Recovered Male (2) ($R_m^{2-}$)]{
			\includegraphics[width=0.45\textwidth]{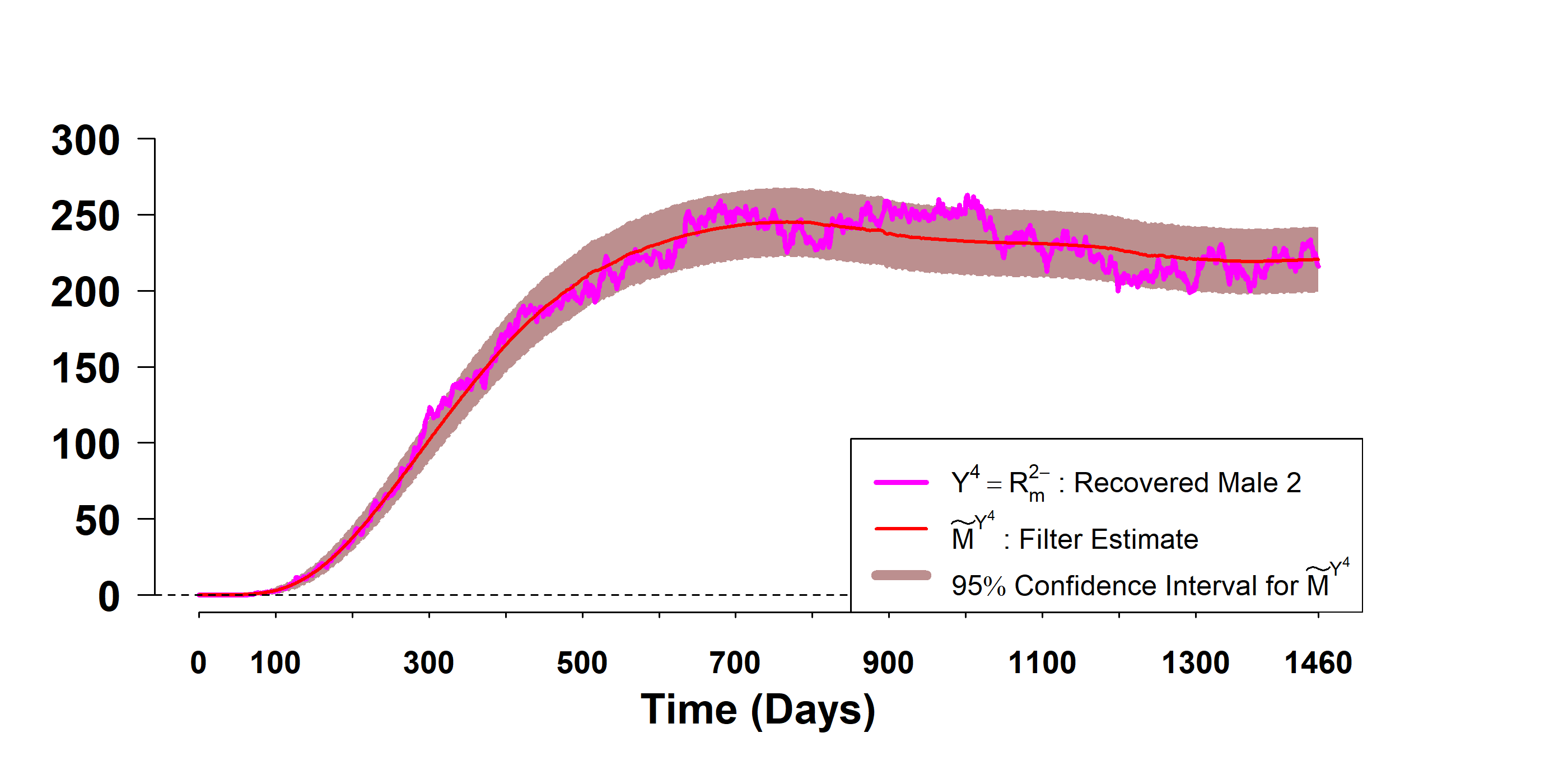}
		}
		\hfill
		\subfloat[Recovered Female (2) ($R_f^{2-}$)]{
			\includegraphics[width=0.45\textwidth]{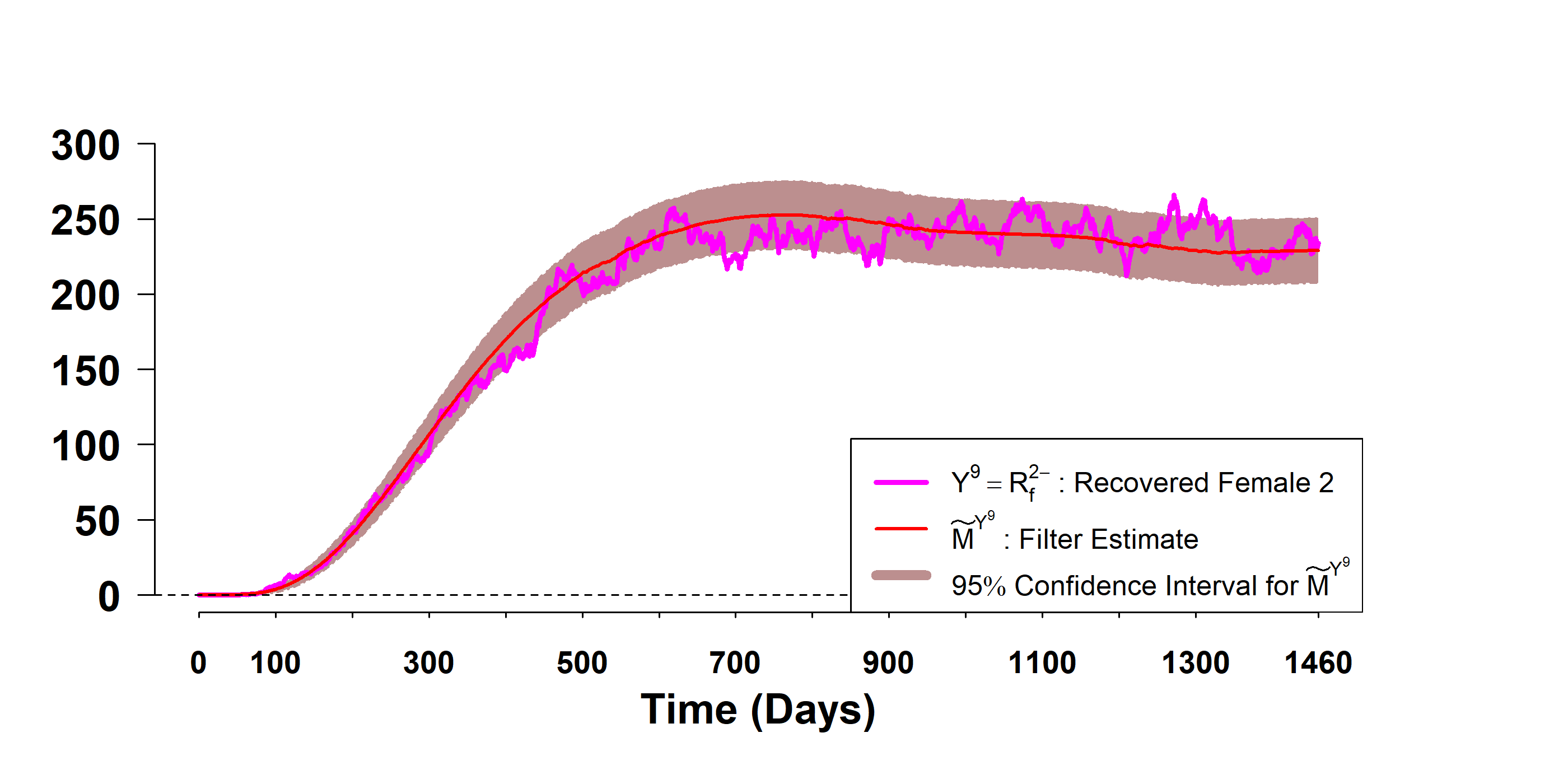}
		}
		
		\subfloat[Susceptible Male ($S_m$)]{
			\includegraphics[width=0.45\textwidth]{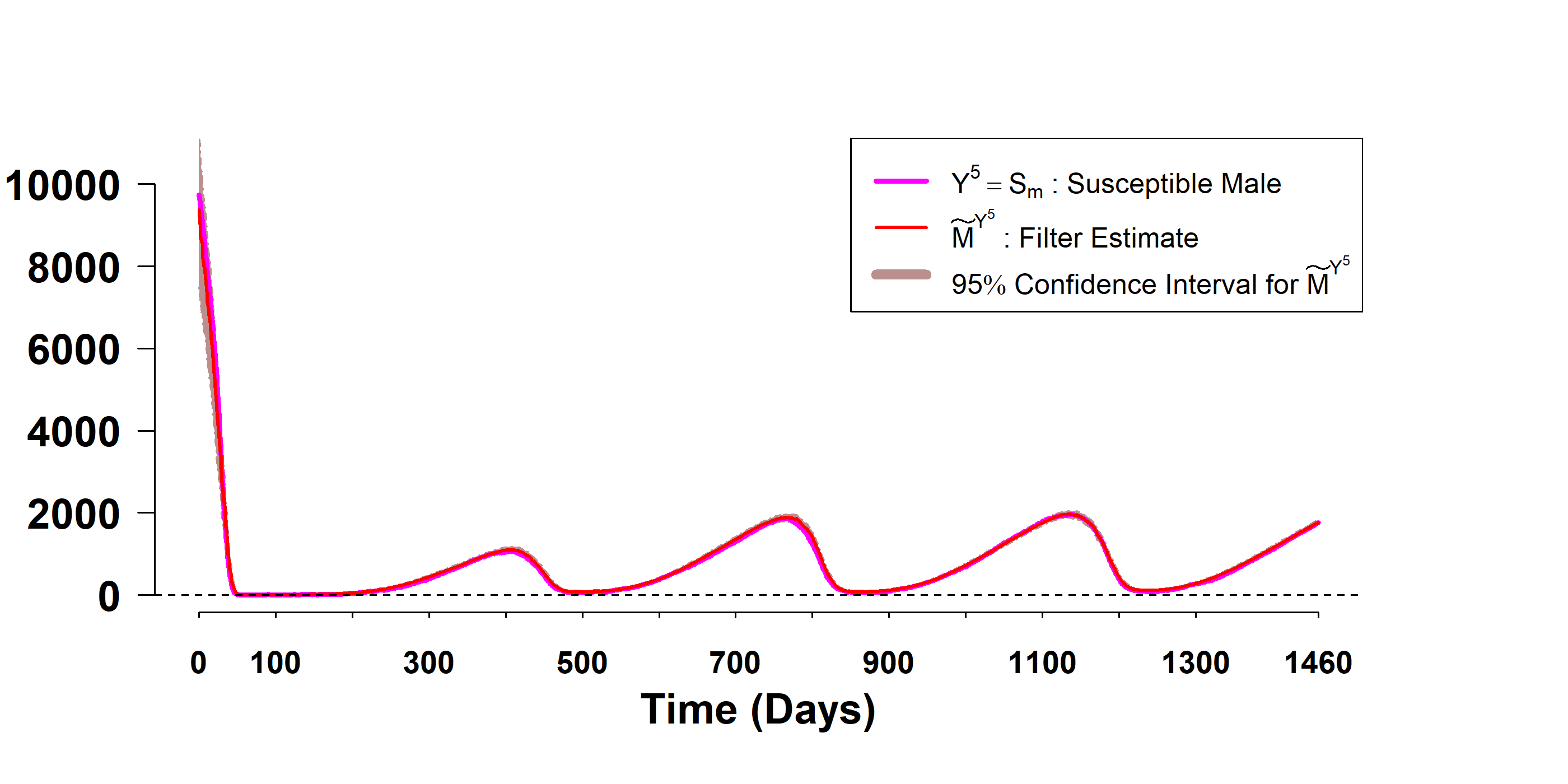}
		}
		\hfill
		\subfloat[Susceptible Female ($S_f$)]{
			\includegraphics[width=0.45\textwidth]{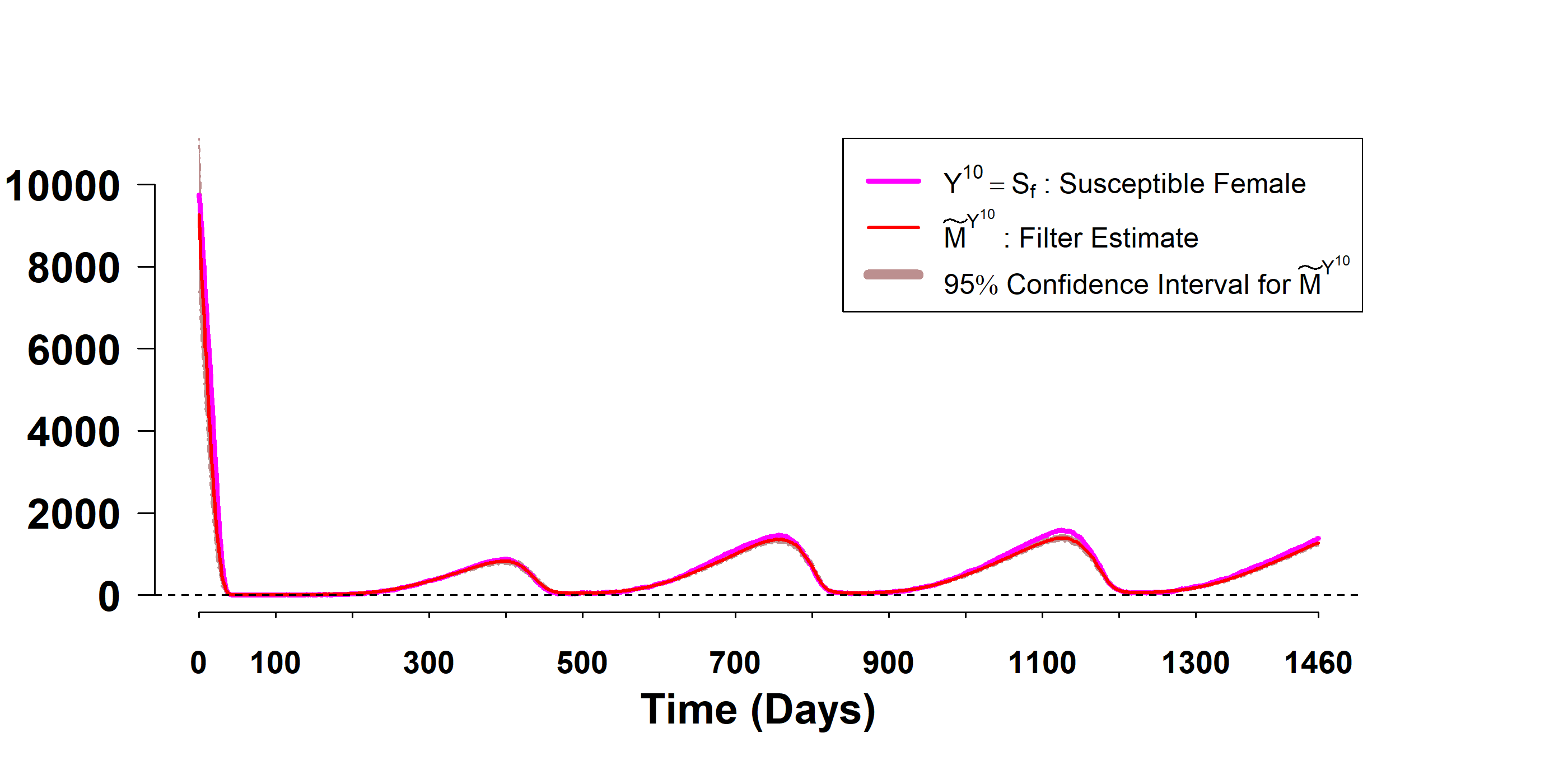}
		}
		
		\caption{Filter estimates and confidence bands for hidden human state variables over time. The left column shows male compartments, and the right column shows female compartments. The magenta curve shows the true simulated state, red the filter estimate, and the gray band represents the 95\% confidence interval.}
		\label{fig:filter_estimates_human}
	\end{figure}
	
	\begin{figure}[h]
		\centering
		\subfloat[Exposed Vectors ($E_v$)]{
			\includegraphics[width=0.45\textwidth]{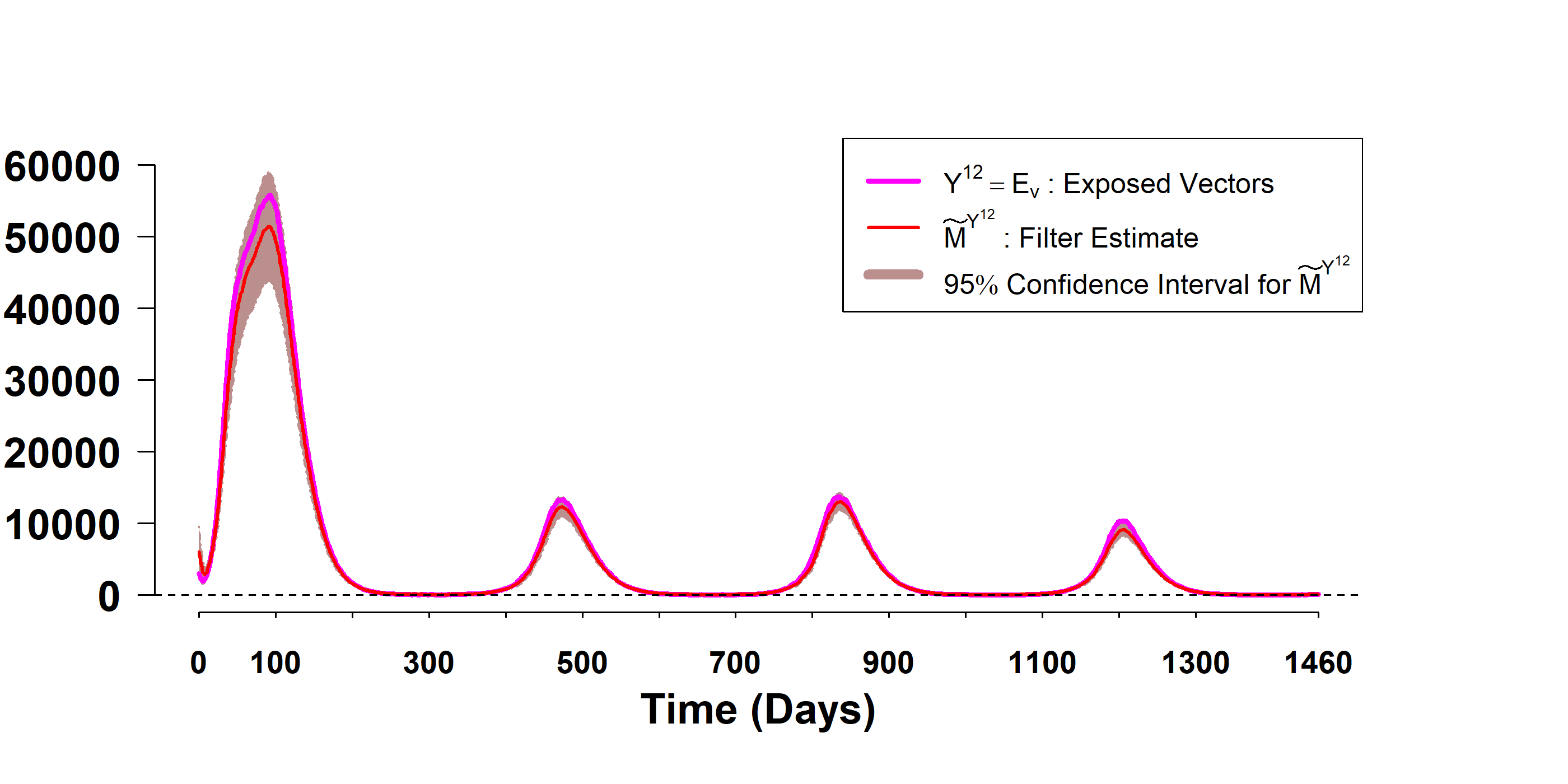}
		}
		\hfill
		\subfloat[Infected Vectors ($I_v$)]{
			\includegraphics[width=0.45\textwidth]{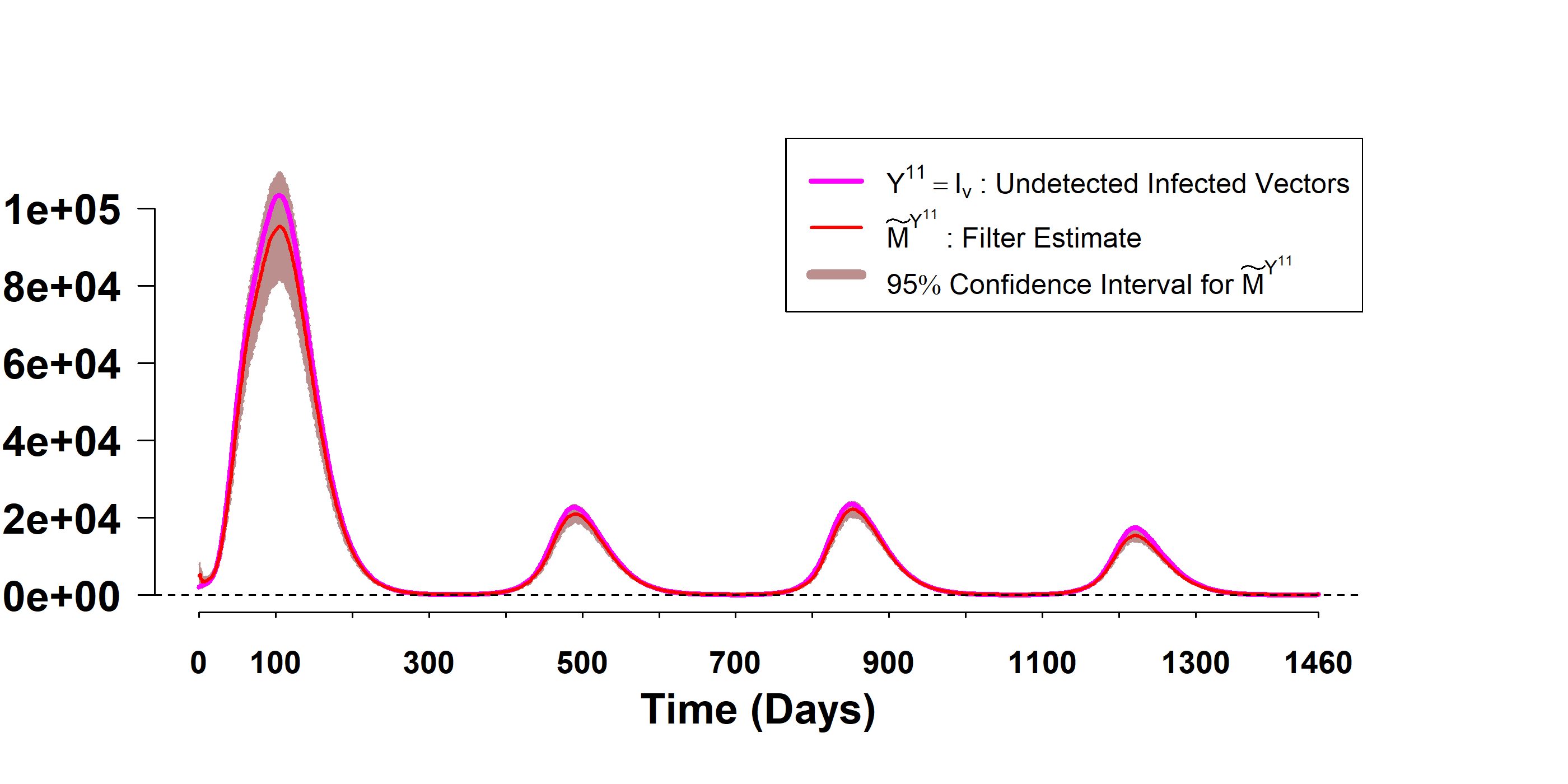}
		}
		
		\vspace{0.5cm}
		
		\subfloat[Susceptible Vectors ($S_v$)]{
			\includegraphics[width=0.6\textwidth]{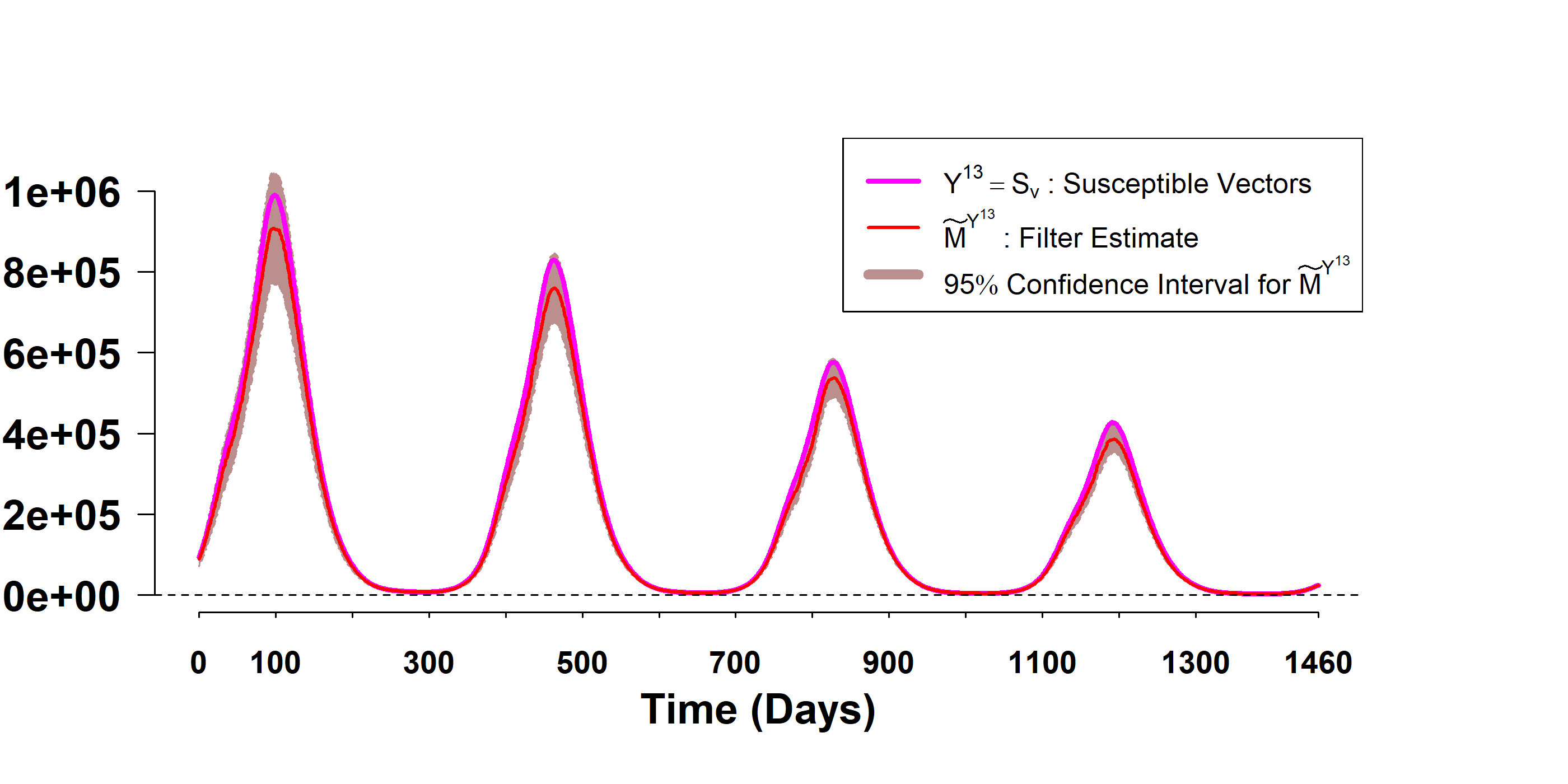}
		}
		
		\caption{Filter estimates and confidence bands for hidden mosquito vector state variables.}
		\label{fig:filter_estimates_vector}
	\end{figure}

	\begin{figure}[htb]
		\centering
		
		\begin{minipage}{0.48\textwidth}
			\centering
			\includegraphics[width=\linewidth]{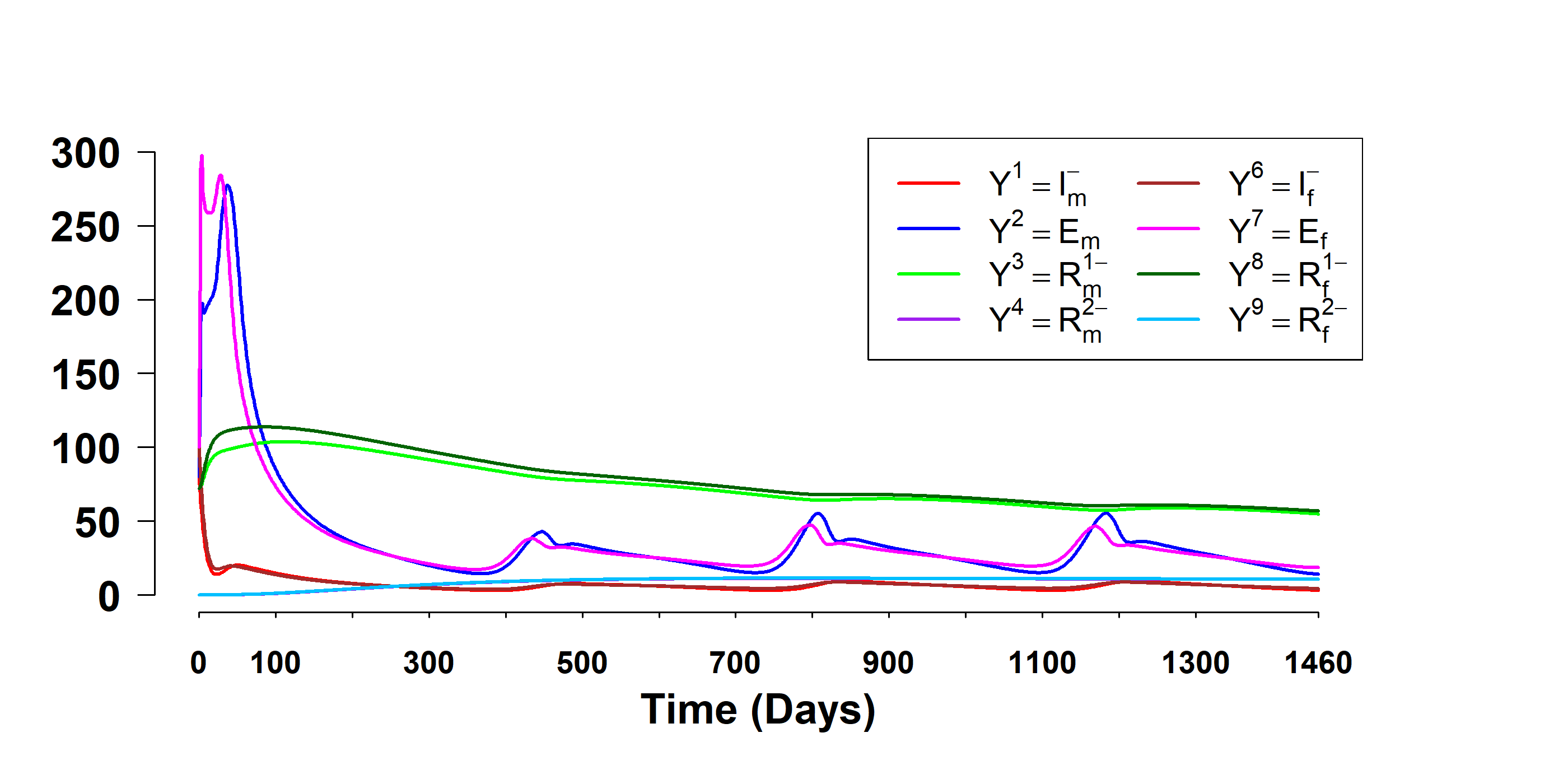}
			\caption*{\small (a) Standard deviations $\sigma^{Y^i} = \sqrt{Q_{ii}}$ for the human compartments: $Y^1$, $Y^2$, $Y^3$, $Y^4$, $Y^6$, $Y^7$, $Y^8$, $Y^9$. These states correspond to latent and recovered classes for both males and females. The variance profiles reveal differences in uncertainty propagation across the infectious and exposed stages.}
		\end{minipage}
		\hfill
		\begin{minipage}{0.48\textwidth}
			\centering
			\includegraphics[width=\linewidth]{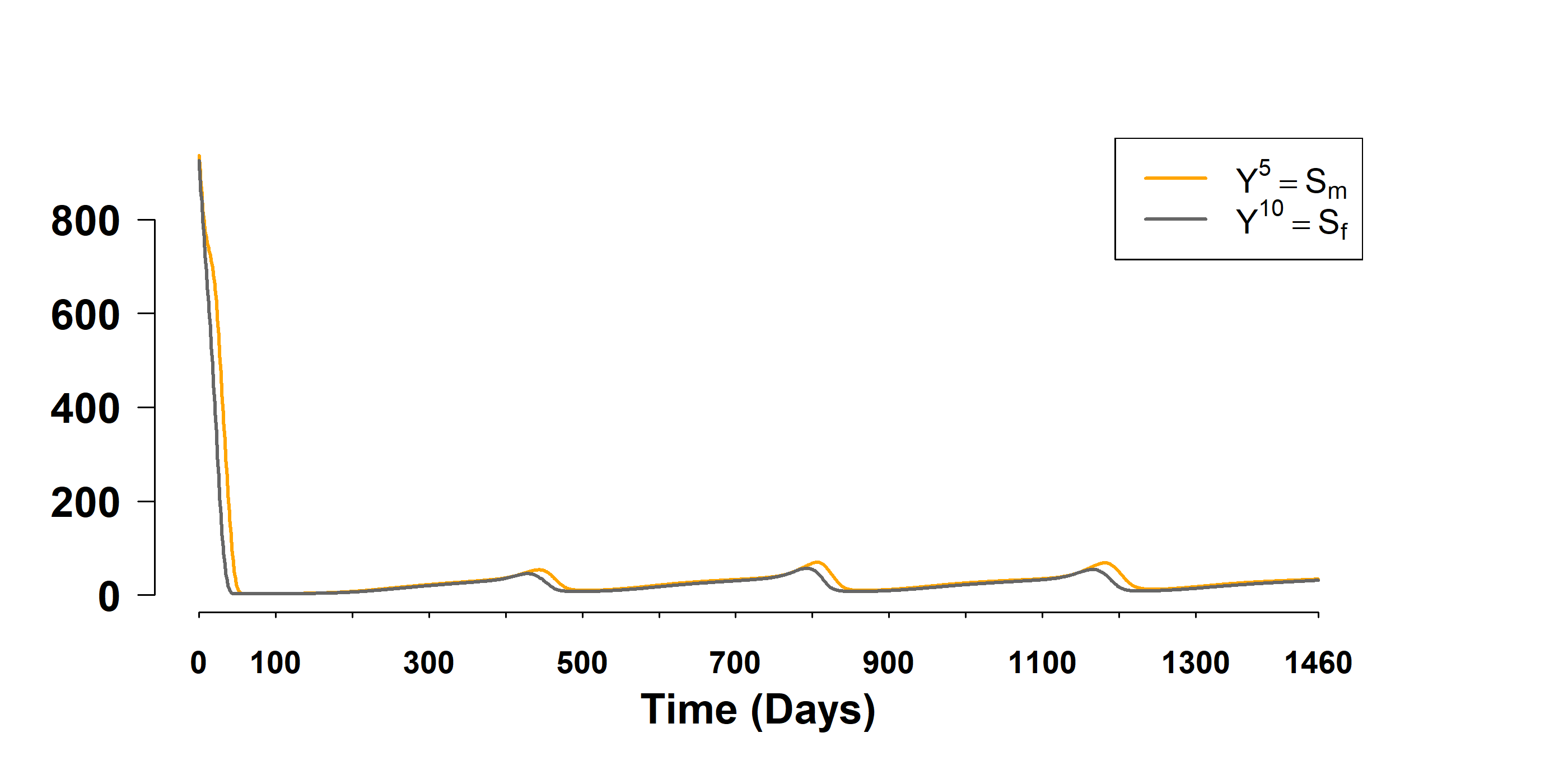}
			\caption*{\small (b) Standard deviations $\sigma^{Y^i} = \sqrt{Q_{ii}}$ for the susceptible compartments $Y^5$ and $Y^{10}$ (susceptible males and females). Their uncertainties remain relatively small and stable, consistent with population conservation and their indirect role in the observation process.}
		\end{minipage}
		
		\vspace{0.6cm}
		
		\begin{minipage}{0.6\textwidth}
			\centering
			\includegraphics[width=\linewidth]{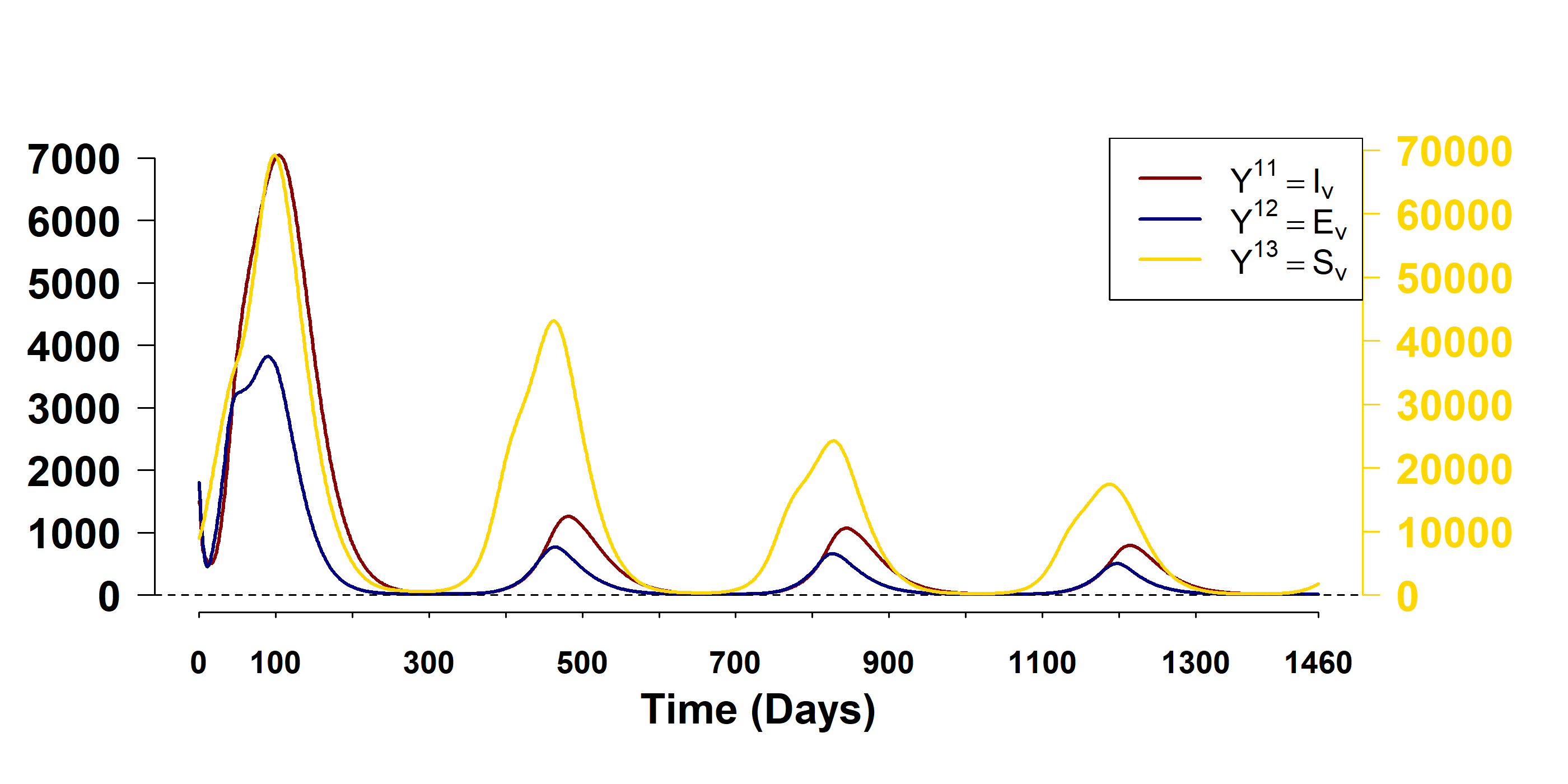}
			\caption*{\small (c) Standard deviations $\sigma^{Y^i} = \sqrt{Q_{ii}}$ for vector-related compartments: $Y^{11}$, $Y^{12}$, $Y^{13}$. These show higher uncertainty levels, reflecting the indirect and unobserved dynamics of mosquito infection stages in the model.}
		\end{minipage}
		
		\caption{Standard deviations of the estimated states $\sigma^{Y^i} = \sqrt{Q_{ii}}$ for different compartments in the Zika model under partial observation. Each panel highlights the uncertainty evolution for distinct groups: (a) human infectious and recovery classes, (b) human susceptible classes, and (c) vector populations.}
		\label{fig:std_devs}
	\end{figure}

	\subsection{Impact of  Initial Estimates }\label{Impact_IEtim}
	We will now analyze the impact of initial estimates on the performance of the proposed filter method through a series of numerical experiments. To do this, we will vary the various initial estimates, namely the conditional mean $\condmeanEKF_0=\condmean_0$ and the conditional variance $ \condvarEKF_0=\condvar_0$, and observe the filter accuracy in the short and long term. Again, we focus on the results for the compartment $I^{-}_{m}$ of undetected  infected male  individuals. For the other hidden compartments, we observed similar results. 
	
	\begin{table}[ht]
		\centering
		\begin{tabular}{|c|c|c|c|}
			\hline
			&  & \multicolumn{2}{|c|}{Estimation} \\
			&True value & Conditional mean & Conditional variance\\
			& $I^{-}_{m,0}$ & $\condmean_{0}^{I^{-}_{m}}$ &   $\condvar_0^{I^{-}_{m}}$  \\ \hline
			Scenario 1            & $80$            & $210$  &$ 30^2$                    \\ \hline
			Scenario 2            & $80$            & $210$     & $0$                 \\ \hline
			Scenario 3            & $80$            & $80$    & $0$                 \\ \hline
		\end{tabular}
		\caption{Different scenarios  for the  initials estimates $\condmean_{0}^{I^{-}_{m}}$ and    $\condvar_0^{I^{-}_{m}}$}
		\label{tab:Val_Impact}
	\end{table}
	
	\begin{figure}[ht]
		\centering
		\subfloat[ Scenario 1 (red): Large initial uncertainty is reduced by learning from observations.
		Scenario 2 (blue):  Incorrectly specified initial estimate with perfect accuracy  needs long time to be corrected.
		]
		{\includegraphics[width=0.48\textwidth]{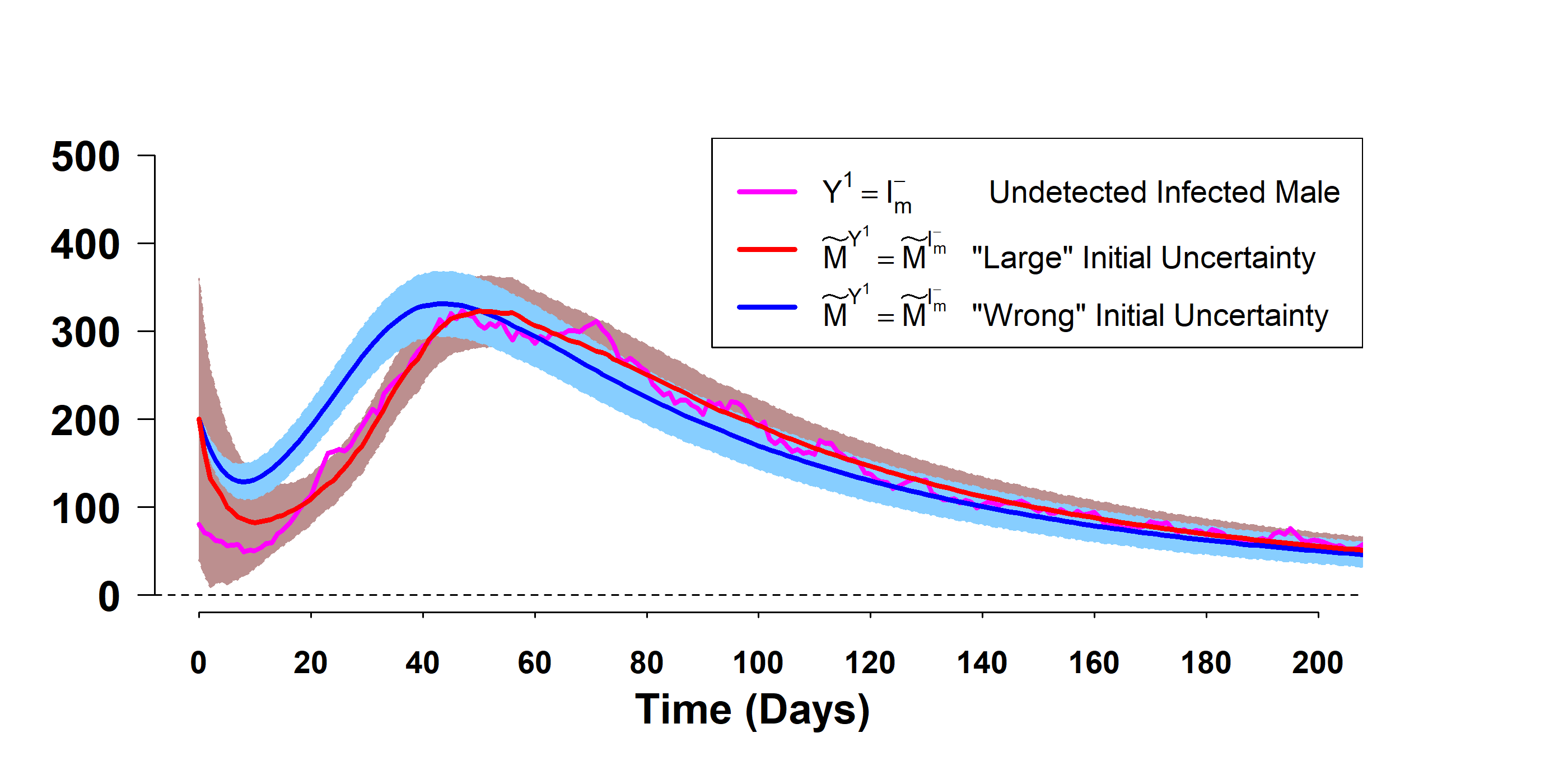}}\label{fig:f179}
		\hfill
		\subfloat[ Scenario 1 (red), scenario 3 (blue):  Zero initial uncertainty is fading out by observation noise. ]{\includegraphics[width=0.48\textwidth]{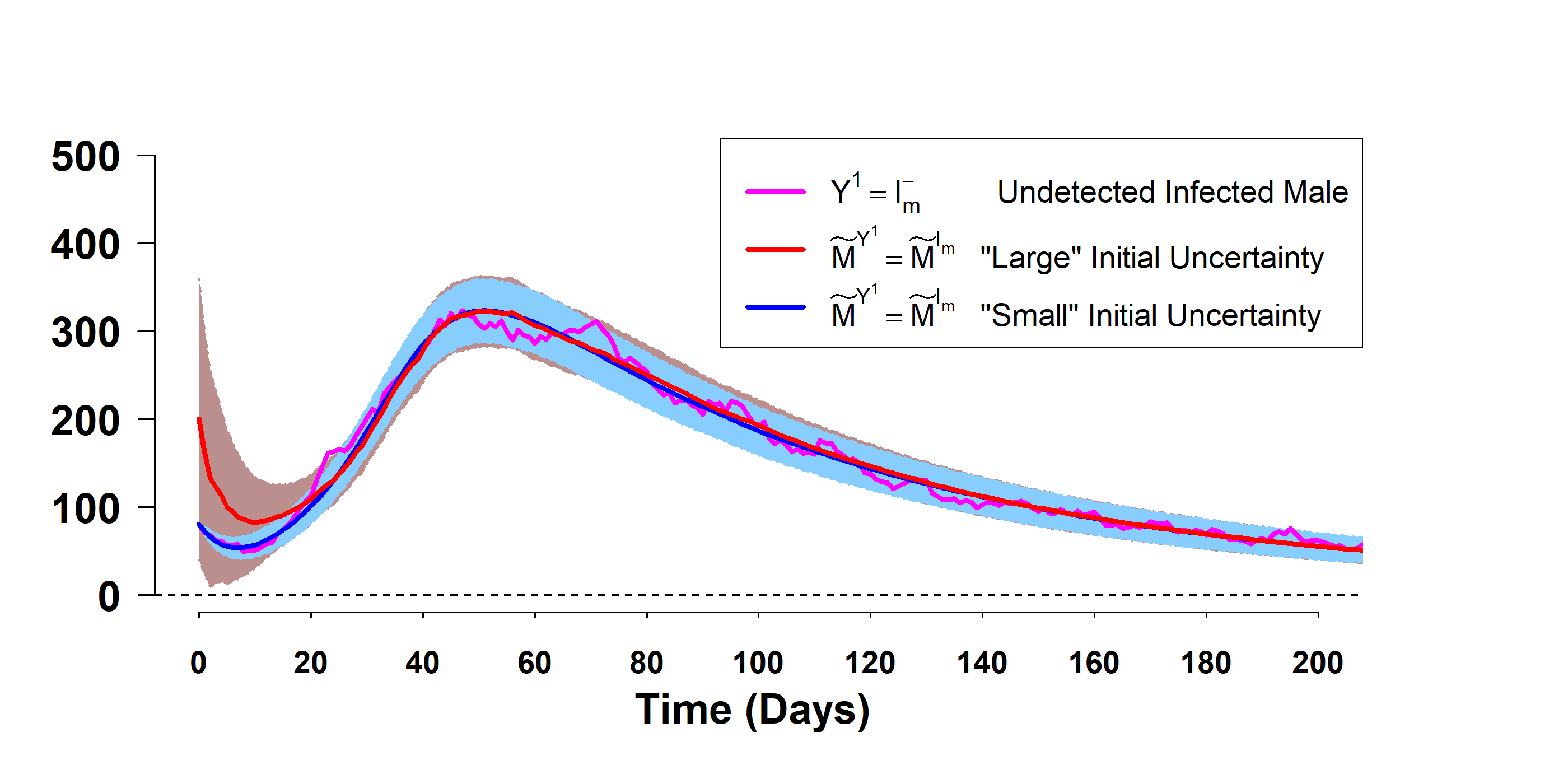}}\label{fig:f180}\\

		\caption{   
			Comparison of the effect of initial uncertainty on filtering performance. The figure shows for the three scenarios the true hidden state $I^{-}_{m}$, the filter estimate $\condmeanEKF_1=\condmeanEKF^{I^{-}_{m}}$, and the associated \(95\%\) confidence band.  }
		\label{fig:filter_initial}
	\end{figure}
	
	\subsection{Impact of Cascade States}
	
	\begin{figure}[htbp]
		\centering
		\subfloat[Full trajectory over 4 years]{
			\includegraphics[width=0.45\textwidth]{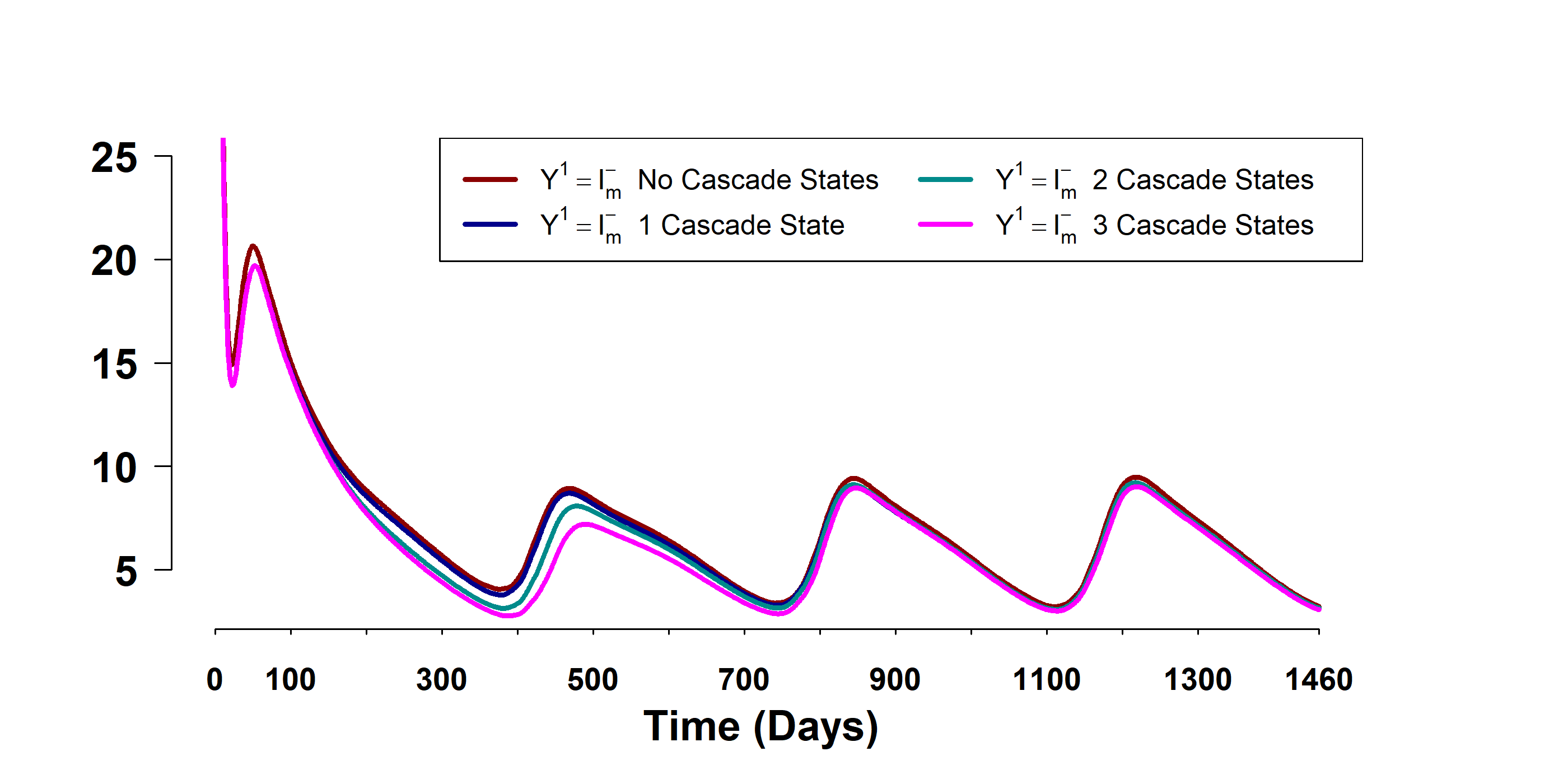}
		}
		\hfill
		\subfloat[Zoom: first year (Day 0–400)]{
			\includegraphics[width=0.45\textwidth]{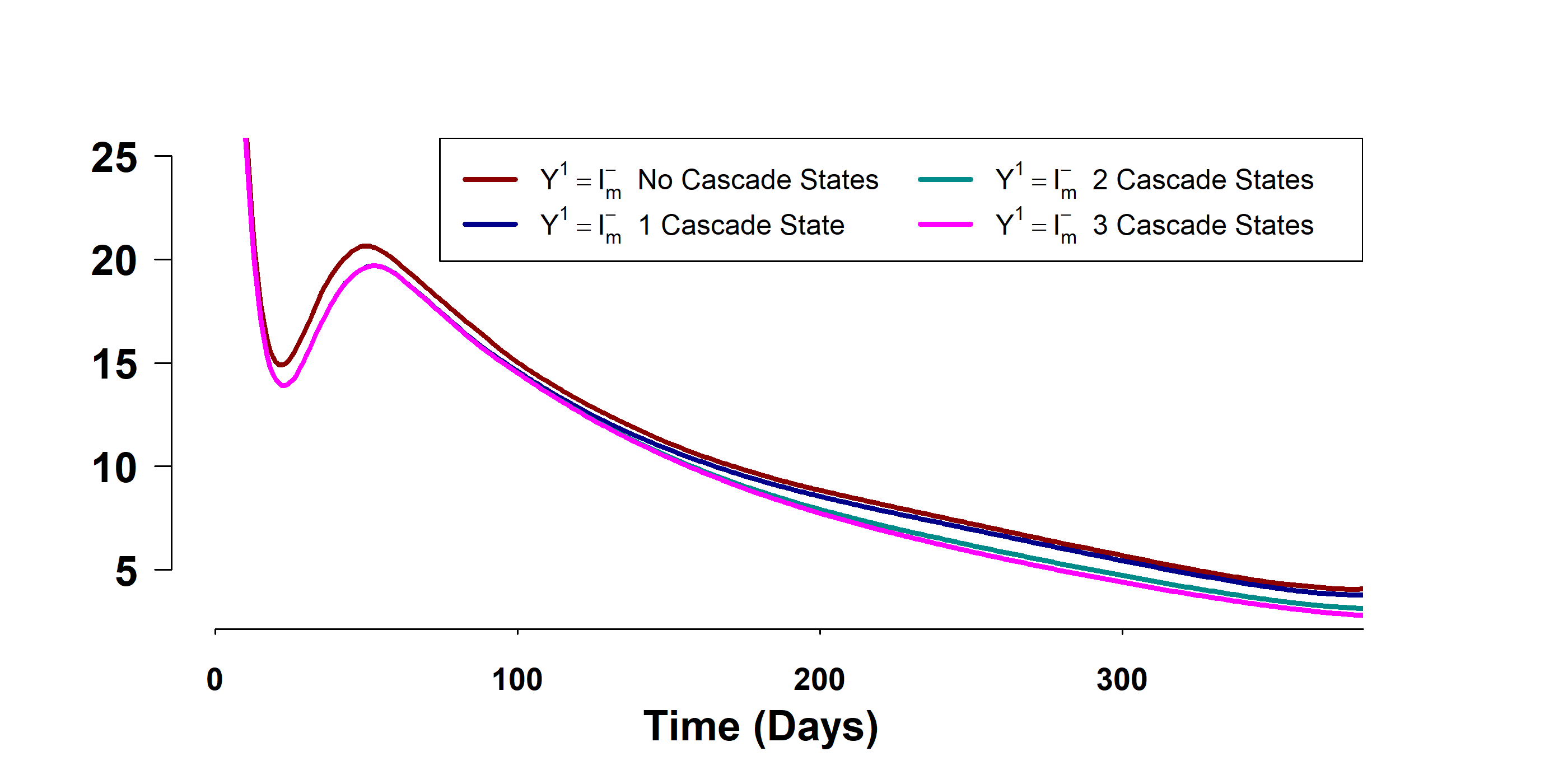}
		}\\
		\subfloat[Zoom: epidemic peak around Day 400–700]{
			\includegraphics[width=0.45\textwidth]{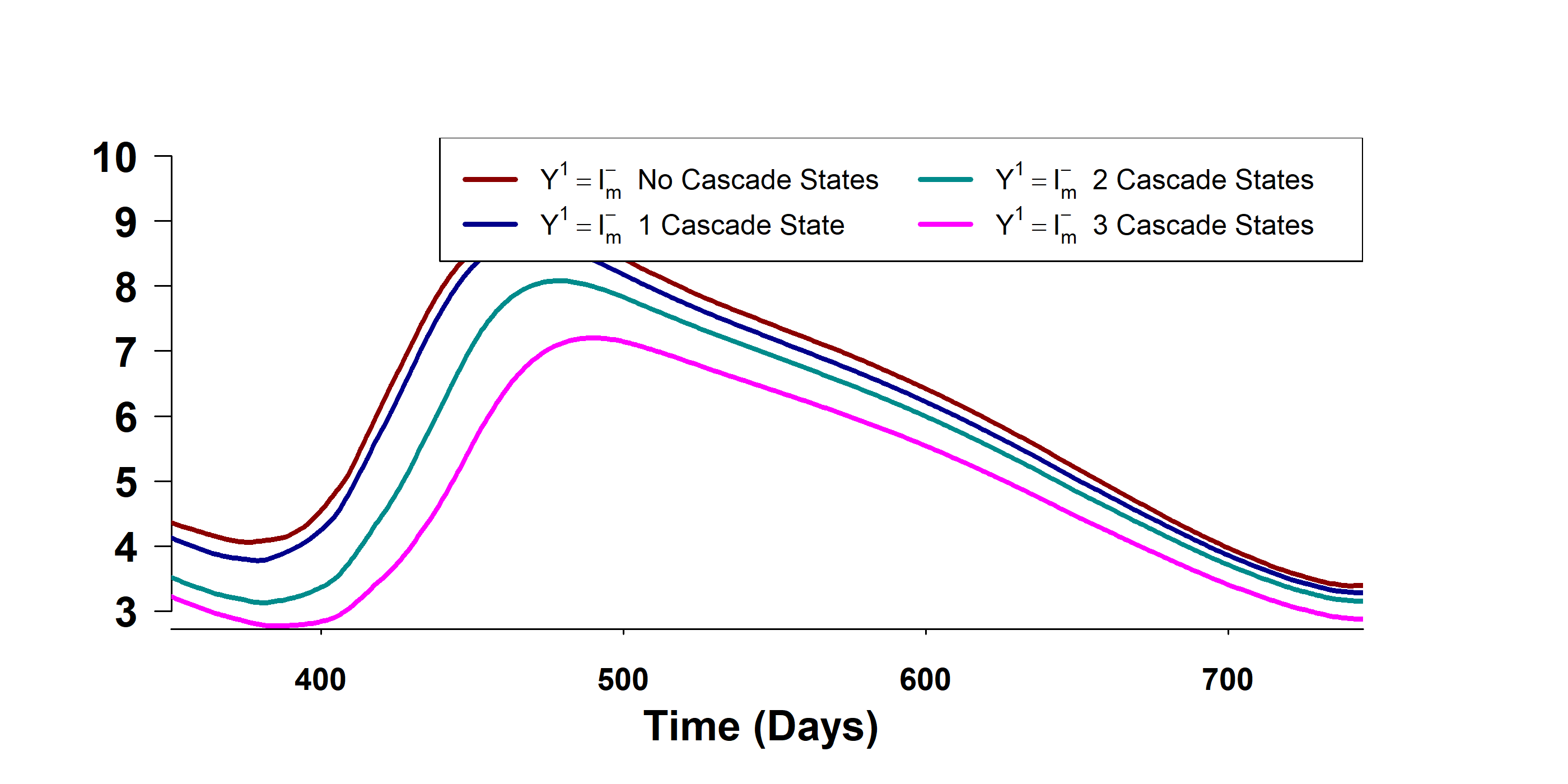}
		}
		\hfill
		\subfloat[Zoom: epidemic peak around Day 800–1100]{
			\includegraphics[width=0.45\textwidth]{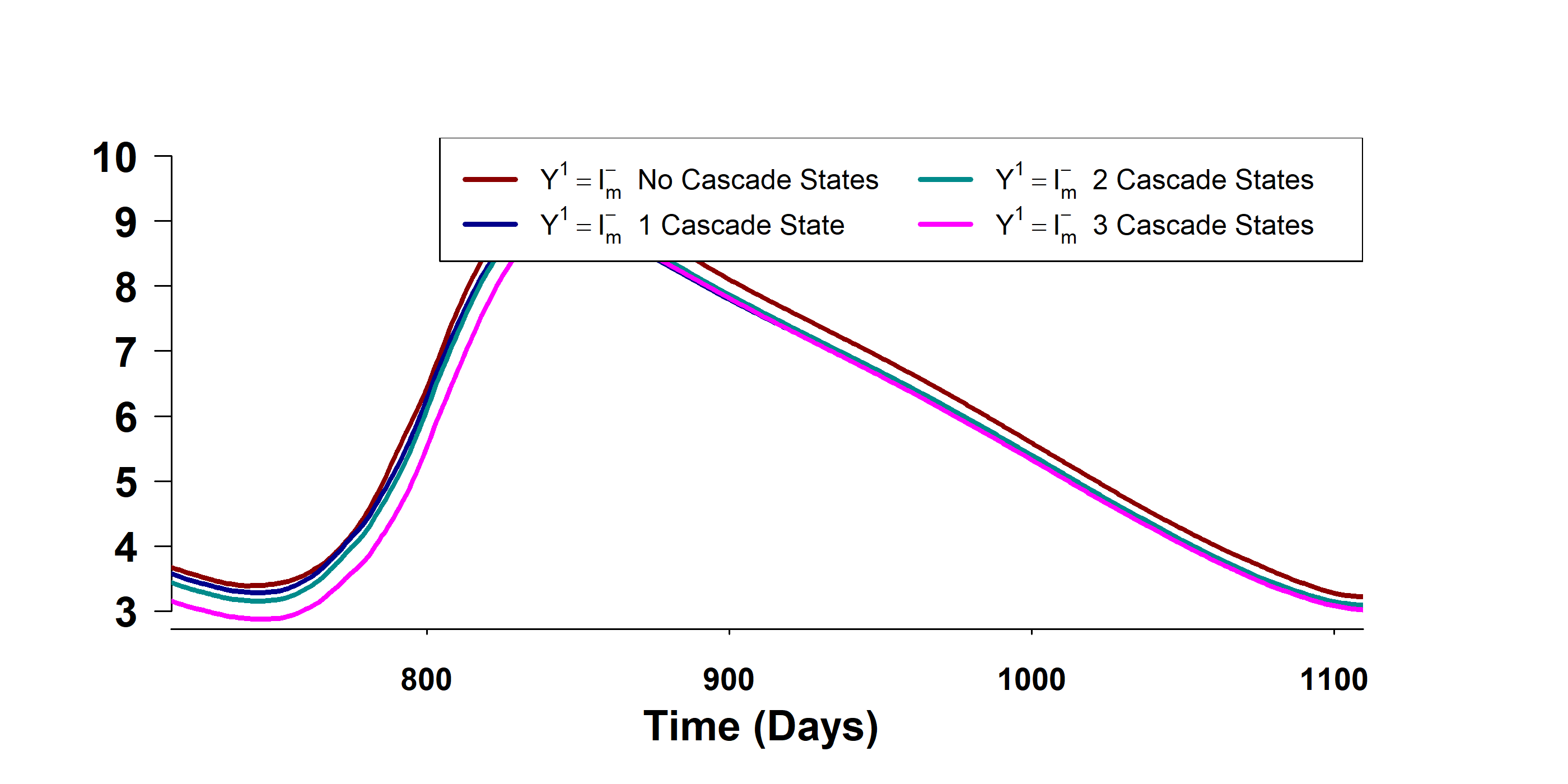}
		}\\
		\subfloat[Zoom: epidemic peak around Day 1100–1460]{
			\includegraphics[width=0.45\textwidth]{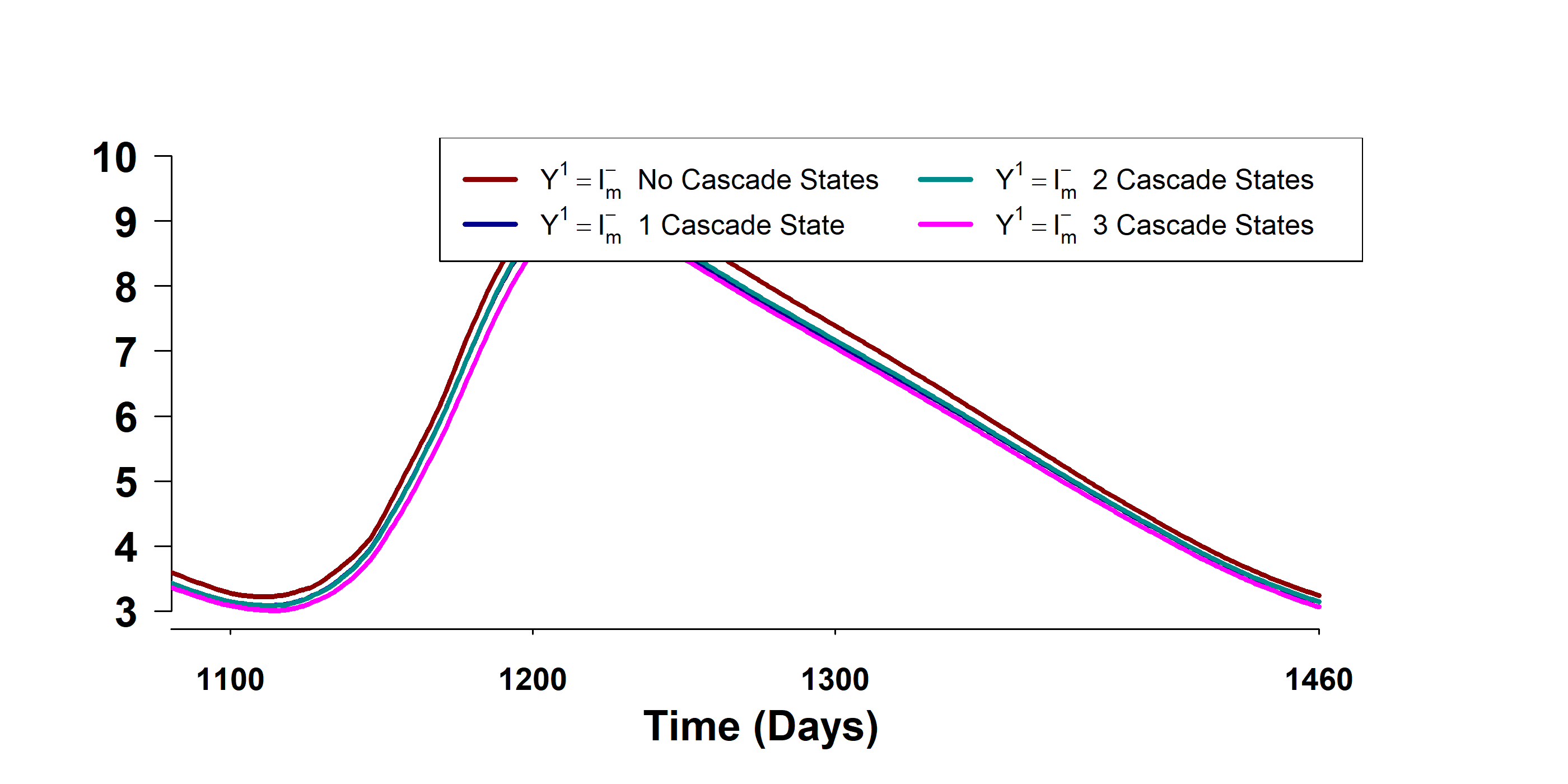}
		}
		\caption{Evolution of the standard deviation across different models depending on the number of cascade states. 
			Panel (a) shows the full four-year trajectory, while Panels (b)–(e) zoom into different epidemic waves. 
		}

	\end{figure}

		In this section we investigate the impact of introducing cascade compartments on the performance of the filter, with particular emphasis on the accuracy of hidden--state estimation as measured by the corresponding standard deviations. We compare the base model with its extended counterparts that incorporate one, two, or three cascade compartments.  
		
		In both frameworks, individuals recovering from a confirmed infection are assumed to remain fully immune for a fixed period of $L=730$ days. In the base model these individuals are assigned to the hidden compartments $R_{\dagger}^{2-}$, where they are aggregated with those whose recovery time already exceeds $730$ days and who may subsequently undergo an unobservable transition back into the susceptible class. Consequently, the compartments $R_{\dagger}^{2-}$ must be treated as hidden.  		
		By contrast, the extended Zika models introduce observable cascade compartments, thereby allowing a more refined description of post-recovery and post-vaccination dynamics.  		
		
		To ensure comparability, identical initial values are imposed across all models for the hidden states $S_{\dagger}, E_{\dagger}, I_{\dagger}^{-}$, and $R_{\dagger}^{1-}$. In the extended model, the first observable recovered cascade compartment is initialized as $R_{\dagger}^{1}=18$, with the remaining cascade compartments set to zero. In the baseline model, the corresponding hidden compartment $R_{\dagger}^{2-}$ is initialized at the same value, i.e.\ $18$.  
		
		When cascade compartments are introduced, the immunity period of $L=730$ days is subdivided into $d^R=1,2,$ or $3$ subintervals, corresponding respectively to one compartment of length $730$, two compartments of $365$ days each, or three compartments of $730/3$ days.  		
		The loss rates of immunity are calibrated so that the expected duration of complete immunity remains identical across the baseline and extended models. Specifically, for immunity loss after recovery we set $\rho_{\dagger}^{2}=1/760$ in the baseline model and, analogously, $\rho_{\dagger}^{2}=1/30$ in the extended formulation, yielding in both cases an expected immune duration of $760=L+30$ days. Similarly, for immunity loss following recovery from $I_{\dagger}^{-}$ we fix $\rho_{\dagger}^{1}=1/760$ in both frameworks, ensuring again an expected duration of $760=L+30$ days.

	\paragraph{Discussion on the convergence of standard deviation} 
	An interesting feature of the results is that the standard deviations of the model without cascade states and the one with three cascade states converge in the long run, despite the expectation that additional cascade states should systematically reduce uncertainty. This behavior can be explained by the structure of the model: the only source of randomness enters through the inflow to the first cascade state, while subsequent cascade dynamics are purely deterministic. As a result, cascade states primarily act as filters that refine short- and medium-term uncertainty, which is particularly evident during outbreak and wave dynamics, when infections rise and decline sharply. In these phases, the system is more sensitive to randomness, and the models with cascade states display smaller variances by capturing more observable information. However, as the epidemic progresses toward the endemic equilibrium, the stochastic input naturally diminishes, leading to negligible long-run variability. Consequently, the additional information carried by cascade states loses its effect, and the variance across all models converges to the same baseline level. In other words, cascade states effectively reduce uncertainty during transient epidemic phases, but their impact fades once the driving noise vanishes and the system approaches equilibrium.

\newpage\clearpage
\appendix 
\section*{Appendix} \label{Notations}
\addcontentsline{toc}{section}{Appendix}
\footnotesize 
\section{Notation}
\label{sec:appNotation}

\vspace*{-3ex}
\renewcommand{\arraystretch}{0.8}
	\begin{longtable}{ll}%
$\alpha_m, \alpha_f$ & Force of infection rate from mosquitoes to susceptible males/females \\
$\beta^+_m,\beta^+_f$ & Testing rates from exposed and asymptomatic male/female  \\
$\beta^s_m,\beta^s_f$ & Infection rates for symptomatic male/female  \\
$\beta^-_m,\beta^-_f$ & Infection rates for asymptomatic male/female  \\
$\beta_m,\beta_f$ & Infection rates for asymptomatic male/female in simplified model \\
$\gamma_m^-,\gamma_f^-$ & Recovery rate from asymptomatic  males/females \\
$\gamma_m^+, \gamma_f^+$ & Recovery rate from detected males /females\\
$\gamma_m^s, \gamma_f^s$ & Recovery rate from symptomatic males /females\\
$\rho^1_m,\rho_m^3, \rho^1_f,\rho_f^3$ & Rate of immunity loss in asymptomatic males/females \\
$\rho^2_m,\rho_m^4, \rho^2_f,\rho_f^4$ & Rate of immunity loss in symptomatic males/females \\
$\rho_m,\rho_f$ & Rate of immunity loss in  males/females in simplified model \\
$\theta, \nu,\mu, B_v $ & Exposure rate, infection rate, natural mortality rate, and birth rate vectors \\
$S_m,S_f $ &   Susceptible male, female  \\	
$E_m,E_f $ &   Exposed male, female  \\	
$R_{m},R_{f} $ & Undetected/detected recovered   in simplified model   \\
$R_{m}^{2-},R_{f}^{2-}$ &  Detected recovered  with fading immunity  in extended model   \\
$R_{m}^{1-},R_{f}^{1-}$ & Undetected recovered  with fading immunity  in extended model   \\
$I^{-}_m,I^{+}_m,I^{s}_m,I^{+}_f,I^{-}_f,I^{s}_f, $ & Undetected  / detected, symptomatic infected \\
$S_v,E_v,I_v$ & Susceptible,Exposed,Infected  vectors \\
$\Delta t$ &  Time step\\
$\Count_{k}$ &    Counting process (for the number of transition $k$ in $[0,t]$)\\
$Y_n ,Z_n $ & Hidden state  and Observation vector at discrete time $n$ \\
$\Pi$ &   Standard Poisson process with unit intensity\\		
$\mathcal{E} ^1,\mathcal{E}^2$ &    Independent   $\mathcal{N}(0,\one)$ random vectors\\
$\mathcal{F}_n^Z$ & Filtration generated by observations up to time $n$ \\
$\mathbb{E}[Y_n | \mathcal{F}_n^Z]$ & Conditional mean (filter estimate) of the hidden state $Y_n$  up to time $n$ \\
$\operatorname{Var}(Y_n | \mathcal{F}_n^Z)$ & Conditional covariance matrix of $Y_n$ given $\mathcal{F}_n^Z$ \\
$m_0, q_0$ & Initial prior mean and initial prior covariance of the hidden state $Y_0$ \\
$\cf_0, \cf_1$ & Constant term and linear operator in the drift of the state transition model \\
$\csigma$ & Diffusion coefficient matrix for process noise in the state dynamics \\
$\cg$ & Noise sensitivity matrix in the state equation \\
$\ch_0 , \ch_1$ & Constant term and linear operator mapping hidden states  in the observation model \\
$\cell$ & Noise sensitivity matrix in the observation equation \\
$[A]^+$ & Moore–Penrose pseudoinverse of matrix $A$\\
$\Zpathn$ & Observation trajectory $(Z_0, Z_1, \ldots, Z_n)$ up to time $n$ \\
$\widehat{Y}_n$ & Estimate of $Y_n$ adapted to $\mathcal{F}_n^Z$ and square-integrable \\
$\Omega$ & Underlying probability space \\
$\Fprior$ & $\sigma$-algebra with prior information on the initial distribution of $Y_0$ \\
$f, f_0,f_1$ & Nonlinear drift function for deterministic dynamics of the hidden state \\
$\overline{Y}_n$ & Reference point of linearized  drift function $f$  \\
$[Df](\overline{Y}_n)$ & Jacobian matrix of $f$ evaluated at the reference point $\overline{Y}_n$ \\
$R_n^{in} $ & newly recovered individuals entering the first cascade\\
$R_n^{out} $ & The outflow from the final hidden compartment \\
$d,d_1,d_2$ &  Number of all, hidden and observable states  \\		   
 $\KR$ &  Number of cascade compartments \\
$\one_{n}$ & $n\times n$ Identity matrix of order $n$ \\
$K$ &  Total number of different transitions  \\
$ \LR$ &  Number of time steps with complete immunity\\
$\ell$ &    Diffusion coefficient, observable  state\\
$\condmean,\condmeanEKF$ & Conditional mean/ EKF approximation\\
$N, N_t$ &  Total population size and total number of time steps \\
$\condvar,\condvarEKF$ & Conditional variance/ EKF approximation\\
$\PR_k,$ &   Number of original cascade compartments grouped to one compartment \\
$R_{j}^{+}$ & Cascade compartment with respect to observable recovered \\		
$W, W^{1},W^{2}$ &   multi-dimensional standard Brownian motions\\
$X,\overline{X}$ &  State vector for absolute and relative subpopulation size \\
$Y,\widetilde{Y}$ & Hidden state in original/linearized system \\
$\overline{Y}$ & Reference point for Taylor expansion \\
$Z,\widetilde{Z}$ &  Observable state in original/linearized system \\
$\halfsat$ &  half-saturation parameter \\
$\Isum$ &   The total number of human infectious \\
$\vartheta$ &   Transition parameter for cascade states  \\

\end{longtable}

\section{Coefficients of Discrete-time Recursions of the State Process}
\label{app:coeff}

\subsection{Simplified  Zika Model}
\label{app:sym_coeff}

In this section, we provide an explicit formulation of the model coefficients \( f, h, \sigma, g, \ell \) that define the stochastic dynamics of the system introduced in Subsection \ref {sec:symplistic}. These coefficients appear in the state-space recursion equations \eqref{state_YZ_Zika}, which govern the evolution of the hidden and observable states of the model given by $Y = (I^{-}_m, R_m, S_m, I_f^- R_f, S_f, I_v, S_v)^\top$, and $Z = (I^{+}_m, I^{+}_f)^\top$, respectively. This results in a total of $d = 10$ distinct states and $K = 11$ possible transition pathways.  In the coefficient \(f\), the placeholder notation  $x^{\mathcal{I}}$ is used in place of $\Isum $.

\footnotesize
\begin{flushleft}
$
f(n,y,z) = y + \Delta t \cdot
\left[
\begin{array}{l}
\alpha_m y^3 { \frac{y^7}{\halfsat+y^7}} -  \beta_m y^1  -\gamma_m^- y^1  \\
\gamma_m^- y^1 + \gamma_m^+ z^1  -  \rho_m y^2\\
-\alpha_m y^3 { \frac{y^7}{\halfsat+y^7}}+ \rho_m y^2  \\
\alpha_f y^7 { \frac{y^7}{\halfsat+y^7}} + \omega y^7 y^1 - \beta_f y^4 - \gamma ^-_f y^4   \\
\gamma_f^- y^4 + \gamma_f^+ z^2 -  \rho_f y^6  \\
-\alpha_f y^7 { \frac{y^7}{\halfsat+y^r}}-\omega y^7 y^1 + \rho_f y^6  \\
 \theta y^9   {  \frac{x^\mathcal{I}}{N}} - \mu y^8 \\
{ B_v (y^6+y^7+y^8)} - \theta y^9   { \frac{x^\mathcal{I}}{N}} - \mu y^9

\end{array}
\right]
$
\vspace{0.2cm}

$h(n,y,z) =  h_0(n,z) + h_1(n,z)y $
\end{flushleft}

\footnotesize
\begin{flushleft}
$
h_0(n,z) = z+  
\left[
\begin{array}{l}

-\gamma_m^+  z^1 \\
-\gamma_m^s z^2 \\
 \end{array}
\right] \Delta t $  \,\,\,\,\,\,\,\,\,\,\,  $ h_1(n,z)= \left [ \begin{array}{ccccccccccccccccc}
\beta_m& 0&0&0&0&0&0&0\\
0&0&0&\beta_f&0&0&0&0
 \end{array} \right ] \Delta t$
\end{flushleft}
\begin{flushleft}
\[
\sigma (n,y,z)=
\left[
\scalebox{0.85}{$
\begin{array}{ccccccccccc}
\sqrt{ \alpha_m  y^3 { \frac{y^7}{\halfsat+y^7}}}  & - \sqrt{\gamma ^-_m  y^1  }  & 0 & 0 & 0 & 0 & 0  & 0 & 0 & 0 & 0  \\
0 &  \sqrt{\gamma_m^- y^1}  & - \sqrt{\rho _m y^2} & 0 & 0 & 0 & 0 & 0 & 0 & 0 & 0   \\
-\sqrt{ \alpha_m  y^3 { \frac{y^7}{\halfsat+y^7}}} & 0   & \sqrt{\rho_m y^2}  & 0 & 0 & 0 & 0 & 0 & 0 & 0  & 0  \\
0 & 0 & 0  & \sqrt{ \alpha_f  y^7 { \frac{y^7}{\halfsat+y^7}}}   & \sqrt{\omega y^7 I^-_m }  & -\sqrt{ \gamma _f y^4} & 0 & 0 & 0 & 0 & 0  \\
 0 & 0    & 0 & 0 & 0   & \sqrt{\gamma_f^- y^4} & - \sqrt{\rho _f y^6}& 0 & 0 & 0 & 0  \\
0 & 0 & 0  &- \sqrt{ \alpha_f  y^7 { \frac{y^7}{\halfsat+y^7}}}  & -\sqrt{\omega y^7 I^-_m }    & 0 & \sqrt{\rho _f y^6} & 0 & 0 & 0 & 0 \\
0 & 0 & 0 & 0 & 0 & 0 &  0  & \sqrt{ \theta y^9({ \frac{x^\mathcal{I}}{N}})}   & - \sqrt{ \mu y^8} & 0 & 0 \\
0 & 0 & 0 & 0 & 0 & 0 & 0  & -\sqrt{ \theta y^9 ({  \frac{x^\mathcal{I}}{N}})}  & 0  &\sqrt{ B_v(y^6+y^7+y^8)} & -\sqrt{\mu_ v} y^9 \\
0 &  0  & 0 & 0 & 0 & 0  & 0 & 0 & 0 & 0  & 0\\
0 & 0 & 0 & 0 & 0 & 0 & 0 & 0 & 0 & 0 & 0\\
\end{array}
$}
\right]\sqrt{\Delta t}\]
\end{flushleft}
\footnotesize
\[ g(n,y,z) =   \left[
\scalebox{0.85}{$
\begin{array}{ccccccccc}
-\sqrt{\beta_{m} y^1}  &0
& 0 & 0   \\
0 &  \sqrt{\gamma_m^+ z^1} & 0 & 0  \\
0 & 0 & 0 & 0   \\
0 & 0 &-\sqrt{\beta_{f} y^4} & 0   \\
0 & 0 & 0 & \sqrt{\gamma_f^+ z^2}  \\
0 & 0 & 0 & 0   \\
0 & 0 & 0 & 0  \\
0 & 0 & 0 & 0  \end{array} $}\right] \sqrt{ \Delta t},   \,\,\,\,\,\\   \,\,\,\,\,  \ell (n,y,z)= \left[
\scalebox{0.85}{$
\begin{array}{ccccccccc}
\sqrt{\beta_{m} y^1}  & -\sqrt{\gamma_m^+ z^1}
& 0 & 0   \\
0 & 0 &\sqrt{\beta_{f} y^4} & -\sqrt{\gamma_f^+ z^2}  \\
\end{array}
$}
 \right] \sqrt{\Delta t}\]

The functions \( f_0 \) and \( f_1 \) appearing in Lemma~\ref{lem:linearized_system} arise from the first-order linearization of the drift function \( f \) and are given as follows:

\begin{flushleft}
$
f_0 (n,y,z)=
\left[
\scalebox{1.0}{$
\begin{array}{c}
-\alpha_m \frac{y^3y^{7}}{N}  \\
\gamma_m^+z^1  \\
\alpha_m \frac{y^3y^{7}}{N}   \\
- \alpha_f \frac{y^{6}y^{7}}{N} - \omega \frac {y^1 y^6}{N_m}\\
\gamma_f^+z^2   \\
\alpha_f \frac{y^{6}y^{7}}{N} +\omega \frac {y^1 y^6}{N_m}\\
-\theta y^8 (y^1 +y^4)\\
\theta y^8 (y^1 +y^4)
\end{array}
$} \right] $
\end{flushleft}
\vspace{1em}
\begin{flushleft}
$f_1 (n,y,z)=
\left[
\scalebox{0.95}{$
\begin{array}{ccccccccccccccccc}
 -\gamma_m^- - \beta _m^-  & 0 &\alpha_m \frac{y^{7}}{N} & 0 & 0 & 0 & \alpha_m \frac{y^{3}}{N}& 0   \\
\gamma_m^-   &  -\rho_m   & 0 & 0 & 0 & 0 & 0 & 0  \\
0 &  \rho_m  & -\alpha_m \frac{y^{7}}{N} & 0 & 0 & 0 & -\alpha_f \frac{y^{6}}{N} & 0  \\
\omega y^6 & 0 & 0 & -\gamma_f^- - \beta _f^-  & 0  & \alpha_f \frac{y^{7}}{N}+  \omega y^1 & \alpha_m \frac{y^{3}}{N}   & 0 \\
0 & 0 & 0  &  \gamma_f^-    & -\rho_f &  0 & 0 & 0  \\
-\omega y^6 & 0 & 0 & 0 & \rho_f  &  - \alpha_f \frac{y^{7}}{N}{N_m}- \omega y^1 & -\alpha_f \frac{y^{6}}{N} & 0 \\
\theta y^8 & 0 & 0 & \theta^8   & 0 & 0 & -\mu  & \theta (y^1+z^1 +y^4+z^2)\\
-\theta y^8  & 0 & 0 & -\theta y^8  & 0 & 0 &0  &  {  B_v(y^6+y^7+y^8)} - \theta (y^1+z^1 +y^4+z^2) - \mu
\end{array}
$}
\right]$
\end{flushleft}

\subsection{The Base Zika Model } \label{app:Base}

Here, we provide an explicit formulation of the model coefficients \( f, h, \sigma, g, \ell \) that define the stochastic dynamics of the system introduced in Subsection \ref{sec: Base}. These coefficients appear in the state-space recursion equations \eqref{state_YZ_Zika}, which govern the evolution of the hidden and observable states of the model are given by $Y = (I^{-}_m, E_m, R_m^{1-}, R_m^{2-} S_m, I_f^-, E_f, R_f^{1-}, R_f^{2-}, S_f, I_v, E_v, S_v)^\top$, while the vector of observable states is $Z = (I^{+}_m, I^{s}_m, I^{+}_f, I^{s}_f)^\top$. This results in a total of $d = 17$ distinct states and $K = 20$ possible transition pathways.

{\footnotesize
\begin{flushleft}
$
f(n,y,z) = y + \Delta t \cdot
\left[
\begin{array}{l}
\beta_m^- y^2 - \gamma_m^- y^1 - \beta_m^+ y^1 \\
\alpha_m y^5 { \frac{y^{11}}{\halfsat+y^{11}}} - (\beta_m^- + \beta_m^s + \beta_m^+) y^2 \\
\gamma_m^- y^1 - \rho_1 y^3 \\
\gamma_m^s z^2 + \gamma_m^+ z^1 - \rho_2 y^4 \\
-\alpha_m y^5 { \frac{y^{11}}{\halfsat+y^{11}}} + \rho_1 y^3 + \rho_2 y^4 \\
\beta_f^- y^7 - \gamma_f^- y^6 - \beta_f^+ y^6 \\
\alpha_f y^{10} { \frac{y^{11}}{\halfsat+y^{11}}} + \omega y^{10 }\kappa_m - (\beta_f^- + \beta_f^s + \beta_f^+) y^7 \\
\gamma_f^- y^6 - \rho_3 y^8 \\
\gamma_f^s z^4 + \gamma_f^+ z^3 - \rho_4 y^9 \\
-\alpha_f y^{10} { \frac{y^{11}}{\halfsat+y^{11}}} - \omega y^{10} \kappa_m + \rho_3 y^8 + \rho_4 y^9 \\
\nu y^{12} - \mu y^{11} \\
-\nu y^{12} + \theta y^{13} ( \kappa_m + \kappa_f ) - \mu y^{12} \\
B_v(y^{11} + y^{12}+ y^{13}) - \theta y^{13}( \kappa_m + \kappa_f ) - \mu y^{13}
\end{array}\right]
$

\end{flushleft}
\begin{flushleft}
\footnotesize

$h(n,y,z) =  h_0(n,z) + h_1(n,z)y $

\end{flushleft}
\begin{flushleft} $
  h_0(n,z) = z+ \left[  \begin{array}{l}
-\gamma_m^+  z^1 \\
-\gamma_m^s z^2 \\
-\gamma_f^+ z^3\\
-\gamma_f^s z^4  
\end{array}\right] \Delta t,   \,\,\,\,\quad  h_1(n,z)= \left [  \begin{array}{ccccccccccccccccc}
\beta_m^+ &\beta_m^+& 0&0&0&0&0&0&0&0&0&0&0\\
0&\beta_m^s&0&0&0&0&0&0&0&0&0&0&0\\
0&0&0&0&0&\beta_f^+&\beta_f^+&0&0&0&0&0&0\\
0&0&0&0&0&0&\beta_f^s&0&0&0&0&0&0
 \end{array} \right ] \Delta t $
\end{flushleft}

\footnotesize
\begin{flushleft}$
g (n,y,z)=
\left[ 
\scalebox{0.85}{$
\begin{array}{ccccccccccc}
 0 & 0 & 0 & 0 & 0 & 0 & 0& 0& - \sqrt{ \beta ^+_m y^1}&0 \\
-\sqrt{\beta_{m}^s y^2}  & -\sqrt{\beta_{m}^+ y^2}  & 0 & 0 & 0 & 0 & 0 & 0& 0 & 0 \\
0 & 0 & 0 & 0 & 0 & 0 & 0 & 0& 0 & 0  \\
0 & 0 & \sqrt{\gamma_m^s z^2} & \sqrt{\gamma_m^+ z^1} & 0 & 0 & 0 & 0 & 0 & 0 \\
0 & 0 & 0 & 0 & 0 & 0 & 0 & 0& 0 & 0  \\
0 & 0 & 0 & 0 & 0 & 0 & 0 & 0 &0&- \sqrt{ \beta ^+_f y^6   }  \\
0 & 0 & 0 & 0 &   -\sqrt{\beta_f^s y^7} & -\sqrt{\beta_f^+ y^7} & 0 & 0& 0 & 0 \\
0 & 0 & 0 & 0 & 0 & 0 & 0 & 0& 0 & 0  \\
0 & 0 & 0 & 0 & 0 & 0 & \sqrt{\gamma_{fs} z^4}
& \sqrt{\gamma_{f+} z^3} & 0 & 0
\\
0 & 0 & 0 & 0 & 0 & 0 & 0 & 0& 0 & 0   \\
0 & 0 & 0 & 0 & 0 & 0 & 0 & 0  & 0 & 0 \\
0 & 0 & 0 & 0 & 0 & 0 & 0 & 0& 0 & 0  \\
0 & 0 & 0 & 0 & 0 & 0 & 0 & 0 & 0 & 0\\
\end{array}
$}
 \right]\sqrt{\Delta t}$
\end{flushleft}

\newgeometry{bottom=2.5cm, top=2cm, left=2.5cm, right=2.5cm, landscape}

\begin{landscape}
\footnotesize


\noindent
$\ell (n,y,z)=
\left[
\scalebox{0.85}{$
\begin{array}{ccccccccccc}
0 & \sqrt{\beta_{m}^+ y^2 }  & 0 & -\sqrt{\gamma _m^+ z^1} & 0 &0  & 0 & 0 &\sqrt{ \beta ^+_m y^1}&0    \\
  \sqrt{\beta_{m}^s y^2 } 
  & 0  &- \sqrt{\gamma _m^s z^2}& 0 & 0 & 0 & 0  &0 & 0  &0    \\
0  &0   &0  &0  &0  & \sqrt{\beta_{f}^+ y^7 }  & 0 &  -\sqrt{\gamma _f^+ z^3} &0& \sqrt{ \beta ^+_f y^6}  \\
  0 & 0  &0  &0  & \sqrt{\beta_{f}^s y^7 } &0  & - \sqrt{\gamma _f^s z^4 } & 0 & 0  &0  \\ 
\end{array}
$}
 \right] \sqrt{\Delta t}$

\vspace{3mm}

\noindent
$\sigma (n,y,z)=
\left[
\scalebox{0.80}{$
\begin{array}{ccccccccccccccccc}
 \sqrt{\beta_{m}^- y^2}  & -\sqrt{\gamma_m^- y^1}  & 0 & 0 & 0 & 0 & 0 & 0 & 0 & 0 & 0 & 0 & 0 & 0 & 0 & 0& 0 \\
  - \sqrt{\beta_{m}^- y^2} &0 & \sqrt{ \alpha_m  y^5 { \frac{y^{11}}{\halfsat+y^{11}r}}}  & 0 & 0 & 0 & 0 & 0 & 0 & 0 & 0 & 0 & 0 & 0 & 0 & 0& 0  \\
0  & \sqrt{\gamma_m^- y^1}  & 0 & -\sqrt{\rho_1 y^3}  & 0 & 0 & 0 & 0 & 0 & 0 & 0 & 0 & 0 & 0 & 0 & 0 & 0 \\
0 & 0 & 0 & 0 & -\sqrt{\rho_2 y^4}  & 0 & 0 & 0 & 0 & 0 & 0 & 0 & 0 & 0 & 0 & 0& 0 \\
 0 & 0 &-\sqrt{\alpha_m y^5 { \frac{y^{11}}{\halfsat+y^{11}}}}  & \sqrt{\rho_1 y^3}  & \sqrt{\rho_2 y^4}  & 0 & 0 & 0 & 0 & 0 & 0 & 0 & 0& 0 & 0 & 0 & 0 \\
 0 & 0 & 0 & 0 & 0 &  \sqrt{\beta_f^- y^7}  & -\sqrt{\gamma_f^- y^6}  & 0 & 0 & 0 & 0 & 0& 0 & 0 & 0 & 0& 0\\
0 & 0 & 0 & 0 & 0 &  -\sqrt{\beta_f^- y^7}& 0 & \sqrt{\alpha_f y^{10 }{ \frac{y^{11}}{\halfsat+y^{11}}}} & \sqrt{\omega y^{10} \kappa_m}   & 0 & 0 & 0 & 0 & 0 & 0 & 0 & 0\\
0 & 0 & 0 & 0 & 0 & 0 &  \sqrt{\gamma_f^- y^6}  & 0 & 0 & -\sqrt{\rho_3 y^8}  & 0 & 0 & 0& 0 & 0 & 0& 0 \\
0 & 0 & 0 & 0 & 0 & 0 & 0 & 0 & 0 & 0 & -\sqrt{\rho_4 y^9}  &0 & 0 & 0 & 0 & 0& 0 \\
0 & 0 & 0 & 0 & 0 &  0 & 0 &-\sqrt{\alpha_f y^{10} { \frac{y^{11}}{\halfsat+y^{11}}}}  & -\sqrt{\omega y^{10} \kappa_m}  & \sqrt{\rho_3 y^8}  & \sqrt{\rho_4 y^9} & 0 & 0 & 0 & 0 & 0 & 0\\
0 & 0 & 0 & 0 & 0 & 0 & 0 & 0 & 0 & 0 & 0 & \sqrt{\nu y^{12}}   & - \sqrt{\mu y^{11} } &0 &0&0 & 0\\
0 & 0 & 0 & 0 & 0 & 0 & 0 & 0 & 0 & 0 & 0 & -\sqrt{\nu y^{12}}  &0 &  \sqrt{\theta y^{13} (\kappa_m + \kappa_f)}& - \sqrt{ \mu y^{12}} & 0 & 0 \\
0 & 0 & 0 & 0 & 0 & 0 & 0 & 0 & 0 & 0 & 0 & 0   &  0 & -\sqrt{\theta y^{13} (\kappa_m + \kappa_f)}& 0 & \sqrt{B_v ( I_v +E_v +y^{13}} &  - \sqrt{ \mu y^{13}} 
\end{array}
$}
\right] \sqrt{\Delta t}$\\
\vspace{2.1mm}
The functions \( f_0 \) and \( f_1 \)  appearing in Lemma~\ref{lem:linearized_system} arise from the first-order linearization of the drift function \( f \) and are given as follows:
 
\vspace{3mm}
$f_1 (n,y,z)=
\left[
\scalebox{0.85}{$
\begin{array}{ccccccccccccccccc}
 -\gamma_m^- -\beta^+_m & \beta _m^-  & 0 & 0 & 0 & 0 & 0 & 0 & 0 & 0 & 0 & 0 & 0  \\
  0 & -(\beta_m^- +\beta_m^s +\beta_m^+ )  & 0 & 0 & \alpha_m \frac{y^{11}}{N} & 0 & 0 & 0 & 0 & 0 & \alpha_m \frac{y^{5}}{N} & 0 & 0   \\
\gamma_m^-   & 0 & -\rho_1   & 0 & 0 & 0 & 0 & 0 & 0 & 0 & 0 & 0 & 0  \\
0 & 0 & 0 &  -\rho_2  & 0 & 0 & 0 & 0 & 0 & 0 & 0 & 0 & 0  \\
 0 & 0 &\rho_1  & \rho_2  & \alpha_m \frac{y^{11}}{N}  & 0 & 0 & 0 & 0 & 0 & \alpha_m \frac{y^{5}}{N} & 0 & 0 \\
 0 & 0 & 0 & 0 & 0 &  -\gamma_f^- -\beta ^+_f & \beta _f^-   & 0 & 0 & 0 & 0 & 0& 0 \\
\omega \frac{y^{10}}{N_m} & \omega \frac{y^{10}}{N_m} & 0 & 0 & 0 & 0 & -(\beta_f^- +\beta_f^s +\beta_f^+ ) &  0 & 0 &  \alpha_f \frac{y^{11}}{N}+ \omega \frac{y^{1} + y^{2} +z^1+z^2}{N_m} &\alpha_f \frac{y^{10}}{N} & 0 & 0 \\
0 & 0 & 0 & 0 & 0 &  \gamma_f^-   & 0 & -\rho_3 &  0 & 0 & 0& 0 & 0  \\
0 & 0 & 0 & 0 & 0 & 0 & 0 & 0 & - \rho_4 & 0  & 0 & 0& 0 \\
-\omega \frac{y^{10}}{N_m} & -\omega \frac{y^{10}}{N_m} & 0 & 0 & 0 &  0 & 0 &\rho_3  & \rho_4 &  -\alpha_f \frac{y^{11}}{N}-\omega \frac{y^{1} + y^{2} +z^1+z^2}{N_m} & -\alpha_f \frac{y^{10}}{N} & 0 & 0\\
0 & 0 & 0 & 0 & 0 & 0 & 0 & 0 & 0 & 0 &-\mu & \nu  & 0\\
\theta \frac{y^{13}}{N_m} &\theta \frac{y^{13}}{N_m} & 0 & 0 & 0 & \theta \frac{y^{13}}{N_f} & \theta \frac{y^{13}}{N_f} & 0 & 0 & 0 & 0 & -\nu -\mu  &  \theta (\frac{y^{1} + y^{2} +z^1+z^2}{N_m} +\frac{y^{6} + y^{7} +z^3+z^4}{N_f})\\
-\theta \frac{y^{13}}{N_m} &-\theta \frac{y^{13}}{N_m} & 0 & 0 & 0 & -\theta \frac{y^{13}}{N_f} &-\theta \frac{y^{13}}{N_f}  & 0 & 0 & 0 & B_v & B_v  & - \theta (\frac{y^{1} + y^{2} +z^1+z^2}{N_m} +\frac{y^{6} + y^{7} +z^3+z^4}{N_f})-  \mu +B_v
\end{array}
$}
\right]$
\end{landscape}
\restoregeometry
\vspace{-0.2cm}
$f_0 (n,y,z)=
\left[
\scalebox{0.8}{$
\begin{array}{c}
0\\
-\alpha_m \frac{y^5y^{11}}{N}  \\
  0\\
\gamma_m^+z^1 + \gamma_m^sz^2  \\
\alpha_m \frac{y^5y^{11}}{N}   \\
 0    \\
- \alpha_f \frac{y^{10}y^{11}}{N} - \omega \frac {(y^1 +y^2)}{N_m}\\
0\\
\gamma_f^+z^3 + \gamma_f^sz^4  \\
\alpha_f \frac{y^{10}y^{11}}{N} +\omega \frac {(y^1 +y^2)}{N_m}\\
0\\
-\theta y^{13}(\frac{y^1+ y^2}{N_m} + \frac{y^6+ y^7}{N_f})\\
 \theta y^{13}(\frac{y^1+ y^2}{N_m} + \frac{y^6+ y^7}{N_f})
\end{array}
$}
 \right]$
\vspace{-3mm}
\subsection{The Extended Zika Model }
\label{app:Extended}
Here, we provide an explicit formulation of the model coefficients \( f, h, \sigma, g, \ell \) that define the stochastic dynamics of the system introduced in Subsection \ref{sec:Extended}. These coefficients appear in the state-space recursion equations \eqref{state_YZ_Zika}, which govern the evolution of the hidden and observable states of the model are given by $Y = (I^{-}_m, E_m, R_m^{1-}, R_m^{2-} S_m, I_f^-, E_f, R_f^{1-}, R_f^{2-}, S_f, I_v, E_v, S_v)^\top$, while the vector of observable states is $Z = (I^{+}_m, I^{s}_m, I^{+}_f, I^{s}_f, R^1_m, R^2_m, R^3_m, R^1_f, R^2_f, R^3_f)^\top$. This results in a total of $d = 23$ distinct states and $K = 20$ possible transition pathways.
\vspace{-0.3cm}
{\footnotesize
\begin{flushleft}
$f(n, y, z) = y +
\begin{bmatrix}
\beta_m^- y^2 - \gamma_m^- y^1 -\beta^+_m y^1 \\
\alpha_m y^5 { \frac{y^{11}}{\halfsat+y^{11}}} - (\beta_m^- + \beta_m^s + \beta_m^+) y^2\\
\gamma_m^- y^1 - \rho_1 y^3 \\
- \rho_2 y^4 \\
- \alpha_m y^5 { \frac{y^{11}}{\halfsat+y^{11}}} + \rho_1 y^3 + \rho_2 y^4 \\
\beta_f^-y^7 - \gamma_f^- y^6 -\beta^+_f y^6 \\
\alpha_f y^{10} { \frac{y^{11}}{\halfsat+y^{11}}} + \omega y^{10} \kappa_m - (\beta_f^- + \beta_f^s + \beta_f^+) y^7 \\
\gamma_f^- y^6 - \rho_3 y^8 \\
- \rho_4 y^9\\
- \alpha_f y^{10} { \frac{y^{11}}{\halfsat+y^{11}}} + \omega y^{10} \kappa_m + \rho_3 y^8 + \rho_4 y^9 \\
\nu y^{12} - \mu y^{11}\\
- \nu y^{12} + \theta y^{13} (\kappa_m + \kappa_f) - \mu y^{12} \\
B_v (y^{11} + y^{12} + y^{13}) - \theta y^{13} (\kappa_m + \kappa_f) - \mu y^{13}
\end{bmatrix} \Delta t
+ 
\begin{bmatrix}
0 \\
0 \\
0 \\
\vartheta_m^3 z^{7} \\
0 \\
0 \\
0 \\
0 \\
\vartheta_f^3 R_z^{10} \\
0 \\
0 \\
0 \\
0
\end{bmatrix}$
\vspace{0.1cm}

$h(n,y,z) =  h_0(n,z) + h_1(n,z)y $
\end{flushleft}
\vspace{-0.5cm}
{\footnotesize
\begin{flushleft}
$
 h_0(n,z) = z+\left[   \begin{array}{l}
-\gamma_m^+  z^1 \\
-\gamma_m^s z^2 \\
-\gamma_f^+ z^3\\
-\gamma_f^s z^4  \\
\gamma_m^+ z^1 +\gamma_m^s z^2 -\vartheta_m^1 z^5\\
\vartheta_m^1z^5 - \vartheta_m^2z^6\\
\vartheta_m^2 z^6 - \vartheta_m^3 z^7\\
\gamma_f^+ z^3 +\gamma_f^s z^4 -\vartheta_f^1 z^8\\
\vartheta_f^1z^8 - \vartheta_f^2z^9\\
\vartheta_f^2 z^9 - \vartheta_f^3 z^10\\
\end{array} \right] \Delta t,   \,\,\,\,\,   h_1(n,z)= \left [ \begin{array}{ccccccccccccccccc}
\beta_m^+&\beta_m^+& 0&0&0&0&0&0&0&0&0&0&0\\
0&\beta_m^s&0&0&0&0&0&0&0&0&0&0&0\\
0&0&0&0&0&\beta_f^+&\beta_f^+&0&0&0&0&0&0\\
0&0&0&0&0&0&\beta_f^s&0&0&0&0&0&0\\
0&0&0&0&0&0&0&0&0&0&0&0&0\\
0&0&0&0&0&0&0&0&0&0&0&0&0\\
0&0&0&0&0&0&0&0&0&0&0&0&0\\
0&0&0&0&0&0&0&0&0&0&0&0&0\\
0&0&0&0&0&0&0&0&0&0&0&0&0\\
0&0&0&0&0&0&0&0&0&0&0&0&0
 \end{array} \right ] \Delta t$

\end{flushleft}

\newgeometry{bottom=1.5cm, top=2cm, left=2.5cm, right=2.5cm, landscape}

\begin{landscape}
\footnotesize


\noindent

$\sigma (n,y,z)=
\left[
\scalebox{0.80}{$ 
\begin{array}{ccccccccccccccccc}
 \sqrt{\beta_{m}^- y^{2}}  & -\sqrt{\gamma_m^- y^{1}}  & 0 & 0 & 0 & 0 & 0 & 0 & 0 & 0 & 0 & 0 & 0 & 0 & 0 & 0& 0 \\
  - \sqrt{\beta_{m}^- y^2} &0 & \sqrt{ \alpha_m y^5 { \frac{y^{11}}{\halfsat+y^{11}}}}  & 0 & 0 & 0 & 0 & 0 & 0 & 0 & 0 & 0 & 0 & 0 & 0 & 0& 0  \\
0  & \sqrt{\gamma_m^- y^1}  & 0 & -\sqrt{\rho_1y^3}  & 0 & 0 & 0 & 0 & 0 & 0 & 0 & 0 & 0 & 0 & 0 & 0 & 0 \\
0 & 0 & 0 & 0 & -\sqrt{\rho_2 y^4}  & 0 & 0 & 0 & 0 & 0 & 0 & 0 & 0 & 0 & 0 & 0& 0 \\
 0 & 0 &-\sqrt{\alpha_m y^5 { \frac{y^{11}}{\halfsat+y^{11}}}}  & \sqrt{\rho_1 y^3}  & \sqrt{\rho_2 y^4}  & 0 & 0 & 0 & 0 & 0 & 0 & 0 & 0& 0 & 0 & 0 & 0 \\
 0 & 0 & 0 & 0 & 0 &  \sqrt{\beta_f^- y^7}  & -\sqrt{\gamma_f^- y^6}  & 0 & 0 & 0 & 0 & 0& 0 & 0 & 0 & 0& 0\\
0 & 0 & 0 & 0 & 0 &  -\sqrt{\beta_f^-y^7}& 0 & \sqrt{\alpha_f y^{10} { \frac{y^{11}}{\halfsat+y^{11}}}} & \sqrt{\omega y^{10} \kappa_m}   & 0 & 0 & 0 & 0 & 0 & 0 & 0 & 0\\
0 & 0 & 0 & 0 & 0 & 0 &  \sqrt{\gamma_f^- y^6}  & 0 & 0 & -\sqrt{\rho_3 y^8}  & 0 & 0 & 0& 0 & 0 & 0& 0 \\
0 & 0 & 0 & 0 & 0 & 0 & 0 & 0 & 0 & 0 & -\sqrt{\rho_4 y^9}  &0 & 0 & 0 & 0 & 0& 0 \\
0 & 0 & 0 & 0 & 0 &  0 & 0 &-\sqrt{\alpha_f y^{10} { \frac{y^{11}}{\halfsat+y^{11}}}}  & -\sqrt{\omega y^{10} \kappa_m}  & \sqrt{\rho_3 y^8}  & \sqrt{\rho_4 y^9} & 0 & 0 & 0 & 0 & 0 & 0\\
0 & 0 & 0 & 0 & 0 & 0 & 0 & 0 & 0 & 0 & 0 & \sqrt{\nu y^{12}}   & - \sqrt{\mu y^{11} } &0 &0&0 & 0\\
0 & 0 & 0 & 0 & 0 & 0 & 0 & 0 & 0 & 0 & 0 & -\sqrt{\nu y^{12}}  &0 &  \sqrt{\theta y^{13} (\kappa_m + \kappa_f)}& - \sqrt{ \mu y^{12}} & 0 & 0 \\
0 & 0 & 0 & 0 & 0 & 0 & 0 & 0 & 0 & 0 & 0 & 0   &  0 & -\sqrt{\theta y^{13}(\kappa_m + \kappa_f)}& 0 & \sqrt{B_v ( y^{11} + y^{12} + y^{13}} &  - \sqrt{ \mu y^{13}} 
\end{array}
$}
\right] \sqrt{\Delta t}$

\vspace{1cm}
\footnotesize
$g (n,y,z)=
\left[
\scalebox{0.85}{$
\begin{array}{cccccccccccccc}
 0 & 0 & 0 & 0 &  0& 0 & 0 & 0 &  -\sqrt{\beta^+_m y^1  } \\
-\sqrt{\beta_{m}^s y^2}  & -\sqrt{\beta_{m}^+y^2}  & 0& 0 & 0& 0 & 0 & 0 & 0 & 0 \\
0 & 0 & 0  & 0  & 0& 0 & 0 & 0 & 0 & 0   \\
0 & 0 & 0  & 0 & 0 & 0 & 0 & 0& 0 & 0 \\
0 & 0 & 0 & 0 & 0 & 0 & 0 & 0 & 0 & 0  \\
0 & 0 & 0 & 0 & 0 & 0 & 0 & 0 &0 & -\sqrt{\beta^+_fy^6  } \\
0 & 0   &  -\sqrt{\beta_f^s y^7} & -\sqrt{\beta_f^+ y^7} & 0& 0 & 0 & 0& 0 & 0   \\
0 & 0 & 0 & 0 & 0 & 0 & 0 & 0 & 0 & 0 \\
0 & 0 & 0 & 0 & 0   &  0 & 0 & 0& 0 & 0 \\
0 & 0 & 0 & 0 & 0 & 0 & 0 & 0  & 0 & 0\\
0 & 0 & 0 & 0 & 0 & 0  & 0 & 0 & 0 & 0  \\
0 & 0 & 0 & 0 & 0 & 0  & 0 & 0 & 0 & 0 \\
0 & 0 & 0 & 0 & 0 & 0& 0 & 0  & 0 & 0 \\
\end{array}
$}
 \right]\sqrt{\Delta t}$

\footnotesize
$\ell (n,y,z)=
\left[
\scalebox{0.85}{$
\begin{array}{cccccccccccccc}
0 & \sqrt{\beta_{m}^+ y^2 }  &  0 & 0 & 0&-\sqrt{\gamma _m^+ z^1}  & 0 & 0 & \sqrt{\beta^+_m y^1  } &0     \\
  \sqrt{\beta_{m}^s y^2 } & 0  &0& 0 & - \sqrt{\gamma _m^s z^2} & 0 & 0  &0 & 0 &0  \\
0  &0   &0  & \sqrt{\beta_{f}^+ y^7 }  & 0   &0&0& -\sqrt{\gamma _f^+ z^3} &0& \sqrt{\beta^+_f y^6  }  \\
  0 & 0    & \sqrt{\beta_{f}^s y^7 } &0  &0&0& - \sqrt{\gamma _f^s z^4 } & 0 & 0 &0   \\
  0 & 0&   0 & 0   & \sqrt{\gamma_m^s z^4} & \sqrt{\gamma_m^+ z^3}  & 0 &0& 0 &0\\
  0 & 0 & 0  & 0 & 0  & 0 & 0&0 & 0 &0 \\
   0 & 0 & 0 & 0 & 0  & 0 & 0  & 0 & 0 &0\\
      0 & 0 & 0 & 0 & 0 & 0 &  \sqrt{\gamma_{fs} z^4}& \sqrt{\gamma_{f+} z^3} & 0 &0\\
       0 & 0 & 0 & 0 & 0 & 0  & 0 & 0 & 0 &0\\
         0 & 0 & 0 & 0 & 0 &0&0  &0 & 0 &0 \\ 
\end{array}
$}\right]
 \sqrt{\Delta t}$

\vspace{0.5cm}
 \normalsize
The functions \( f_0 \) and \( f_1 \)  appearing in Lemma~\ref{lem:linearized_system} arise from the first-order linearization of the drift function \( f \) and are given as follows:

 \vspace{1cm}
\footnotesize
$f_1 (n,y,z)=
\left[
\scalebox{0.95}{$
\begin{array}{ccccccccccccccccc}
 -\gamma_m^--\beta^+_m & \beta _m^-  & 0 & 0 & 0 & 0 & 0 & 0 & 0 & 0 & 0 & 0 & 0  \\
  0 & -(\beta_m^- +\beta_m^s +\beta_m^+ )  & 0 & 0 & \alpha_m \frac{y^{11}}{N} & 0 & 0 & 0 & 0 & 0 & \alpha_m \frac{y^{5}}{N} & 0 & 0   \\
\gamma_m^-   & 0 & -\rho_1   & 0 & 0 & 0 & 0 & 0 & 0 & 0 & 0 & 0 & 0  \\
0 & 0 & 0 &  -\rho_2  & 0 & 0 & 0 & 0 & 0 & 0 & 0 & 0 & 0  \\
 0 & 0 &\rho_1  & \rho_2  & \alpha_m \frac{y^{11}}{N}  & 0 & 0 & 0 & 0 & 0 & \alpha_m \frac{y^{5}}{N} & 0 & 0 \\
 0 & 0 & 0 & 0 & 0 &  -\gamma_f^- -\beta ^+_f & \beta _f^-   & 0 & 0 & 0 & 0 & 0& 0 \\
\omega \frac{y^{10}}{N_m} & \omega \frac{y^{10}}{N_m} & 0 & 0 & 0 & 0 & -(\beta_f^- +\beta_f^s +\beta_f^+ ) &  0 & 0 &  \alpha_f \frac{y^{11}}{N}+ \omega \frac{y^{1} + y^{2} +z^1+z^2}{N_m} &\alpha_f \frac{y^{10}}{N} & 0 & 0 \\
0 & 0 & 0 & 0 & 0 &  \gamma_f^-   & 0 & -\rho_3 &  0 & 0 & 0& 0 & 0  \\
0 & 0 & 0 & 0 & 0 & 0 & 0 & 0 & - \rho_4 & 0  & 0 & 0& 0 \\
-\omega \frac{y^{10}}{N_m} & -\omega \frac{y^{10}}{N_m} & 0 & 0 & 0 &  0 & 0 &\rho_3  & \rho_4 &  -\alpha_f \frac{y^{11}}{N}-\omega \frac{y^{1} + y^{2} +z^1+z^2}{N_m} & -\alpha_f \frac{y^{10}}{N} & 0 & 0\\
0 & 0 & 0 & 0 & 0 & 0 & 0 & 0 & 0 & 0 &-\mu & \nu  & 0\\
\theta \frac{y^{13}}{N_m} &\theta \frac{y^{13}}{N_m} & 0 & 0 & 0 & \theta \frac{y^{13}}{N_f} & \theta \frac{y^{13}}{N_f} & 0 & 0 & 0 & 0 & -\nu -\mu  &  \theta (\frac{y^{1} + y^{2} +z^1+z^2}{N_m} +\frac{y^{6} + y^{7} +z^3+z^4}{N_f})\\
-\theta \frac{y^{13}}{N_m} &-\theta \frac{y^{13}}{N_m} & 0 & 0 & 0 & -\theta \frac{y^{13}}{N_f} &-\theta\frac{y^{13}}{N_f}  & 0 & 0 & 0 & B_v & B_v  & - \theta(\frac{y^{1} + y^{2} +z^1+z^2}{N_m} +\frac{y^{6} + y^{7} +z^3+z^4}{N_f})-  \mu + B_v
\end{array}
$}
\right]$
\vspace*{\fill}
\end{landscape}

\restoregeometry

$f_0 (n,y,z)=
\left[
\scalebox{1.0}{$
\begin{array}{c}
0\\
-\alpha_m \frac{y^5y^{11}}{N}  \\
  0\\
\vartheta_m^3 z^7 \\
\alpha_m \frac{y^5y^{11}}{N}   \\
 0    \\
- \alpha_f \frac{y^{10}y^{11}}{N} - \omega \frac {(y^1 +y^2)}{N_m}\\
0\\
\vartheta_f^3 z^{10}  \\
\alpha_f \frac{y^{10}y^{11}}{N} +\omega \frac {(y^1 +y^2)}{N_m}\\
0\\
-\theta y^{13}(\frac{y^1+ y^2}{N_m} + \frac{y^6+ y^7}{N_f})\\
 \theta y^{13}(\frac{y^1+ y^2}{N_m} + \frac{y^6+ y^7}{N_f})
\end{array}
$}
 \right]$

\section{Proof of Proposition \ref{prop:init_law} }\label{Proof_details_IFE}
	
	Here, we provide the detailed derivations of the initial filter estimate. 
	The computations rely on standard properties of expectation and covariance applied to 
	the Gaussian random factors introduced in Assumption~\ref{ass:dfc_zika}. 
	In particular, we show how the conditional mean $M_0=\mathbb{E}[Y_0\mid \mathcal{F}_0^Z]$ and the 
	conditional covariance $Q_0=\operatorname{Cov}(Y_0\mid \mathcal{F}_0^Z)$ are obtained. 
	The arguments consist mainly of using linearity of the expectation, 
	linearity of the variance, and independence of the driving random coefficients. 
	
	We work conditionally on $\mathcal{F}_0^Z$, and use the placeholders 
	$\dagger\in\{m,f\}$ (sex) and $\#\in\{I,E,R\}$ (compartment). 
	Set
	\[
	A_{\dagger}:=I_{\dagger,0}^{+}+I_{\dagger,0}^{s}, 
	\qquad 
	B_{\dagger}:=R_{\dagger,0}^{1}+R_{\dagger,0}^{2}+R_{\dagger,0}^{3}.
	\]
	By Assumption~\ref{ass:dfc_zika}, the initial dark figure coefficients (DFCs) are independent Gaussians,
	\[
	\dfc_{0,\dagger}^{\#}\sim\mathcal N\!\big(\dfcdmean^{\#},(\dfcdvar^{\#})^2\big),
	\qquad \#\in\{I,E,R\},\ \dagger\in\{m,f\},
	\]
	and are independent of $Z_0$. The vector components $(I_{v,0},E_{v,0},S_{v,0})$ are independent Gaussians with means $(\eccmean^I,\eccmean^E,\eccmean^S)$ and variances $((\eccvar^I)^2,(\eccvar^E)^2,(\eccvar^S)^2)$, and are independent of the human blocks.
	
	\smallskip
	\noindent\textbf{Conditional mean $M_0=\E[Y_0\mid\mathcal{F}_0^Z]$.}
	By linearity of expectation and the DFC definitions,
	\[
	I_{\dagger,0}^{-}\;=\;\dfc_{0,\dagger}^{I}\,A_{\dagger}, 
	\qquad 
	E_{\dagger,0}\;=\;\dfc_{0,\dagger}^{E}\,A_{\dagger}, 
	\qquad 
	R_{\dagger,0}^{1-}\;=\;\dfc_{0,\dagger}^{R}\,B_{\dagger}.
	\]
	Using $R_{\dagger,0}^{2-}=0$ and the population balance
	\(
	S_{\dagger,0}=N_{\dagger}-\big(I_{\dagger,0}^{-}+E_{\dagger,0}+R_{\dagger,0}^{1-}+R_{\dagger,0}^{2-}+I_{\dagger,0}^{+}+I_{\dagger,0}^{s}\big),
	\)
	we obtain, for each $\dagger\in\{m,f\}$, $M_0^{\dagger}$.
	For vectors,
	\(
	M_0^{v}=(\eccmean^{I},\,\eccmean^{E},\,\eccmean^{S})^\top.
	\)
	Thus the full mean is given by $	M_0=\big(	M_0^{m}, M_0^{f}, M_0^{v}\big).$
	
	\smallskip
	\noindent\textbf{Conditional covariance $Q_0=\cov(Y_0\mid\mathcal{F}_0^Z)$.}
	Each random component is an affine function of mutually independent Gaussian factors 
	$\{\dfc_{0,\dagger}^{\#}\}$ (human) and $(I_{v,0},E_{v,0},S_{v,0})$ (vector). Therefore $Y_0\mid\mathcal{F}_0^Z$ is Gaussian and $Q_0$ is block diagonal with human male/female blocks (size $5\times 5$) and a vector block (size $3\times 3$).
	
	\medskip
	\noindent\emph{(a) Human block for a fixed $\dagger\in\{m,f\}$.}
	With the ordering $(I_{\dagger,0}^{-},\,E_{\dagger,0},\,R_{\dagger,0}^{1-},\,R_{\dagger,0}^{2-},\,S_{\dagger,0})$, set
	\[
	a_I:=A_{\dagger}^2\,(\dfcdvar^{I})^2,\qquad
	a_E:=A_{\dagger}^2\,(\dfcdvar^{E})^2,\qquad
	b_R:=B_{\dagger}^2\,(\dfcdvar^{R})^2.
	\]
	Independence across $\#\in\{I,E,R\}$ yields
	\[
	\Var(I_{\dagger,0}^{-})=a_I,\quad
	\Var(E_{\dagger,0})=a_E,\quad
	\Var(R_{\dagger,0}^{1-})=b_R,\quad
	\Var(R_{\dagger,0}^{2-})=0,
	\]
	and all pairwise covariances among $I_{\dagger,0}^{-}$, $E_{\dagger,0}$, $R_{\dagger,0}^{1-}$ vanish. Since
	\(
	S_{\dagger,0}
	=
	C_{\dagger}-(I_{\dagger,0}^{-}+E_{\dagger,0}+R_{\dagger,0}^{1-})
	\)
	with $C_{\dagger}$ constant given $(\mathcal{F}_0^Z)$, we have
	\[
	\Var(S_{\dagger,0})=a_I+a_E+b_R,\qquad
	\cov(S_{\dagger,0},I_{\dagger,0}^{-})=-a_I,\quad
	\cov(S_{\dagger,0},E_{\dagger,0})=-a_E,\quad
	\cov(S_{\dagger,0},R_{\dagger,0}^{1-})=-b_R.
	\]
	Collecting the entries gives the explicit $5\times 5$ block, $Q_0^{\dagger}$.
	By independence between the male and female DFCs, the cross-covariances between $Q_0^{m}$ and $Q_0^{f}$ are zero.
	
	\smallskip
	\noindent\emph{(b) Vector block.}
	With ordering $(I_{v,0},E_{v,0},S_{v,0})$ and independence, we obtain $	Q_0^{v}$.

	\smallskip
	\noindent\emph{(c) Full covariance.}
	Therefore, based on (a) and (b) we obtain $ Q_0$.

 \paragraph{Acknowledgment and Funding}
 L.~Oluoch and R.~Wunderlich gratefully acknowledge financial support from the cooperation program between the Deutsche Forschungsgemeinschaft (DFG) and The World Academy of Sciences (TWAS) under grant number 710382.
F.~Ouabo Kamkumo and R.~Wunderlich gratefully acknowledge the  support by the Deutsche Forschungsgemeinschaft (DFG), award number 458468407.

\addcontentsline{toc}{section}{References}
\bibliographystyle{acm}

\end{document}